\documentclass[11pt]{article}
\usepackage{amsmath}
\usepackage{fullpage}
\usepackage{amsthm}
\usepackage{hyperref}
\usepackage{amssymb}

\usepackage{mathtools}
\usepackage{bbm}
\usepackage{bm}
\usepackage{graphicx}

\usepackage{xcolor}

\usepackage{algorithm}
\usepackage{algpseudocode}

\usepackage{comment}

\usepackage{authblk}

\usepackage[numbers]{natbib}

\newtheorem{assumption}{Assumption} 
\newtheorem{theorem}{Theorem}
\newtheorem{lemma}[theorem]{Lemma} 
\newtheorem{proposition}[theorem]{Proposition}
\newtheorem{remark}[theorem]{Remark}

\newcommand{\E}{\mathbb{E}}
\newcommand{\Prob}{\mathbb{P}}
\newcommand{\Var}{\mathrm{Var}}

\author[1]{Cyrill Scheidegger}
\author[2]{Zijian Guo}
\author[1]{Peter B\"uhlmann}
\affil[1]{Seminar for Statistics, ETH Zurich}
\affil[2]{Center for Data Science, Zhejiang University}

\begin{document}

\title{Inference for Heterogeneous Treatment Effects with Efficient Instruments and Machine Learning}
\maketitle

\begin{abstract}
We introduce a new instrumental variable (IV) estimator for heterogeneous treatment effects in the presence of endogeneity. Our estimator is based on double/debiased machine learning (DML) and uses efficient machine learning instruments (MLIV) and kernel smoothing. We prove consistency and asymptotic normality of our estimator and also construct confidence sets that are more robust towards weak IV. Along the way, we also provide an accessible discussion of the corresponding estimator for the homogeneous treatment effect with efficient machine learning instruments. The methods are evaluated on synthetic and real datasets and an implementation is made available in the \texttt{R} package \texttt{IVDML}.
\end{abstract}

\section{Introduction}
A large part of research in causal inference and econometrics deals with estimation and inference for treatment effects based on observational data. As an example, suppose that we want to understand the effect of years of education on wages, as for example investigated by \cite{CardCollege}. In this case, a randomized controlled trial is infeasible, since one cannot randomly assign years of education to students and one needs to rely on observational data alone. With purely observational data, a major issue is the potential endogeneity of treatment and response. For example, if there are unmeasured confounding factors that influence both the years of education as well as the wage (for example what can be summarized by ``ability''), there will be endogeneity. A second issue that often appears is heterogeneity of the treatment effect. For example, the effect of education on wage could be different between men and women, in different geographical locations or it could depend on the age or the experience of the worker.

\subsection{Our Contribution}
We provide a contribution that deals with both of these issues and propose a novel procedure for estimation and inference of a heterogeneous treatment effect in the presence of endogeneity using an instrumental variable. Concretely, we consider the model
\begin{equation}\label{eq_Model}
    Y_i = \beta(V_i)D_i + g(X_i) + \epsilon_i, \, i =1,\ldots, N
\end{equation}
for a response variable $Y_i$, a univariate continuous covariate $V_i$, the treatment variable $D_i$, covariates $X_i$ and an error term $\epsilon_i$. We focus here on univariate $V_i$, but the extension to the multivariate case is conceptually straightforward, see also the discussion in Section \ref{sec_Discussion}.
The goal is to perform inference for $\beta(v)$ for some specific value $v$. The error $\epsilon_i$ satisfies $\E[\epsilon_i|X_i]= 0$ but we have endogeneity, meaning that $\E[\epsilon_i|D_i, X_i]\neq 0$. However, we assume that we have access to an instrumental variable $Z_i$ that satisfies $\E[\epsilon_i|Z_i,X_i]= 0$.
Our method is based on double/debiased machine learning \citep{ChernozhukovDML} and uses a kernel smoothing approach to model the heterogeneity in the continuous variable $V_i$. We give conditions under which our method yields consistency and asymptotically valid confidence sets.
{Our main focus is neither to provide generic meta-algorithms for heterogeneous treatment effect estimation in the spirit of \cite{NieQuasiOracle, FosterOrthogonalStatisticalLearning, SyrgkanisMLEstimationOfHTE} nor a detailed investigation of optimality properties \cite{KennedyTowardsOptimalDoublyRobustEstimation, KennedyMinimaxRates}. We rather want to introduce a simple, interpretable and easy-to-use method that works with user-chosen machine learning methods for nuisance function estimation. Nevertheless, }
many extensions of our method are possible by replacing kernel smoothing with other nonparametric regression or machine learning methods. Along the way, we also provide a discussion of some aspects of double machine learning with instrumental variables for homogeneous treatment effects (i.e., $\beta(\cdot) \equiv \beta$ is a constant) which, to our knowledge, have not yet been comprehensively discussed.
Concretely, we discuss the usage of machine learning instruments and weak-instrument-robust confidence intervals. Finally, we provide an implementation of our method in the R package \texttt{IVDML}, {which is available on CRAN \citep{IVDMLCRAN}.} {}

\subsection{Partially Linear IV Models}\label{sec_PLIVM}
If the treatment effect $\beta(\cdot)$ in \eqref{eq_Model} is constant (i.e., does not depend on $V_i$), the model reduces to a standard partially linear model with endogeneity (i.e. $\E[\epsilon_i|D_i, X_i]\neq 0$)
\begin{equation}\label{eq_ModelHom}
    Y_i = \beta D_i + g(X_i) + \epsilon_i, \, i = 1,\ldots, N,
\end{equation}
where we assume the availability of an instrument $Z_i$. In our context, an instrument $Z_i$ must satisfy $\E[D_i|Z_i, X_i]\neq\E[D_i|X_i]$ and $\E[\epsilon_i|Z_i, X_i] = 0$. The first condition is the relevance condition (the instrument is associated with the treatment conditionally on the covariates), and the second condition is the validity condition (the instrument is not associated with unmeasured confounders and affects the outcome only through the treatment conditionally on the covariates). Classical two-stage least squares instrumental variables theory, as for example discussed in many textbooks on econometrics, e.g. \cite{WooldridgeIntroductoryEconometrics, StockIntroductionToEconometrics}, assumes that $g(X_i)$ is a linear function in $X_i$ and utilizes only linear versions of relevance and validity. The validity of an instrument can in general not be tested but must be argued using domain knowledge. As noted for example by \cite{DieterleASimpleDiagnostic, ChenMostlyHarmlessML}, when researchers argue for the validity of instruments they often implicitly argue for mean independence and not for uncorrelatedness, so our stronger mean independence assumption is not too restrictive. On the other hand, the relevance condition $\E[D_i|Z_i, X_i]\neq\E[D_i|X_i]$ is weaker than the classical assumption of correlation between $D_i$ and $Z_i$.

There is a large body of work that considers the partially linear model \eqref{eq_ModelHom} in the endogenous setting, see for example \cite{OkuiDoublyRobustIVRegression, EmmeneggerRegularizingDML} for a review. 
Double/debiased machine learning (DML) \citep{ChernozhukovDML} provides a general framework for performing inference on low-dimensional parameters that are identified using moment restrictions in the presence of potentially high-dimensional nuisance functions. If one uses a score function that is Neyman-orthogonal and sample-splitting with cross-fitting, one can use estimators for the nonparametric part of the partially linear model that converge comparably slow, which opens the possibility to use many modern machine learning methods and potentially high-dimensional covariates. As one concrete application, \cite{ChernozhukovDML} apply their methodology to the partially linear IV model \eqref{eq_ModelHom}.
DML in the context of model \eqref{eq_ModelHom} has been studied more comprehensively by \cite{EmmeneggerRegularizingDML}, allowing for multivariate $D_i$ and $Z_i$, potential overidentification and using a regularization scheme to obtain better finite sample performance and smaller confidence intervals.

An omnipresent issue with IV estimators is that they have large variance if the instrument is not very strong. Under the mean independence assumption $\E[\epsilon_i|X_i, Z_i] = 0$, every transformation $\zeta(Z_i, X_i)$ of the instrument, and the covariates is also a valid instrument and one can search for the transformation that leads to the smallest variance of the estimator. This fits into the larger context of semiparametric efficiency, where one seeks the estimator with the smallest asymptotic variance \citep{VanDerVaartAsymptotic, TsiatisSemiparametric}.

The semiparametrically efficient estimator in the context of model \eqref{eq_ModelHom} is for example presented in \cite{VansteelandtImprovingRobustness}. While semiparametrically efficient estimators are appealing from a theoretical perspective, they are often not recommended in practice, since they are too complicated to be estimated well with finite samples, see for example \cite{YoungROSE} for a related discussion. 

For this reason, one often resorts to estimators that only achieve the semiparametric efficiency under the homoscedasticity assumption $\E[\epsilon_i^2| Z_i, X_i] = \mathrm{const}$. In this case, the transformation leading to the smallest asymptotic variance of the estimator is $\zeta(Z_i, X_i) = \E[D_i|Z_i, X_i]$, which means that one can learn the optimal instruments using machine learning \citep{BelloniSparseModels, ChenMostlyHarmlessML, LiuDeepIV, GuoTSCI}. 
Somewhat surprisingly, the solutions for model \eqref{eq_ModelHom} using DML, do not use these ideas but use the instruments linearly \citep{ChernozhukovDML, EmmeneggerRegularizingDML}. To our knowledge, the only place in the literature where double machine learning is combined with these efficient instruments is in \cite{AhrensDMLStata}, which presents a STATA implementation of DML that has one particular function that uses these efficient instruments (see Section 3.2 there), but the asymptotic results are not discussed in detail.

We will use the same idea of learning the efficient instruments from data and using DML for our treatment of the heterogeneous treatment effect in model \eqref{eq_Model}. Along the way, we will also provide some additional details for the homogeneous treatment effect in model \eqref{eq_ModelHom} that have not yet been comprehensively discussed in the literature.

{
\subsection{Varying Coefficient Models}
Another type of models related to \eqref{eq_Model} are \textit{varying coefficient models}, see, e.g., \cite{ClevelandLocalRegressionMOdels, HastieVaryingCoefficientModels, FanStatisticalEstimation, FanStatisticalMethods}.
These are linear models, where the coefficient vector can vary depending on an additional variable, i.e.,
\begin{equation}\label{eq_VarCoef}
   Y_i = \gamma(U_i)^TW_i + \epsilon_i. 
\end{equation}
Moreover, varying coefficient models have also been studied in the endogenous setting with instrumental variables \cite{CaiFunctionalCoefficientIVModels, CaiPartiallyVaryingCoefficientIVModels, SuLocalLinearGMMEstimation}.
Structurally, our model \eqref{eq_VarCoef} is related to model \eqref{eq_Model} by setting $U_i = (V_i, X_i)$, $W_i = (D_i, 1)$ and $\gamma(\cdot) = (\beta(\cdot), g(\cdot))$. However, there are important differences: we are only interested in conducting inference for the single component function $\beta(\cdot)$, whereas $g(\cdot)$ is a (potentially high-dimensional) nuisance function that allows to include additional control variables in a flexible fashion with arbitrary machine learning. We are not interested in estimating or conducting inference for the function $g(\cdot)$. In contrast, with varying coefficient models, typically, all the component functions are of equal interest and are modeled using low-dimensional techniques. Hence, we will look at model \eqref{eq_Model} not through the lens of varying coefficient models but heterogeneous treatment effects instead.}

\subsection{Heterogeneous Treatment Effects}
{When the heterogeneous treatment effect $\beta(\cdot)$ is of interest,}
there is a large body of works, that consider (variants of) model \eqref{eq_Model} but assume no endogeneity, that is $\E[\epsilon_i|D_i, X_i] = 0$ and there is no need for an instrument. An overview can be found in Chapters 14 and 15 of \cite{ChernozhukovCausalML}. Most works consider a binary treatment $D_i$ taking values in $\{0,1\}$. In that case, $\beta(V_i)$ in \eqref{eq_Model}, is the Conditional Average Treatment Effect (CATE) $\beta(V_i) = \E[Y_i(1)-Y_i(0)|V_i]$, where $Y_i(1)$ and $Y_i(0)$ are the potential outcomes of $Y_i$ that would be observed when the treatment is assigned to be $D_i=0$ or $D_i=1$, respectively \citep{ChernozhukovCausalML, RubinPotentialOutcomes}. We restrict ourselves here to the literature which is most related to our contribution.

A wide-spread approach for estimation and inference for the CATE is to express it as a conditional expectation $\beta(V_i) = \E[W_i(\eta_0)|V_i]$, where $W_i(\eta_0)$ is a random variable depending on the observed quantities and the unknown potentially high-dimensional nuisance parameter $\eta_0$ which has to be estimated. If $W_i(\eta_0)$ satisfies the Neyman orthogonality, one can use a DML approach with cross-fitting to obtain an estimate $W_i(\hat \eta)$, which is then regressed on $V_i$ \citep{ChernozhukovCausalML}. If $V_i$ is low-dimensional (as it is in our work), one can use standard methods and results for low-dimensional nonlinear regression to perform inference and estimation for the CATE. Such approaches are followed in \cite{SemenovaDMLCATE} using series expansion, \cite{FanEstimationOfCATE} using local linear regression, and \cite{ZimmertNonparametricEstimation} using locally constant regression. In \cite{WagerEstimationInferenceHTE}, a modified version of the random forest algorithm is proposed, allowing to perform inference on heterogeneous treatment effects based on a similar idea. Another approach is to identify the heterogeneous treatment effect $\beta(V_i)$ by a conditional moment restriction. This approach is particularly popular with random forest based methods, see \cite{AtheyGRF, OprescuORF}.
The above works mostly consider binary treatment $D_i$. Moreover, they do not cover the endogenous case, with the exception of \cite[Sec.7]{AtheyGRF}, who apply their method to instrumental variables regression.
Perhaps most related to our method is \cite{SyrgkanisMLEstimationOfHTE} which considers the endogenous setting and uses instrumental variables. The method also starts from a conditional moment equation similar to ours and then constructs a least squares objective which can be minimized by a user-chosen machine learning algorithm.
Whereas the problem setup is the same, our method complements the work by \cite{SyrgkanisMLEstimationOfHTE} in several ways. To summarize, our approach is tailored to the special case of a treatment effect that depends smoothly on a univariate quantity, whereas their approach is more generic. With our method, we are able to perform inference for the treatment effect $\beta(v)$ at a concrete value $v$ and we extensively discuss the efficiency and both standard and weak-instrument robust inference. In contrast, \cite{SyrgkanisMLEstimationOfHTE} only briefly touch on inference for projections of the heterogeneous treatment effect on a simple space (e.g., the heterogeneity can be approximated by a linear function of a low number of covariates) and they do not explicitly discuss efficiency.

\subsection{Outline}
The rest of the paper is structured as follows. In Section \ref{sec_HomTE}, we focus on the homogeneous treatment effect with machine learning instruments. This section mainly serves as a preparation and motivation for the following discussion of the heterogeneous treatment effect, but we also emphasize some aspects that have not been discussed in detail until now, most notably the construction of weak-instrument robust confidence sets. In Section \ref{sec_HetTE}, we consider the heterogeneous treatment effect model and introduce our estimator. Moreover, we give conditions under which we can perform inference for the heterogeneous treatment effect. In Section \ref{sec_Applications}, we perform a simulation study and apply our method to two well-known instrumental variable datasets. Proofs and some additional simulation results can be found in the appendix.

\subsection{Notation}
For two sequences $(a_n)_{n\in \mathbb N}$ and $(b_n)_{n\in \mathbb N}$ of positive real numbers, we write $a_n\lesssim b_n$ if there exists a constant $C> 0$ such that for all $n\in \mathbb N$, $a_n\leq C b_n$. We write $a_n \asymp b_n$ if $a_n\lesssim b_n$ and $b_n \lesssim a_n$. We write $a_n\ll b_n$ if $a_n/b_n\to 0$ as $n\to\infty$. Moreover, we use stochastic order symbols: for a sequence of random variables $(X_n)_{n\in \mathbb N}$ and a sequence of positive real numbers $(a_n)_{n\in \mathbb N}$, we write $X_n = o_P(a_n)$ if for all $\epsilon >0$, $\Prob(|X_n|/a_n \geq \epsilon)\to 0$, as $n\to \infty$. We write $X_n = O_P(a_n)$ if $\limsup_{n\to\infty}\Prob(|X_n|/a_n\geq M)\to 0$, as $M\to \infty$. Equalities involving conditional expectations should be understood to hold almost surely.

\section{Homogeneous Treatment Effect}\label{sec_HomTE}
In this section, we review and add to the discussion of DML in the partially linear IV model. Concretely, we observe $N$ i.i.d. realizations of $(Y_i, D_i, X_i, Z_i)$ from the model
\begin{equation}\label{eq_ModelHom1}
    Y_i = \beta D_i + g(X_i) + \epsilon_i
\end{equation}
for random variables $Y_i\in \mathbb R$, $D_i\in \mathbb R$, $X_i\in \mathbb R^p$ and the (unobserved) random error $\epsilon_i \in \mathbb R$ that is correlated with $D_i$ making the treatment $D_i$ endogenous (i.e. $\E[\epsilon_i|D_i, X_i]\neq 0$). To deal with the endogeneity, we assume the availability of an instrumental variable $Z_i\in \mathbb R^d$ satisfying Assumption \ref{ass_Hom} below. We allow the data generating process to vary with $N$, even if it is not reflected in the notation. This is needed to consider settings where the strength of the instrument $Z_i$ decreases with $N$. Furthermore, it allows for $X_i$ to be high-dimensional, i.e. $p\to\infty$ as $n\to\infty$. We assume that the instrument $Z_i$ satisfies the following conditions.
\begin{assumption}\label{ass_Hom}
\mbox\newline
    \begin{enumerate}
        \item $\E[D_i|Z_i, X_i] \neq \E[D_i|X_i]$,
        \item $\E[\epsilon_i| Z_i, X_i] = 0$.
    \end{enumerate}
\end{assumption}
The first item states that the instrument is associated with the treatment conditionally on the covariates. For the second item, note that in model \eqref{eq_ModelHom1}, we can assume without loss of generality that $\E[\epsilon_i|X_i] = 0$, since we can otherwise change $g(X_i) \leftarrow g(X_i) + \E[\epsilon_i|X_i]$. Then, assuming $\E[\epsilon_i|Z_i, X_i]$ means that $Z_i$ is a valid instrument (see also the discussion in Section \ref{sec_PLIVM}).

If the dimension of $Z_i$ is $d=1$, \cite{ChernozhukovDML} and \cite{EmmeneggerRegularizingDML} identify $\beta$ as
\begin{equation}\label{eq_IdentificationLinear}
\beta = \frac{\E[(Y_i - \E[Y_i|X_i])(Z_i - \E[Z_i|X_i])]}{\E[(D_i - \E[D_i|X_i])(Z_i - \E[Z_i|X_i])]},
\end{equation}
provided that $\E[(D_i - \E[D_i|X_i])(Z_i- \E[Z_i|X_i]]\neq 0$. To construct an estimator, one replaces the inner conditional expectations by machine learning estimates and the outer expectations by sample means using a cross-fitting scheme (see details below). With this, one obtains an asymptotically normal estimator with variance that can be consistently estimated (see below). However, since $\E[\epsilon_i|Z_i, X_i] = 0$, one can in principle replace $Z_i$ by $\zeta(Z_i, X_i)$ for an arbitrary function $\zeta$ of $Z_i$ and $X_i$, i.e.
\begin{equation}\label{eq_IdentificationZeta}
    \beta = \frac{\E[(Y_i - \E[Y_i|X_i])(\zeta(Z_i, X_i) - \E[\zeta(Z_i, X_i)|X_i])]}{\E[(D_i - \E[D_i|X_i])(\zeta(Z_i, X_i) - \E[\zeta(Z_i, X_i)|X_i])]}.
\end{equation}
Note that here, we do not have the restriction $d=1$ of univariate $Z_i$ anymore but the function $\zeta$ should be real-valued.
A natural goal is to find a transformation $\zeta$ that is easy to obtain and such that the resulting estimator has a small variance. We use $\zeta(Z_i, X_i) = f(Z_i, X_i)\coloneqq \E[D_i|Z_i, X_i]$, which can be easily obtained by nonlinear regression of $D_i$ versus $(Z_i, X_i)$ and coincides with the semiparametrically efficient instrument (see e.g. equation (9) in \cite{VansteelandtImprovingRobustness}) for the homoscedastic case $\E[\epsilon_i^2|Z_i, X_i] = \mathrm{const.}$ This means that under homoscedasticity, the estimator based on the identification \eqref{eq_IdentificationZeta} achieves minimum asymptotic variance with $\zeta = f$ (see also Section \ref{sec_EfficiencyEstimator} below).

In the literature, this approach of learning so-called optimal instruments is popular \citep{BelloniSparseModels, ChenMostlyHarmlessML, LiuDeepIV, GuoTSCI}, see also section 9.2 in \cite{WagerCausalInference} for an accessible discussion. However, to our knowledge, the only place in the literature where it is applied to double machine learning is \cite{AhrensDMLStata}, where the software package \texttt{ddml} for DML in Stata is presented which includes a function for the \textit{flexible partially linear IV model} which coincides with the approach outlined above. In the following, we present the corresponding theory in our notation with the aim of motivating and facilitating the discussion of the heterogeneous treatment effect model in Section \ref{sec_HetTE}.

\subsection{Construction}
We now present the construction of the estimator using the cross-fitting scheme. We assume that we observe $N$ i.i.d. observations $(Y_i, D_i, X_i, Z_i)_{i=1}^N$ of model \eqref{eq_ModelHom1}. Let us define the quantities
\begin{align*}
    f(Z_i, X_i) &\coloneqq \E[D_i|Z_i, X_i],\\
    \phi(X_i) & \coloneqq \E[D_i|X_i] = \E[f(Z_i, X_i)| X_i],\\
    l(X_i) & \coloneqq \E[Y_i|X_i].
\end{align*}
Fix $K\in \mathbb N$ (independent of $N$). Let $\left\{I_k\right\}_{k=1}^K$ be a $K$-fold random partition of the observation indices $[N]=\{1,\ldots, N\}$ such that each fold has size $|I_k|=n=N/K$ (for simplicity, we assume that $N/K$ is an integer). {For $k = 1,\ldots, K$, let $\hat f^k$ and $\hat l^k$ be estimates of $f$ and $l$ and let $\hat\phi_1^k$ and $\hat\phi_2^k$ be estimates of $\phi$, where $\hat f^k$, $\hat l^k$, $\hat \phi_1^k$ and $\hat\phi_2^k$ are obtained by machine learning only using the data $I_k^c$} (i.e. on all data apart from the $k$th fold). {The reason for having two different estimators $\hat \phi_1^k$ and $\hat \phi_2^k$ for $\phi$ will be made clear below in Remark \ref{rmk_TwoStage}.}

Define the residuals
\begin{align}
    R_{Y, i}^k &\coloneqq Y_i - \hat l^k(X_i), \, i\in I_k,\, k = 1,\ldots, K \label{eq_DefRY}\\
    R_{D, i}^k &\coloneqq D_i - \hat \phi_1^k(X_i), \, i\in I_k,\, k = 1,\ldots, K \label{eq_DefRD} \\
    R_{f, i}^k &\coloneqq \hat f^k(Z_i, X_i) - \hat \phi_2^k(X_i), \, i\in I_k,\, k = 1,\ldots, K. \label{eq_DefRf}
\end{align}
and define the estimator
\begin{equation}\label{eq_DefEstimator}
    \hat \beta = \frac{\frac{1}{K}\sum_{k=1}^K\frac{1}{n}\sum_{i\in I_k}R_{Y, i}^k R_{f, i}^k}{\frac{1}{K}\sum_{k=1}^K\frac{1}{n}\sum_{i\in I_k}R_{D, i}^k R_{f, i}^k}.
\end{equation}
Essentially, this is a two-stage least squares estimator of $\left(R_{Y, i}^k\right)_{i,k}$ on $\left(R_{D, i}^k\right)_{i,k}$ with instrument $\left(R_{f, i}^k\right)_{i,k}$ combined with the cross-fitting aggregation. This estimator is the same as equation (14) in \cite{AhrensDMLStata}.

\begin{remark}
This estimator is the cross-fitting sample version of \eqref{eq_IdentificationZeta} for $\zeta(Z_i, X_i) = f(Z_i, X_i) = \E[D_i|Z_i, X_i]$. The estimator is the analogue of the ``DML2''-type estimator from \cite{ChernozhukovDML}. The construction and analysis of the ``DML1''-type estimator would be completely analogous, however, we omit this here, as the ``DML2''-type estimator is typically preferred as it is more stable.
\end{remark}

\begin{remark}\label{rmk_TwoStage}
{The estimators $\hat f^k$, $\hat l^k$ and $\hat \phi_1^k$ can be obtained by directly fitting a machine learning model of $D_i$ vs. $(Z_i, X_i)$, $Y_i$ vs. $X_i$ and $D_i$ vs. $X_i$, respectively, using only $i\in I_k^c$. In principle, one could use the same direct estimator $\hat \phi_2^k = \hat\phi_1^k$ in the definition of the residuals $R_{f, i}^k$. However, similarly to \cite{AhrensDMLStata}, we recommend to use a two-stage approach for $\hat\phi_2^k$: after obtaining $\hat f^k$ on $I_k^c$, regress $(\hat f^k(Z_i, X_i))_{i\in I_k^c}$ on $(X_i)_{i\in I_k^c}$ to obtain $\hat\phi^k_2$, which is a different estimate of $\phi$ than $\hat\phi_1^k$.
The intuition behind using $\hat\phi_k^2$ instead of $\hat\phi_k^1$ in the definition \eqref{eq_DefRf} of $R_{f,i}^k$} is that due to misspecification it might happen that $\hat f^k$ does not converge to the true $f$ but to some function $\zeta$. Then, we are in the setting of the identification \eqref{eq_IdentificationZeta} and should subtract an estimate of $\E[\zeta(Z_i, X_i)|X_i]$ rather than an estimate of $\phi(X_i) = \E[D_i|X_i]$. {Hence, we expect the estimator $\hat\beta$ with $R_{f, i}^k$ based on $\hat\phi_2^k$ to be more robust against poor estimation of $f$.}
\end{remark}

\subsection{Variance Estimator}
In Section \ref{sec_SummaryHTE} below, we will see conditions under which the estimator $\hat\beta$ is asymptotically normal, $\sqrt{N}(\hat\beta-\beta)/\sigma\to\mathcal N(0,1)$ in distribution and that the asymptotic variance $\sigma^2$ can be consistently estimated by 
\begin{equation}\label{eq_VarEst}
\hat\sigma^2 = \frac{\frac{1}{K}\sum_{k=1}^K\frac{1}{n}\sum_{i\in I_k}\left(R_{Y, i}^k-\hat \beta R_{D, i}^k\right)^2\left(R_{f, i}^k\right)^2}{\left(\frac{1}{K}\sum_{k=1}^K\frac{1}{n}\sum_{i \in I_k}R_{D, i}^kR_{f, i}^k\right)^2}.
\end{equation}
Consequently, an asymptotically valid $(1-\alpha)$-confidence interval is given by
\begin{equation}\label{eq_ClassCI}
    \left[\hat\beta - z_{1-\alpha/2}\hat\sigma/\sqrt{N},\, \hat\beta - z_{1-\alpha/2}\hat\sigma/\sqrt{N}\right],
\end{equation}
where $z_{1-\alpha/2}$ is the $(1-\alpha/2)$-quantile of a standard Gaussian distribution.

\subsection{Robust Confidence Intervals}\label{sec_RobCI}
The confidence interval constructed in \eqref{eq_ClassCI} is not very robust to weak instrumental variables. In our case, weak IV means that $f(Z_i, X_i)-\phi(X_i)$ is close to zero, i.e. there is no strong effect of $Z_i$ on $D_i$ conditional on $X_i$. In such cases, it can be helpful to use an alternative construction of confidence intervals. This construction is similar to confidence intervals obtained by inverting the Anderson-Rubin test \citep{AndersonRubinTest}. We apply the method that is described in Chapter 13.3 in \cite{ChernozhukovCausalML} based on identification \eqref{eq_IdentificationLinear} to our estimator.

For $\beta\in \mathbb R$, define the functions
\begin{equation}\label{eq_RobNumHom}
    \hat Q(\beta)=\frac{1}{K}\sum_{k=1}^K\frac{1}{n}\sum_{i\in I_k}\left(R_{Y, i}^k-\beta R_{D, i}^k\right)R_{f, i}^k
\end{equation}
and 
\begin{equation}\label{eq_RobSEHom}
    \widehat{SE}_Q(\beta)=\sqrt{\frac{1}{K}\sum_{k=1}^K\frac{1}{n}\sum_{i\in I_k}\left(R_{Y, i}^k-\beta R_{D, i}^k\right)^2\left(R_{f,i}^k\right)^2-\hat Q(\beta)^2}.
\end{equation}
Under suitable conditions (see Section \ref{sec_SummaryHTE} below), under the null hypothesis $H_0:\beta = \beta_0$ (i.e. if the data is generated according to model \eqref{eq_ModelHom1} with $\beta = \beta_0$), we have that 
$$\frac{\sqrt N \hat Q(\beta_0)}{\widehat{SE}_Q(\beta_0)}\to \mathcal N(0,1)\text{ in distribution (under $H_0:\beta = \beta_0$)}$$
and consequently,
$$\mathbb P\left(\beta_0\in \left\{\beta: |\hat Q(\beta)|\leq z_{1-\alpha/2}\widehat{SE}_Q(\beta)/\sqrt{N}\right\}\right)\to 1-\alpha,$$
where $z_{1-\alpha/2}$ is the $(1-\alpha/2)$-quantile of a standard Gaussian distribution. Hence, the set $\mathcal C_N^\alpha = \left\{\beta: |\hat Q(\beta)|\leq z_{1-\alpha/2}\widehat{SE}_Q(\beta)/\sqrt{N}\right\}$ is an asymptotic $(1-\alpha)$-confidence set for $\beta$. The set $\mathcal C_N^\alpha$ is defined as the sublevel set of a parabola in $\beta$. Hence it is possible to derive an explicit expression, see Appendix \ref{app_ExplicitRobCI} for the expression in the heterogeneous case (the homogeneous case is analogue). Moreover, one can show that under some assumptions on the strength of the IV, the confidence interval $\mathcal C_N^\alpha$ is asymptotically equivalent to the classical confidence interval given in \eqref{eq_ClassCI}, see Appendix \ref{app_AdaRobCI} for the derivation in the heterogeneous case (the homogeneous case is analogue). Hence, one does not lose asymptotic efficiency from using the more robust version and we would recommend to use the robust confidence sets in practice. However, from a purely theoretical point of view, the advantage is not large: we still need that the nuisance function estimates converge fast enough, and essentially, the conditions on the convergence of the nuisance function estimates get scaled by the IV strength, which means that with weak IV, we get stronger requirements on how fast the nuisance function estimates need to converge, see also Remark \ref{rmk_HomRobCIComparison}. {Nevertheless, in simulations one can observe that the robust confidence intervals do not suffer from undercoverage in settings with weak IV where the coverage of the standard confidence intervals breaks down (see Appendix \ref{sec_VaryIVStrength}).}
\subsection{Summary of Results for the Homogeneous Treatment Effect}\label{sec_SummaryHTE}
In this section, we give a summary of the theoretical results that one can obtain for the homogeneous treatment effect. These results essentially follow from the general theory in \cite{ChernozhukovDML}, since the score
$$\psi(Y, D, X, Z; \beta, l, f,\phi) = \left(Y-l(X)-\beta(D-\phi(X))\right)\left(f(Z, X)-\phi(X)\right)$$
satisfies the Neyman orthogonality property \citep{AhrensDMLStata}, but they can also be proved by directly analyzing the expressions \eqref{eq_DefEstimator} and \eqref{eq_VarEst} by essentially following the proofs that we will give for the heterogeneous treatment effects with some slight modifications. Additionally to some conditions on the moments of the random variables, the main requirements are that the nuisance functions $l$, $\phi$ and $f$ can be estimated at a sufficiently fast rate. Concretely, define
\begin{equation}\label{eq_DefSigmaQ}
    \sigma_Q^2 = \E[\epsilon_i^2(f(Z_i, X_i)-\phi(X_i))^2],
\end{equation}
which can be seen as a measure of IV strength on the population level. Note that $\sigma_Q^2$ is unnormalized, i.e. we do not divide by the noise level $\E[\epsilon_i^2]$. In contrast to the standard theory from \cite{ChernozhukovDML} (Assumption 3.4. (d) there), we do not assume that $\sigma_Q^2$ is bounded away from $0$, but we have to rescale the conditions on the convergence rates. We write $\|f-\hat f^k\|_{L_2}$ for the $L_2$-norm of $f-\hat f^k$. Note that $f$ is fixed and $\hat f^k$ depends on the data in $I_k^c$. Hence, $\|f-\hat f^k\|_{L_2}$ is random. More precisely, it is given by
$$\|f-\hat f^k\|_{L_2}^2 = \E\left[(f(Z, X)-\hat f^k(Z, X))^2| I_k^c\right],$$
where $Z$ and $X$ are drawn from model \eqref{eq_ModelHom1} independently from $\{X_i, Y_i, Z_i, D_i\}_{i\in I_k^c}$ and by conditioning on $I_k^c$ we mean conditioning on $\{X_i, Y_i, Z_i, D_i\}_{i\in I_k^c}$. The analogous reasoning applies to {$\|\phi-\hat \phi_1^k\|_{L_2}$, $\|\phi-\hat\phi_2^k\|_{L_2}$} and $\|l-\hat l^k\|_{L_2}$. The main requirements for the nuisance function estimators are that for $k = 1,\ldots, K$,  it holds that
{\begin{equation}\label{eq_HomNuisanceCond}
    \begin{aligned}
        \frac{1}{\sigma_Q}\|l-\hat l^k\|_{L_2}\|\phi-\hat \phi_2 ^k\|_{L_2} &= o_P\left(N^{-1/2}\right),\\
        \frac{1}{\sigma_Q} \|l-\hat l^k\|_{L_2}\|f-\hat f ^k\|_{L_2}&= o_P\left(N^{-1/2}\right),\\
        \frac{1}{\sigma_Q} \|\phi-\hat \phi_1 ^k\|_{L_2}|\phi-\hat \phi_2 ^k\|_{L_2}&= o_P\left(N^{-1/2}\right),\\
        \frac{1}{\sigma_Q} \|\phi-\hat \phi_1 ^k\|_{L_2}\|f-\hat f ^k\|_{L_2}&= o_P\left(N^{-1/2}\right).
        \end{aligned}
\end{equation}}
If $\sigma_Q^2$ is bounded from below, equation \eqref{eq_HomNuisanceCond} is satisfied if all the nuisance function estimators converge at a rate of $o_P(N^{-1/4})$ which is comparably slow and satisfied by many modern machine learning methods, when assuming appropriate sparsity and/or smoothness conditions. If on the other hand, $\sigma_Q^2$ is decreasing with $N$, which corresponds to the setting of weak IV, the machine learners need to converge at faster rates, {see also Remark \ref{rmk_HomRobCIComparison} below.}

One can show the following results (we will provide a detailed treatment for the heterogeneous treatment effect below and the results here can be obtained analogously).
\begin{description}
    \item[Asymptotic normality of $\hat\beta$:] Under suitable moment conditions on the involved random variables and the functions $f$, $\phi$ and $l$ and assumptions on the nuisance function estimates including \eqref{eq_HomNuisanceCond}, we have that
    \begin{equation}\label{eq_HomAsNorm}
        \frac{\sqrt{N} (\hat\beta - \beta)}{\sigma}\to \mathcal N(0,1)
    \end{equation}
    in distribution with
    \begin{equation}\label{eq_HomDefSigma}
        \sigma^2 = \frac{\E\left[\epsilon_i^2 (\E[D_i|Z_i, X_i] - \E[D_i|X_i])^2\right]}{\E\left[(D_i-\E[D_i|X_i])(\E[D_i|Z_i, X_i]-\E[D_i|X_i])\right]^2}.
    \end{equation}
    \item[Consistency of $\hat\sigma^2$:] Under some additional conditions, we have that $|\hat\sigma^2/\sigma^2 - 1| =  o_P(1)$.
    \item[Validity of robust confidence sets:] Under similar conditions, including \eqref{eq_HomNuisanceCond}, we have that under under $H_0: \beta = \beta_0$, 
    $$\frac{\sqrt{N} \hat Q(\beta_0)}{\widehat{SE}_Q(\beta_0)}\to \mathcal N(0,1)$$
    in distribution, i.e. the robust procedure produces asymptotically valid confidence sets.
\end{description}
\begin{remark}\label{rmk_HomRobCIComparison}
    If the nuisance functions were known, the robust confidence set construction from Section \ref{sec_RobCI}, would yield asymptotically valid inference regardless of the strength of the instruments. However, this nice property does not carry over if the nuisance functions are estimated, since in \eqref{eq_HomNuisanceCond}, the convergence rates are also scaled by the strength of the IV. The fastest convergence rate for the nuisance function estimates that we can hope for is an $L_2$ error of $O_P(N^{-1/2})$ (linear regression). Hence, in the best possible case, it still follows from \eqref{eq_HomNuisanceCond} that $\sigma_Q^2\gg 1/N$. If we additionally assume that $\E[\epsilon_i^2|Z_i, X_i]$ is bounded away from $0$ and $\infty$ independently of $N$, it follows that $\E[(f(Z_i, X_i)-\phi(X_i))^2]\asymp \sigma_Q^2 \gg 1/N$. This rules out the classical ``weak IV asymptotics'' regime \citep{StaigerIVRegressionWeakInstruments} which would correspond to $\E[(f(Z_i, X_i)-\phi(X_i))^2]\asymp 1/N$. {The actual amount of weakness (in terms of the order of $\sigma_Q^2$) that one can in theory allow for, depends on what rates the machine learning algorithm can achieve. For example, for sparse high-dimensional linear models, one can obtain an $L_2$ error of $O_P(\sqrt{\log p / N})$ 
    \citep{BuehlmannHDStats}, so that the requirement on the IV strength is $\sigma_Q^2\gg (\log p)^2/N$, which is only slightly weaker than classical weak IV asymptotics. In contrast, with more flexible models such as (nonlinear) sparse high-dimensional additive models \citep{YuanMinimaxOptimalRates}, neural networks \citep{ChenImprovedRates} or (variants of) random forest \citep{WagerAdaptiveConcentration}, one can afford less weakness.}
    In conclusion, from the analysis based on theoretical bounds
    one does not gain very much from using the robust confidence set construction compared to the standard confidence interval, which is in contrast to the classical treatment of weak IV. The reason for this is that compared to the classical theory,
    we need to estimate nuisance functions. {
    However, for finite sample sizes and in settings with weak instruments, the robust confidence sets exhibit superior empirical coverage compared to the standard confidence intervals in simulations (see Appendix \ref{sec_VaryIVStrength}).}
\end{remark}

\subsection{Efficiency of the estimator $\hat\beta$}\label{sec_EfficiencyEstimator}
As detailed out in the introduction, the estimator \eqref{eq_DefEstimator} coincides with the semiparametrically efficient estimator in the homoscedastic case (see for example equation (9) in \cite{VansteelandtImprovingRobustness}). For illustration, we directly compare the asymptotic variance of $\hat\beta$ with the asymptotic variance of an estimator based on the identification \eqref{eq_IdentificationZeta}. Such an estimator $\beta_\zeta$ is obtained by replacing $R_{f, i}^k$ in \eqref{eq_DefEstimator} by $R_{\zeta, i}^k$, where $R_{\zeta, i}^k = \hat\zeta^k(Z_i, X_i) - \hat\mu^k(X_i)$, with $\hat\mu^k(X_i)$ being an estimate of $\mu(X_i) = \E[\zeta(Z_i, X_i)|X_i]$ and $\hat \zeta^k$ is either estimated or known a priori, as for example in the standard DML with linear instruments \citep{ChernozhukovDML, EmmeneggerRegularizingDML} which is recovered by setting $\zeta(Z_i, X_i) = Z_i$. Analogously to $\hat\beta$, one can show that $\hat\beta_\zeta$ is asymptotically normal with asymptotic variance given by
\begin{equation}\label{eq_AsVarZeta}
    \sigma^2_\zeta = \frac{\E\left[\epsilon_i^2 (\zeta(Z_i, X_i) - \E[\zeta(Z_i, X_i)|X_i])^2\right]}{\E\left[(D_i-\E[D_i|X_i])(\zeta(Z_i, X_i)-\E[\zeta(Z_i, X_i)|X_i])\right]^2}.
\end{equation}
In the homoscedastic setting $\E[\epsilon_i^2|Z_i, X_i] = \mathrm{const.}$, the asymptotic variance of our estimator $\hat\beta$ is always smaller or equal the asymptotic variance of $\hat\beta_\zeta$ with equality if and only if the function $f(Z_i, X_i)$ is equal to the function $\zeta(Z_i,X_i)$ plus a potential shift $\psi(X_i)$ depending only on $X_i$, as the following proposition shows. 
\begin{proposition}\label{pro_CompVar}
   Assume that $\E[\epsilon_i^2|Z_i, X_i]=\E[\epsilon_i^2]$ a.s. Then, $\sigma^2\leq \sigma_\zeta^2$ with equality if and only if there exist $\alpha \neq 0$ and a function $\psi$ such that
    \begin{equation}\label{eq_CondVarEqual}
        f(Z_i, X_i)=\alpha\zeta(Z_i, X_i) + \psi(X_i)\text{ a.s.}
    \end{equation}
    Here, $\sigma^2$ is as in \eqref{eq_HomDefSigma}, namely the asymptotic variance of the proposed estimator $\hat{\beta}$ in \eqref{eq_DefEstimator}.
\end{proposition}
A proof of Proposition \ref{pro_CompVar} can be found in Appendix \ref{app_ProofEfficiency}.
The proposition states that in the homoscedastic setting, we are estimating the optimal instrument. In particular, the estimator $\hat\beta$ is asymptotically more efficient than the standard estimator for double machine learning with instrumental variables \citep{ChernozhukovDML, EmmeneggerRegularizingDML} with equal asymptotic variance if and only if there exist $\alpha\neq 0$ and a function $\psi$ such that $f(Z_i, X_i) = \alpha Z_i + \psi(X_i)$, that is, the instrument has a linear effect on the treatment.

\section{Heterogeneous Treatment Effect}\label{sec_HetTE}
Let us now consider the heterogeneous treatment effect model \eqref{eq_Model}, which we restate for convenience: we assume that we observe $N$ i.i.d. realizations from
\begin{equation}\label{eq_ModelHet}
Y_i=\beta(V_i) D_i+g(X_i)+\epsilon_i.
\end{equation}
The only change compared to the model \eqref{eq_ModelHom1} is that the treatment effect is assumed to be heterogeneous with respect to the univariate continuous random variable $V_i$ and we again assume the availability of an instrumental variable $Z_i\in \mathbb R^d$. As for the homogeneous treatment effect, we allow the data generating process to vary with $N$, even if it is not reflected in the notation.

The goal is to perform inference for $\beta(v)$, that is, the value of the treatment effect for a particular value of $V_i = v\in \mathbb R$. We need the following extension of Assumption \ref{ass_Hom} to make $\beta(v)$ identifiable.
\begin{assumption}\label{ass_Het}
\mbox\newline
    \begin{enumerate}
        \item $\E\left[\left(\E[D_i|Z_i, X_i]- \E[D_i|X_i]\right)^2| V_i = v\right]\neq 0$,
        \item $\E[\epsilon_i|Z_i, X_i] = 0$ ,
        \item $V_i$ is measurable with respect to the sigma-algebra generated by $X_i$,
        \item The distribution of $V_i$ is absolutely continuous with respect to the Lebesgue measure and has density $p_V(\cdot)$.
    \end{enumerate}
\end{assumption}
Assertion 1 encodes that $Z_i$ is a relevant IV given $X_i$ conditional on $V_i = v$. Assertion 2 is the same as for Assumption \ref{ass_Hom} and states that $Z_i$ is a valid instrumental variable. Assertion 3 states that $V_i$ is measurable with respect to $X_i$. Usually, $V_i$ is just a coordinate of $X_i$ and in fact, assertion 3 of Assumption \ref{ass_Het} can be made without loss of generality since we can otherwise just replace $X_i \leftarrow (X_i, V_i)$. {The absolute continuity of the distribution of $V_i$ with respect to the Lebesgue measure in assertion 4 is assumed for the sake of conciseness of the exposition. An analogous approach would work for discrete random variables $V_i$, which would reduce to simple binning.} Under Assumption \ref{ass_Het}, we can identify $\beta(v)$ by simply making the identification from the homogeneous case conditional on $V_i = v$, i.e.

\begin{equation}\label{eq_IdentificationHet}
    \beta(v) = \frac{\E\left[(Y_i - \E[Y_i|X_i])(\E[D_i|Z_i, X_i]-\E[D_i|X_i])|V_i = v\right]}{\E\left[(D_i - \E[D_i|X_i])(\E[D_i|Z_i, X_i]-\E[D_i|X_i])|V_i = v\right]}.
\end{equation}

\subsection{Estimators}
To construct an estimator from this identification, we use ideas from univariate kernel smoothing. For this, we need a kernel $K:\mathbb R\to [0,\infty)$ that satisfies the following assumption (see e.g. \cite{WassermannAllOfNonparametricStatistics}).
\begin{assumption}\label{ass_Kernel}
\mbox\newline
    \begin{enumerate}
        \item $\int_{-\infty}^\infty K(x) \mathrm{d}x = 1$,
        \item K is symmetric, i.e. for all $x\in \mathbb R$, $K(x) = K(-x)$,
        \item $\|K\|_\infty<\infty$, i.e. $K$ is bounded,
        \item $\int_{-\infty}^\infty x^2 K(x)\mathrm{d}x \in (0,\infty)$.
    \end{enumerate}
\end{assumption}
Consider a bandwidth $h>0$ and recall the cross-fitted residuals \eqref{eq_DefRY}, \eqref{eq_DefRD} and \eqref{eq_DefRf}. We can then modify the estimator \eqref{eq_DefEstimator} by simply weighting the products of the residuals according to the kernel, i.e.
\begin{equation}\label{eq_DefEstimatorHet}
\hat\beta(v) = \frac{\frac{1}{Nh}\sum_{k=1}^K\sum_{i\in I_k} R_{Y, i}^k R_{f,i}^k K\left(\frac{V_i - v}{h}\right)}{\frac{1}{Nh}\sum_{k=1}^K\sum_{i\in I_k} R_{D, i}^k R_{f,i}^k K\left(\frac{V_i - v}{h}\right)}.
\end{equation}
The main intuition for \eqref{eq_DefEstimatorHet} is that we can view the treatment effect as being homogeneous when we restrict the attention to $V_i$ being near the target level $v$.
In a similar fashion to the homogeneous treatment effect, we will show below in Section \ref{subsec_theory} that $\sqrt{Nh} (\hat\beta(v)-\hat\beta(v))/\sigma(v)\to \mathcal N(0,1)$ with asymptotic variance $\sigma^2(v)$ that can be consistently estimated by
\begin{equation}\label{eq_DefVarianceHet}
    \hat\sigma^2(v) = \frac{\frac{1}{Nh}\sum_{k=1}^K\sum_{i\in I_k}\left(R_{Y, i}^k-\hat \beta(v)R_{D, i}^k\right)^2\left(R_{f, i}^k\right)^2K\left(\frac{V_i - v}{h}\right)^2}{\left(\frac{1}{Nh}\sum_{k = 1}^K\sum_{i\in I_k} R_{D, i}^kR_{f, i}^kK\left(\frac{V_i - v}{h}\right)\right)^2}.
\end{equation}

\begin{remark}\label{rmk_ZetaHet}
    As for the homogeneous treatment effect, in principle one can replace $R_{f,i}^k$ by $R_{\zeta, i}^k = \hat\zeta^k(Z_i, X_i)-\hat\mu^k(X_i)$ for an arbitrary function $\zeta(Z_i, X_i)$ and $\hat\mu^k(X_i)$ being an estimate of $\mu(X_i)=\E[\zeta(Z_i, X_i)|X_i]$. Most prominently, one could take the ``linear IV'' estimator $\zeta(Z_i, X_i) = Z_i$ and $\mu(X_i) = \E[Z_i|X_i]$. However, when using the same kernel and bandwidth, the choice $\zeta(Z_i, X_i) = f(Z_i, X_i)$ leads to the most efficient estimator of this form under homoscedasticity, as we will see in Remark \ref{rmk_EfficiencyHet}.
\end{remark}

\subsection{Robust Confidence Sets}
Also for the heterogeneous treatment effect, we can apply the robust confidence set construction. For $\gamma \in \mathbb R$, let
\begin{equation*}
    \hat Q(\gamma, v) = \frac{1}{Nh}\sum_{k = 1}^K\sum_{i\in I_k} \left(R_{Y, i}^k- \gamma R_{D, i}^k\right)R_{f, i}^kK\left(\frac{V_i - v}{h}\right),
\end{equation*}
and
\begin{equation*}
    \widehat{SE}^2_Q(\gamma, v) = \frac{1}{Nh}\sum_{k = 1}^K\sum_{i\in I_k} \left(R_{Y, i}^k- \gamma R_{D, i}^k\right)^2\left(R_{f, i}^k\right)^2K\left(\frac{V_i - v}{h}\right)^2- h\hat Q(\gamma, v)^2.
\end{equation*}
Below in Section \ref{subsubsec_robustCI}, we will give conditions such that under $H_0: \beta(v) = \beta_0(v)$,
$$\frac{\sqrt{Nh} \hat Q(\beta_0(v), v)}{\widehat{SE}_Q(\beta_0(v), v)}\to \mathcal N(0,1),$$
and consequently, $\mathcal C_N^\alpha(v) \coloneqq \left\{\gamma: |\hat Q(\gamma, v)|\leq z_{1-\alpha/2}\widehat {SE}_Q(\gamma, v)/\sqrt{Nh} \right\}$ is an asymptotic $(1-\alpha)$-confidence set for $\beta(v)$, similarly to the homogeneous case.

\begin{remark}
In Appendix \ref{app_ExplicitRobCI}, we derive an explicit expression for the robust confidence sets and in Appendix \ref{app_AdaRobCI}, we show that under certain conditions on the model, the robust confidence sets are asymptotically equivalent to standard confidence intervals based on $\hat\sigma^2(v)$. This means that one does not lose asymptotic efficiency by using the robust confidence sets.
\end{remark}

\subsection{Theory}\label{subsec_theory}
In classical univariate kernel smoothing, the optimal bandwidth satisfies $h\asymp N^{-1/5}$. This yields the smallest mean squared error of the nonparametric estimator, since with this choice, the bias and the variance are of the same order. However, to obtain valid inference for $\beta(v)$, we need that the variance dominates the bias and hence, one needs to choose a smaller bandwidth (``undersmoothing'') \citep{WassermannAllOfNonparametricStatistics}. This is also true for our setting and we need the following assumption.
\begin{assumption}\label{ass_Bandwidth}
    $h\ll N^{-1/5}$.
\end{assumption}
The remaining assumptions can be divided into conditions on the convergence rates of the nuisance function estimates and some more technical (mainly smoothness) assumptions needed to apply the kernel smoothing successfully. We start with the assumptions on the nuisance function estimates. In contrast to the usual conditions for double/debiased machine learning that essentially require that all the nuisance functions converge faster than $N^{-1/4}$ in $L_2$-norm (see condition \eqref{eq_HomNuisanceCond} above), we need slightly faster convergence requirements on the nuisance functions, which we outline below. The reason for this is that we only want to make assumptions on the overall convergence rate of the nuisance functions, e.g., $\| \hat f-f\|_{L_2}$ and not make any assumptions on convergence rates that are somehow localized/conditional on $V_i = v$.

Let us introduce the quantity
    \begin{equation}\label{eq_DefSigmaQHet}
        \sigma_Q^2(v) = \E[\epsilon_i^2 (f(Z_i, X_i)-\phi(X_i))^2|V_i = v] p_V(v)\int_{-\infty}^\infty K(x)^2 \mathrm{d}x,
    \end{equation}
where we recall that $p_V(\cdot)$ is the density of $V_i$ with respect to the Lebesgue measure. Similarly to the homogeneous treatment effect, $\sigma_Q^2(v)$ can be seen as an unnormalized measure of IV strength conditional on $V_i = v$ on the population level.

\begin{assumption}\label{ass_NuisanceRatesHet}
\mbox\newline
\begin{enumerate}
    \item For $k = 1,\ldots, K$, it holds that $\|f- \hat f^k\|_{L_2}/\sigma_Q(v) = o_P\left(N^{-1/2} h^{-1}\right)$ and it holds that {$\|\phi- \hat \phi_2^k\|_{L_2}/\sigma_Q(v) = o_P\left(N^{-1/2} h^{-1}\right)$}.\label{ass_Rate_F_Phi_fast}
    \item For $k = 1,\ldots, K$, it holds that $\|f- \hat f^k\|_{L_2}/\sigma_Q(v) = o_P\left(h^{1/2}\right)$ and it holds that {$\|\phi- \hat \phi_2^k\|_{L_2}/\sigma_Q(v) = o_P\left(h^{1/2}\right)$}. \label{ass_Rate_F_Phi_slow}
    \item For $k = 1,\ldots, K$, it holds that {
    \begin{align*}
        \frac{1}{\sigma_Q(v)}\|f- \hat f^k\|_{L_2} \|l-\hat l^k\|_{L_2} &= o_P\left(N^{-1/2} h^{1/2}\right),\\
        \frac{1}{\sigma_Q(v)}\|\phi- \hat \phi_2^k\|_{L_2} \|l-\hat l^k\|_{L_2}&= o_P\left(N^{-1/2} h^{1/2}\right),\\
        \frac{1}{\sigma_Q(v)}\|f- \hat f^k\|_{L_2} \|\phi-\hat \phi_1^k\|_{L_2}&= o_P\left(N^{-1/2} h^{1/2}\right),\\
        \frac{1}{\sigma_Q(v)}\|\phi- \hat \phi^k_1\|_{L_2} \|\phi- \hat \phi^k_2\|_{L_2} &= o_P\left(N^{-1/2} h^{1/2}\right).
    \end{align*}
    }
   \label{ass_RateProducts}
    \item There exists $C<\infty$ independent of $N$ such that $\E[\epsilon_i^2|Z_i, X_i]\leq C$ a.s., that is, $\|\E[\epsilon_i^2|Z_i, X_i]\|_{\infty}\leq C < \infty$.\label{ass_CondVariance}
    \item There exist $p_1, q_1 \in [1, \infty]$ with $1/p_1+1/q_1 = 1$ and $C<\infty$ independent of $N$ such that $\|\E[(f(Z_i, X_i)-\phi(X_i))^2|X_i]\|_{L_{p_1}}/\sigma_Q^2(v)\leq C$ and for $k = 1,\ldots, K$ it holds that  $\|l-\hat l^k\|_{L_{2q_1}} = o_P\left(h^{1/2}\right)$. \label{ass_Rate_l_pq1}
    \item There exist $p_2, q_2 \in [1, \infty]$ with $1/p_2+1/q_2 = 1$ and $C<\infty$ independent of $N$ such that $\|\E[(f(Z_i, X_i)-\phi(X_i))^2|X_i]\|_{L_{p_2}}/\sigma_Q^2(v)\leq C$ and for $k = 1,\ldots, K$, it holds that {$\|\phi-\hat \phi_1^k\|_{L_{2q_2}} = o_P\left(h^{1/2}\right)$}.
\end{enumerate}
\end{assumption}
\begin{remark}
    Assume that $\sigma_Q^2(v) \asymp 1$, which is for example the case if the data generating process \eqref{eq_ModelHet} does not change with $N$. From Assumption \ref{ass_Bandwidth}, we know that $h\ll N^{-1/5}$. Hence, assertion 1 of Assumption \ref{ass_NuisanceRatesHet} is satisfied if $\|f-\hat f^k\|_{L_2} = O_P(N^{-3/10})$ and $\|\phi-\hat \phi_2^k\|_{L_2} = O_P(N^{-3/10})$, which is a slightly stronger requirement than $o_P\left(N^{-1/4}\right)$. The second assertion is comparably weak. The third assertion are the conditions on the products of the $L_2$ errors of the nuisance function estimators. Since $h\ll N^{-1/5}$, these assumptions are stronger than assuming that the respective products are $o_P\left(N^{-3/5}\right)$. Hence, the assumptions on the cross terms are also slightly stronger than for the homogeneous treatment effect.
\end{remark}
When $\sigma_Q^2(v) \to 0$, this corresponds to models with weak IV. As already discussed in Remark \ref{rmk_HomRobCIComparison} we cannot allow $\sigma_Q^2(v)$ to go to zero too fast since otherwise, the requirements on the convergence rates for the nuisance functions cannot be satisfied.

The remaining assumptions are more technical and we defer them to Assumption \ref{ass_RegularityHet} in Appendix \ref{app_TechAssHTE}. Assumption \ref{ass_RegularityHet} can be summarized as follows.
\begin{enumerate}
    \item $p_V(v)$ is bounded away from $0$.
    \item If $\sigma_Q^2(v)\to 0$ as $N\to\infty$, this is attributed to weak IV, i.e. $\E[(f(Z_i, X_i)-\phi(X_i))^2|V_i = v]\to 0$ and not to other changes in the data generating process (for example not because $\E[\epsilon_i^2|Z_i, X_i]\to 0$).
    \item Various functions $x\mapsto \rho(x)$ are twice differentiable for $x$ in the support of $V_i$ with bounded second derivative (when properly rescaled by $\sigma_Q(v)$). The functions $\rho(\cdot)$ are products of $\beta(\cdot)$, $p_V(\cdot)$ and various conditional moments as for example $\E[(f(Z_i, X_i)-\phi(X_i))^2|V_i = \cdot]$.
\end{enumerate}
With these assumptions, we obtain the asymptotic normality of the estimator $\hat \beta(v)$.
\begin{theorem}\label{thm_AsNormHet}
    Under Assumptions \ref{ass_Het}, \ref{ass_Kernel}, \ref{ass_Bandwidth}, \ref{ass_NuisanceRatesHet} and \ref{ass_RegularityHet}, additionally assume that  $\sigma_Q^2(v) \gg (Nh)^{-1}$, $\frac{\|f-\hat f^k\|_{L_2}}{\sigma_Q^2(v)} = o_P\left(h^{1/2}\right)$ and {$\frac{\|\phi-\hat \phi_2^k\|_{L_2}}{\sigma_Q^2(v)} = o_P\left(h^{1/2}\right)$}. Then, we have
    $$\frac{\sqrt{Nh}(\hat\beta(v)-\beta(v))}{\sigma(v)}\to\mathcal N(0, 1)\text{ in distribution,}$$
    where
    $$\sigma^2(v) = \frac{\E[\epsilon_i^2(f(Z_i, X_i)-\phi(X_i))^2|V_i = v]\int_{-\infty}^\infty K(x)^2\mathrm{d}x}{\E[(D_i-\phi(X_i))(f(Z_i, X_i)-\phi(X_i))|V_i = v]^2 p_V(v)}.$$
\end{theorem}

The following theorem ensures that the estimator $\hat\sigma^2(v)$ defined in \eqref{eq_DefVarianceHet} is a consistent estimator of $\sigma^2(v)$ hence one can construct asymptotic confidence intervals for $\beta(v)$.
\begin{theorem}\label{thm_VarHet}
    Additionally to the conditions of Theorem \ref{thm_AsNormHet}, assume that $\frac{\|f-\hat f^k\|_{L_4}}{\sigma_Q(v)}=O_p\left(h^{1/4}\right)$ and {$\frac{\|\phi-\hat \phi_2^k\|_{L_4}}{\sigma_Q(v)}=O_p\left(h^{1/4}\right)$} Then, we have
    $$\frac{\hat\sigma^2(v)}{\sigma^2(v)}\to 1\text{ in probability.}$$
\end{theorem}

\begin{remark}\label{rmk_EfficiencyHet}
Regarding the efficiency of the estimator $\hat\beta(v)$, we can exactly follow the argument from the proof of Proposition \ref{pro_CompVar}. Concretely, given a fixed kernel $K$, fixed bandwidth $h$, there is no transformation $\zeta(Z_i, X_i)$ such that the resulting estimator according to Remark \ref{rmk_ZetaHet} is more efficient than our estimator $\hat\beta(v)$ under homoscedasticity $\E[\epsilon_i^2|Z_i, X_i] = \mathrm{const.}$
\end{remark}

\subsubsection{Robust Confidence Sets}\label{subsubsec_robustCI}
The following theorem ensures that under the null hypothesis, $\hat Q(\beta(v), v)$ is asymptotically normal.
\begin{theorem}\label{thm_RobNumHet}
    Under Assumptions \ref{ass_Het}, \ref{ass_Kernel}, \ref{ass_Bandwidth}, \ref{ass_NuisanceRatesHet} and \ref{ass_RegularityHet}, we have that under $H_0: \beta(v) = \beta_0(v)$
    $$\frac{\sqrt{Nh}\hat Q(\beta_0(v), v)}{\sigma_Q(v)}\to \mathcal N(0, 1) \text{ in distribution}.$$

\end{theorem}
Moreover, the asymptotic variance of $\hat Q(\beta(v), v)$ can be consistently estimated, and hence, the approach for robust confidence sets is valid.
\begin{theorem}\label{thm_RobVarHet}
Under the assumptions of Theorem \ref{thm_RobNumHet} we have that under $H_0:\beta(v) = \beta_0(v)$
$$\frac{\widehat{SE}_Q^2(\beta_0(v), v)}{\sigma_Q^2(v)}\to 1\text{ in probability.}$$
\end{theorem}

\begin{remark}\label{rmk_WeakerConditions}
For the validity of the robust confidence sets, one still needs Assumption \ref{ass_NuisanceRatesHet} which does not allow for IV that are too weak (see also Remark \ref{rmk_HomRobCIComparison}). However, we do not need the additional conditions $\|f-\hat f^k\|_{L_2}/\sigma_Q^2(v),\, \|\phi-\hat \phi_2^k\|_{L_2}/\sigma_Q^2(v) = o_P(h^{1/2})$ that -- given a rate of convergence for $\hat f^k$ and $\hat \phi_2^k$ -- pose even stronger restrictions on the IV-strength $\sigma_Q^2(v)$. In that sense, the robust confidence set construction works under slightly less restrictive conditions than the asymptotic normality of the point estimator.
\end{remark}

\section{Applications}\label{sec_Applications}
We evaluate our method on simulated and real-world datasets.
The method is implemented in the \texttt{R} package \texttt{IVDML} {available on CRAN \citep{IVDMLCRAN}.}
The code to reproduce the simulation and real data results is available on GitHub.\footnote{\url{https://github.com/cyrillsch/IVDML_Application}} Before considering simulations and real data applications, we point out an aspect that is important when applying the methods in practice.

\subsection{Repeated Cross-Fitting}\label{sec_RepeatedCrossFitting}
The partition of the dataset into the $K$ folds is random. The finite sample estimators $\hat\beta$ and $\hat\sigma^2$ for the homogeneous treatment effect and $\hat\beta(v)$ and $\hat\sigma^2(v)$ for the heterogeneous treatment effect depend on this random partition. Let us focus on the homogeneous treatment effect in the following (the aggregation for the heterogeneous treatment effect is analogous). Following \cite{ChernozhukovDML}, the procedure can be repeated $S$ times with independent random partitions, yielding estimators $(\beta_s)_{s=1,\ldots, S}$ and $(\hat\sigma^2_{s})_{s=1,\ldots, S}$. The point estimator is then aggregated using the median, $\hat\beta^* = \mathrm{median}\{\hat\beta_s;\, s = 1,\ldots, S\}$ and the asymptotic variance additionally is updated by a term incorporating the sample variability, $\hat\sigma^{2, *} = \mathrm{median}\{\hat\sigma_s^2 + (\hat \theta_s - \hat\theta^*)^2;\, s = 1,\ldots, S\}$. 
For the robust confidence set construction, we use the same idea. Let $\left(\hat Q_s(\cdot)\right)_{s = 1,\ldots, S}$ and $\left(\widehat{SE}_{Q,s}^2(\cdot)\right)_{b = 1,\ldots, B}$ be obtained from \eqref{eq_RobNumHom} and \eqref{eq_RobSEHom}. Define $\hat Q^*(\cdot) = \mathrm{median}\left\{\hat Q_s(\cdot);\, s = 1,\ldots, S\right\}$ and $\widehat{SE}_{Q}^{2, s}(\cdot) = \mathrm{median}\left\{\widehat{SE}_{Q,s}^2(\cdot) + \left(\hat Q_s(\cdot)-\hat Q^*(\cdot)\right)^2;\, s = 1,\ldots, S\right\}$, where the medians are taken argument-wise.

\subsection{Simulations}\label{sec_Simulations}
We simulate from model \eqref{eq_ModelHet} using the following specification:
\begin{equation}\label{eq_SimSetup}
\begin{aligned}
    X_i, H_i, E_{Z,i}, E_{\delta, i}, E_{\epsilon, i} &\sim_{i.i.d.} \mathcal N(0, 1),\\
    Z_i&\gets 0.5 X_i + E_{Z,i}\\
    \delta_i &\gets 0.7 H_i + 0.7 E_{\delta, i},\\
    \epsilon_i &\gets \mathrm{sign}(H_i) - 0.5 + 0.5 E_{\epsilon, i},\\
    D_i&\gets f(Z_i, X_i) + \delta_i,\\
    V_i&\gets X_i,\\
    Y_i&\gets \beta(V_i) D_i + g(X_i) + \epsilon_i.
\end{aligned}
\end{equation}
In \eqref{eq_SimSetup}, $H_i$ acts as unobserved confounding, which is the reason we need the instrumental variable $Z_i$.
We fix $g(x)= \tanh(x)$ and consider combinations of the following settings for $f$ and $\beta$:
\begin{description}
    \item[(hom.)] $\beta(v) = 1$ (homogeneous treatment effect),
    \item[(het.)] $\beta(v) = 2\exp(-v^2/2)$ (heterogeneous treatment effect),
    \item[(Z lin.)] $f(z, x) = -\sin(x) + z$ (instrument acts linearly on treatment),
    \item[(Z nonlin.)] $f(z, x) = -\sin(x) + \cos(z) + 0.2z$ (instrument acts nonlinearly on treatment).
\end{description}
Moreover, we fit the following models, where we use 5-fold cross-fitting repeated $S = 10$ times and aggregate the results using the method described in Section \ref{sec_RepeatedCrossFitting}.
\begin{description}
    \item[(hom. T.E. mlIV)] The DML estimator, see \eqref{eq_DefEstimator} and Remark \ref{rmk_TwoStage} for the homogeneous treatment effect with machine learning instruments.
    \item[(hom. T.E. linearIV)] The original DML estimator for the homogeneous treatment effect based on \eqref{eq_IdentificationLinear} using the instrument $Z_i$ linearly \citep{ChernozhukovDML, EmmeneggerRegularizingDML}.
    \item[(het. T.E. mlIV)] The estimator \eqref{eq_DefEstimatorHet} for the heterogeneous treatment effect with machine learning instruments. Depends on the choice of a bandwidth and a kernel.
    \item[(het. T.E. linearIV)] Variant of the estimator \eqref{eq_DefEstimatorHet} for the heterogeneous treatment that uses the instruments only linearly. Depends on the choice of a bandwidth and a kernel.
\end{description}
In all the following simulations, we estimate the nuisance functions using generalized additive models with the function \texttt{gam} from the R package \texttt{mgcv} \citep{WoodGLM1, WoodGLM2} using the default settings. As kernel, we use the Epanechnikov kernel \citep{WassermannAllOfNonparametricStatistics}, $K(x) = \frac{3}{4\sqrt 5}(1-x^2/5)1\{|x|\leq \sqrt 5\}$, which is scaled such that $\int x^2K(x)\mathrm{d}x = 1$.

First, we simulate two datasets with $N = 1000$ i.i.d. samples, one from setting (het.)/(Z lin.) and one from setting (het.)/(Z nonlin.) and apply the methods (het. T.E. mlIV) and (het. T.E. linearIV). We use two different bandwidths. The first one is the normal reference rule / Silverman's rule of thumb \citep{SilvermanDensityEstimation, WassermannAllOfNonparametricStatistics} and is given by $h_N = 1.06\min(s, Q/1.34)N^{-1/5}$, where $s$ is the sample standard deviation and $Q$ the interquartile range of $\{V_i\}_{i=1}^N$. For valid inference, we need undersmoothing $h_N\ll N^{-1/5}$ (Assumption \ref{ass_Bandwidth}) and hence for the smaller bandwidth, we replace the $N^{-1/5}$ by $N^{-2/7}$ in the normal reference rule, which is also, what is done in \cite{FanEstimationOfCATE}. The estimated heterogeneous treatment effect functions together with (pointwise) standard confidence intervals (dashed lines) and (pointwise) robust confidence sets (shaded region) can be seen in Figure \ref{fig_Visualization}. From the plots, we see that in the setting (Z lin.), the methods (het. T.E. mlIV) and (het. T.E. linearIV) yield almost identical results, both in terms of point estimate and confidence sets. In contrast, in the setting (Z nonlin.), the method (het. T.E. mlIV) seems to yield more stable results with smaller confidence bounds, whereas the results using (het. T.E. linearIV) are very unstable even in regions where there are a lot of samples $V_i$. This is consistent with our theory (see Remarks \ref{rmk_ZetaHet} and \ref{rmk_EfficiencyHet}). Moreover, we also observe that the standard confidence intervals and the robust confidence sets are almost identical in regions with a lot of samples $V_i$ and large instrument strength. In contrast, if either there are not a lot of samples $V_i$ near a point $v$ (e.g. at the boundary) or the strength of the instrument is not very strong (e.g. method (het. T.E. linearIV) in the setting (Z nonlin.)), the robust confidence sets can be considerably larger than the standard confidence sets.

\begin{figure}[t]
\centering
\includegraphics[width=0.8\textwidth]{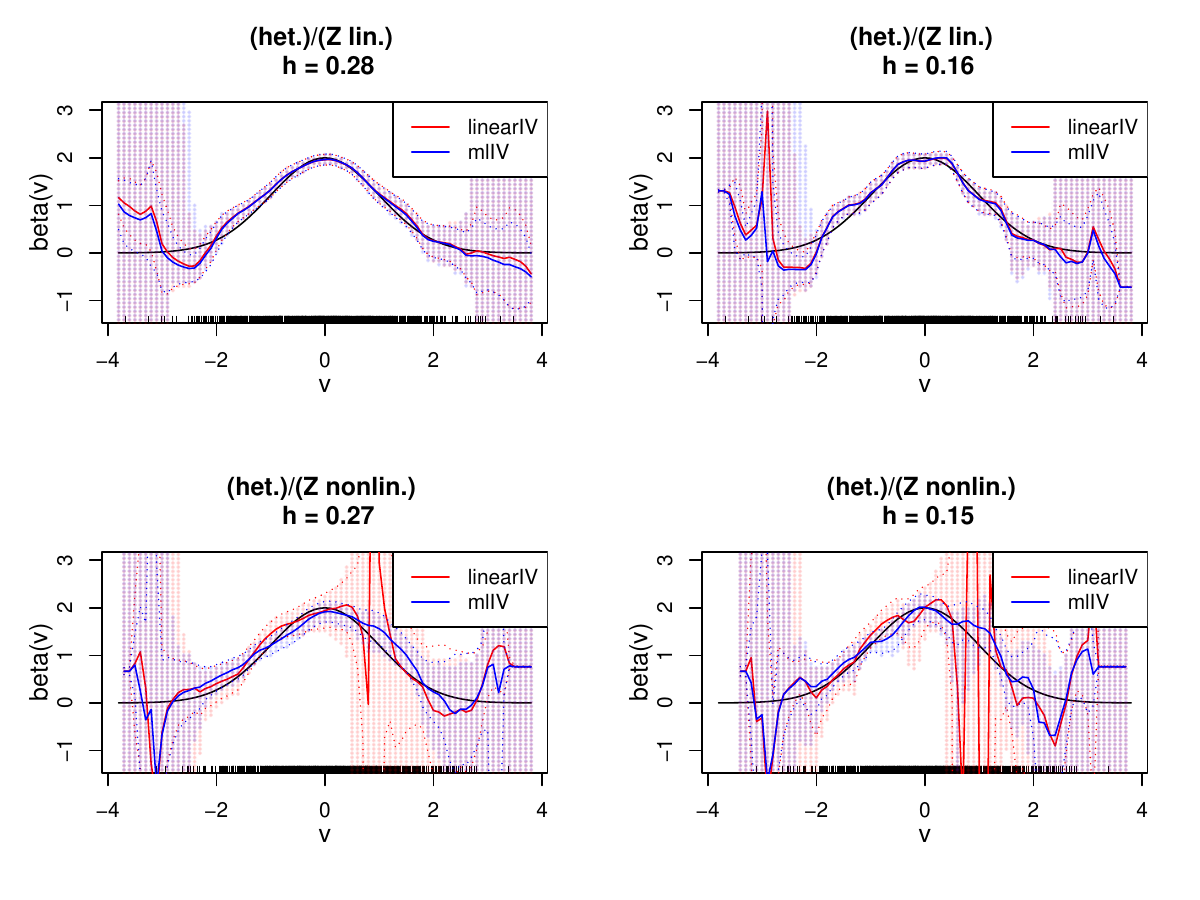}
\caption{Visualization for setting (het.)/(Z lin.) (top) and (het.)/(Z nonlin.) (bottom). On the left, the bandwidth $h$ is chosen according to the normal reference rule / Silverman's rule of thumb. On the right, the bandwidth is chosen according to the normal reference rule multiplied by $N^{1/5}/N^{2/7}$ to ensure undersmoothing. The black solid line is the true heterogeneous treatment effect. Solid lines are the point estimates for the methods (het. T.E. linearIV) (red) and (het. T.E. mlIV) (blue). Dashed lines are the (pointwise) standard 95\%-confidence interval. Shaded regions are the (pointwise) robust 95\%-confidence sets. Red-shaded regions correspond to (het. T.E. linearIV) and blue-shaded regions correspond to (het. T.E. mlIV). Purple-shaded regions are the intersection of the two.}
\label{fig_Visualization}
\end{figure}

We now want to investigate the efficiency of the different estimators and the coverage of the different confidence sets more extensively. For each combination of setting (hom.)/(het.) with (Z lin.)/(Z nonlin.) we consider sample size $N\in \{250, 500, 1000, 1500, 2000\}$ and simulate {$1000$} datasets from \eqref{eq_SimSetup}. We still use generalized additive models (\texttt{gam}) for the nuisance functions and the Epanechnikov kernel for methods (het. T.E. mlIV) and (het. T.E. linearIV). Moreover, we only consider undersmoothing for the bandwidth, i.e. we use the normal reference rule times $N^{1/5}/N^{2/7}$. Let us first consider the settings (hom.)/(Z lin.) and (hom.)/(Z nonlin.). Here, we use the methods (hom. T.E. mlIV) and (hom. T.E. linearIV) as well as the methods (het. T.E. mlIV) and (het. T.E. linearIV) for $\hat\beta(0)$ and $\hat\beta(1.5)$. The results are shown in Figure \ref{fig_SimHom1DStrong}. As it is to be expected, the mean squared errors (MSEs) of the estimators decrease with $N$ and the MSEs of the (hom. T.E.) estimators are smaller than the MSEs of the comparable (het. T.E.) estimators. The MSEs of the (het. T.E.) estimators for $v = 0$ are smaller than the MSEs for the (het. T.E.) estimators for $v = 1.5$, which is due to the fact that the density of $V_i$ is larger at $v = 0$ than at $v = 1.5$. Moreover, we see that both the standard and robust confidence sets exhibit good empirical coverage close to the desired 0.95. Regarding the difference between settings (Z lin.) and (Z nonlin.), we see that when the instrument $Z$ acts linearly on the treatment (Z lin.), the performance of linearIV and mlIV is very similar for both the homogeneous treatment effect and the heterogeneous treatment effect. On the other hand, we observe that in the setting (Z nonlin.), the methods based on mlIV have significantly lower MSE than their corresponding linearIV counterpart. This reflects the theoretical results that using the optimal machine learning instruments can lead to a significant efficiency gain.
\begin{figure}[t]
\centering
\includegraphics[width=0.8\textwidth]{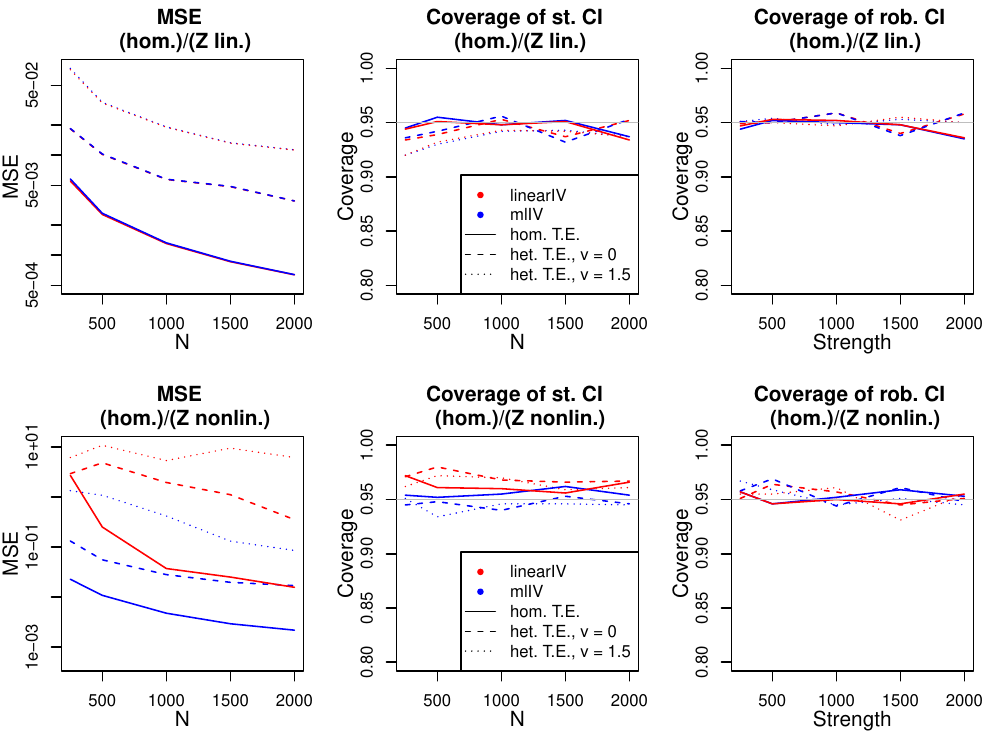}
\caption{Simulation results for setting (hom.)/(Z lin.) (top) and (hom.)/(Z nonlin.) (bottom). Left panel: mean squared error (MSE) of $\hat \beta$ (hom. T.E.), $\hat\beta(0)$ (het. T.E. for $v = 0$) and $\hat\beta(1.5)$ (het. T.E. for $v = 1.5$). Middle panel: coverage of standard confidence intervals. Right panel: coverage of robust confidence sets. Methods based on (linearIV) are in red and methods based on (mlIV) are in blue. Solid lines correspond to the methods (hom. T.E. linearIV) and (hom. T.E. mlIV). Dashed lines correspond to the methods (het. T.E. linearIV) and (het. T.E. mlIV) for $v = 0$. Dotted lines correspond to the methods (het. T.E. linearIV) and (het. T.E. mlIV) for $v = 1.5$.} 
\label{fig_SimHom1DStrong}
\end{figure}

Let us now consider the settings (het.)/(Z lin.) and (het.)/(Z nonlin.). Here, we only fit the methods that actually estimate a heterogeneous treatment effect. The results are shown in Figure \ref{fig_SimHet1DStrong}. The overall picture is similar to before. However, the coverage of the standard confidence interval for the (het. T.E.) methods at $v = 1.5$ are a bit too low, but the coverage of the robust confidence sets is better. In the setting (Z nonlin.), the methods based on mlIV have again lower MSEs than the methods based on linearIV.
\begin{figure}[t]
\centering
\includegraphics[width=0.8\textwidth]{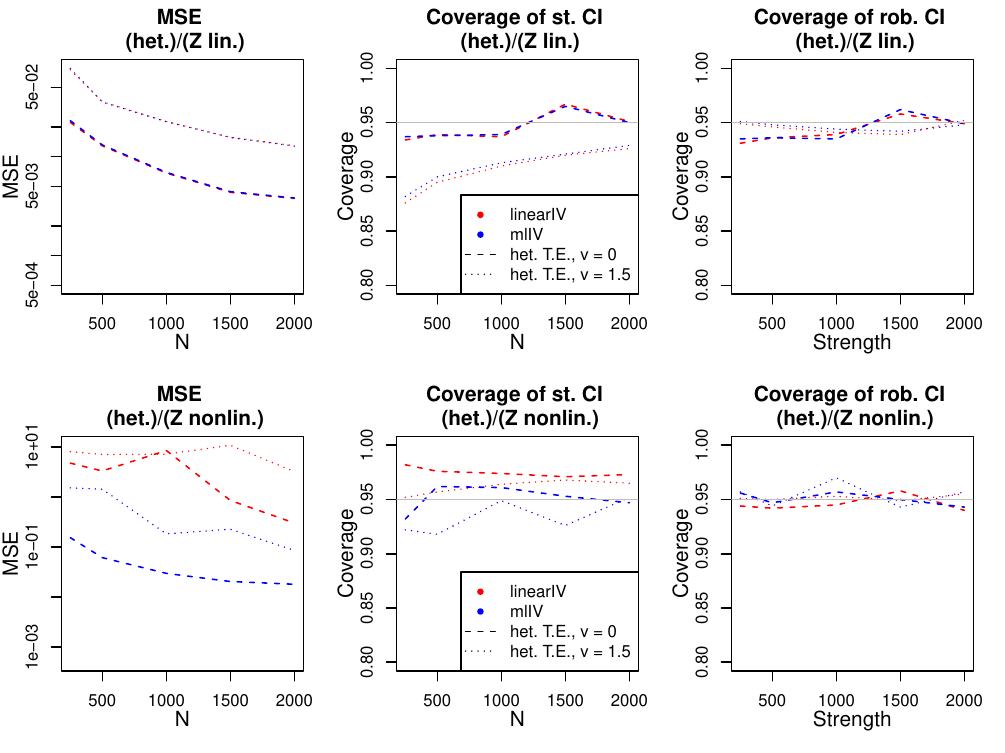}
\caption{Simulation results for setting (het.)/(Z lin.) (top) and (het.)/(Z nonlin.) (bottom). Left panel: mean squared error (MSE) of  $\hat\beta(0)$ (het. T.E. for $v = 0$) and $\hat\beta(1.5)$ (het. T.E. for $v = 1.5$). Middle panel: coverage of standard confidence intervals. Right panel: coverage of robust confidence sets. Methods based on (linearIV) are in red and methods based on (mlIV) are in blue. Dashed lines correspond to the methods (het. T.E. linearIV) and (het. T.E. mlIV) for $v = 0$. Dotted lines correspond to the methods (het. T.E. linearIV) and (het. T.E. mlIV) for $v = 1.5$.} 
\label{fig_SimHet1DStrong}
\end{figure}

{In Appendix \ref{sec_FurtherSim}, we present additional simulations: in Section \ref{sec_VaryIVStrength}, we vary the IV strength and observe that the standard confidence intervals may suffer from substantial undercoverage, whereas the robust confidence intervals still exhibit coverage close to the desired level of 0.95.
In Section \ref{sec_MultiDimSim}, we consider higher-dimensional covariates and nuisance function estimation using boosted regression trees instead of generalized additive models.}

\subsection{Real Data Experiments}\label{sec_RealData}
We apply our methods to two well-known instrumental variables datasets. Additionally to \texttt{gam} (generalized additive models using the \texttt{R} package \texttt{mgcv} \citep{WoodGLM1, WoodGLM2}), we also use boosted regression trees and random forest to estimate the nuisance functions. We use the \texttt{R} packages \texttt{xgboost} \citep{ChenXgboostPaper, ChenXgboostPackage} and \texttt{ranger} \citep{RangerPackage}. We refer to the three machine learning methods as \texttt{gam}, \texttt{xgboost} and \texttt{random forest}. In contrast to \texttt{gam}, for which we simply use the default settings, the methods \texttt{xgboost} and \texttt{random forest} need to be tuned. In view of the cross-fitting scheme, one would need to tune the machine learning methods for each of the $K$ different training datasets of size $N-N/K$ separately to avoid cross-contamination of the sample. However, recent research by \cite{BachHyperparameterTuning} showed that this is usually not an issue and it suffices to tune the machine learning methods once on the full sample. Hence, in the following, for each nuisance function that has to be estimated, we tune the machine learner once on the full sample to obtain the optimal hyperparameters (with respect to cross-validated prediction accuracy for \texttt{xgboost} and with respect to out-of-bag prediction error for \texttt{random forest}) and then keep them fixed for the $K$ subsamples in all $S$ repetitions of the aggregation procedure from Section \ref{sec_RepeatedCrossFitting}.

\subsubsection{Effect of Institutions on Economic Growth}
We revisit the question of the effect of institutions on economic performance. The dataset originates from \cite{AcemogluColonialOrigins} and is available from the \texttt{R}-package \texttt{hdm} \citep{ChernozhukovHDM}. In the following, we refer to it as the AJR dataset. It was previously analyzed in \cite{ChernozhukovDML} and \cite{EmmeneggerRegularizingDML}. The AJR dataset contains 64 country-level observations of the GDP, an index measuring protection against expropriation, settler mortality, and geographic information. In our notation, the response $Y$ is the logarithm of GDP per capita, $D$ is the expropriation protection index, the instrument $Z$ is the logarithm of settler mortality and $X$ contains the latitude and dummies for the continents (Africa, Asia, North America and South America). While the relationship between expropriation protection and GDP might be endogenous, it is argued that settler mortality is a valid instrument, since it is unlikely that settler mortality has an influence on today's GDP other than through institutions, especially when controlling for geographic information. We refer to \cite{AcemogluColonialOrigins, ChernozhukovDML} for a more detailed discussion of the dataset.

We first apply the methods (hom. T.E. linearIV) and (hom. T.E. mlIV) to the dataset. We use 5-fold cross-fitting with $S = 200$ repetitions in the procedure from Section \ref{sec_RepeatedCrossFitting} and for the nuisance function estimates, we use \texttt{gam}, \texttt{xgboost} and \texttt{random forest} as described in the previous subsection. The results can be found in Table \ref{tab_AJR}. Point estimates, standard deviations and standard confidence intervals were also considered in \cite{ChernozhukovDML, EmmeneggerRegularizingDML} and our results are consistent with their findings up to variability due to sampling and tuning of the machine learning methods (the sample size is only $N = 64$). The robust confidence set based on (linearIV) was also considered in Section 13.3 of \cite{ChernozhukovCausalML}. What is new here, are the results based on (mlIV). We see that in most of the cases, the use of (mlIV) compared to (linearIV) leads to a reduction of the standard deviation of the estimator, the length of the standard confidence interval and the size of the robust confidence sets. We also observe that the robust confidence sets can be disconnected. Recall from Section \ref{sec_RobCI} that the robust confidence set is defined as the sublevel set of a parabola, which can be disconnected if the parabola has positive curvature.
\begin{table}[t]
\begin{center}
\begin{tabular}{ r|l|l|l|l} 
  & $\hat\beta$ & $\hat\sigma/\sqrt{n}$ & standard CI & robust CI  \\ 
  \hline
 linearIV (gam) & 0.72 & 0.27 & $[0.20, 1.25]$ & $[0.29, 4.02]$ \\ 
 mlIV (gam) & 0.58 & 0.16 & $[0.27, 0.90]$ & $[0.28, 1.81]$\\ 
 \hline
 linearIV (xgboost) & 0.62 & 0.35 & $[-0.07, 1.31]$ & $(-\infty, \infty)$ \\ 
 mlIV (xgboost) & 1.05 & 0.35 & $[0.36, 1.75]$ & $(-\infty, -3.84]\cup [0.47, \infty)$\\ 
 \hline
 linearIV (random forest) & 0.84 & 0.43 & $[0.00, 1.68]$ & $(-\infty, -0.50]\cup [0.04, \infty)$\\ 
 MLIV (random forest) & 0.94 & 0.36 & $[0.24, 1.65]$ & $(-\infty, -3.90] \cup [0.25,\infty)$\\
\end{tabular}
\end{center}
\caption{Results for the methods (hom. T.E. linearIV) and (hom. T.E. mlIV) on the AJR dataset.}
\label{tab_AJR}
\end{table}

One might wonder if the homogeneous treatment effect is a realistic model assumption for the AJR dataset or if one should consider heterogeneity with respect to geography. We do not aim to give a comprehensive answer to this question, but for illustration purposes, we apply the methods (het. T.E. mlIV) and (het. T.E. linearIV), where we take the latitude as the variable $V_i$ with respect to which we want to quantify the heterogeneity.
In Figure \ref{fig_AJR}, we see the estimated heterogeneous effects with respect to latitude for the machine learning methods \texttt{gam} and \texttt{xgboost}. We omit the results for random forest as they look qualitatively similar. As for the initial example in Section \ref{sec_Simulations}, we use the Epanechnikov kernel and the bandwidth according to the normal reference rule as well as a smaller bandwidth to ensure undersmoothing. From the plots, we cannot infer that there is a significant heterogeneity with respect to latitude. In fact, the robust confidence sets are unbounded for most values of latitude for both (linearIV) and (mlIV), especially when using \texttt{xgboost} for the nuisance functions and the smaller bandwidth. Note that due to the large difference between the robust confidence sets and the standard confidence intervals, we should not trust the standard confidence intervals. Especially, the size of the standard confidence intervals seems to decrease for large latitude. It is important to note, that this is an artifact of the data. For example, the confidence interval with latitude around $0.7$ is only based on the two datapoints with latitude close to $0.7$. Hence, there is no substantive value in the size of those confidence intervals. We conclude that for the AJR dataset, the sample size of $N = 64$ is too small to make a sensible statement about a potential heterogeneous treatment effect.

\begin{figure}[t]
\centering
\includegraphics[width=0.8\textwidth]{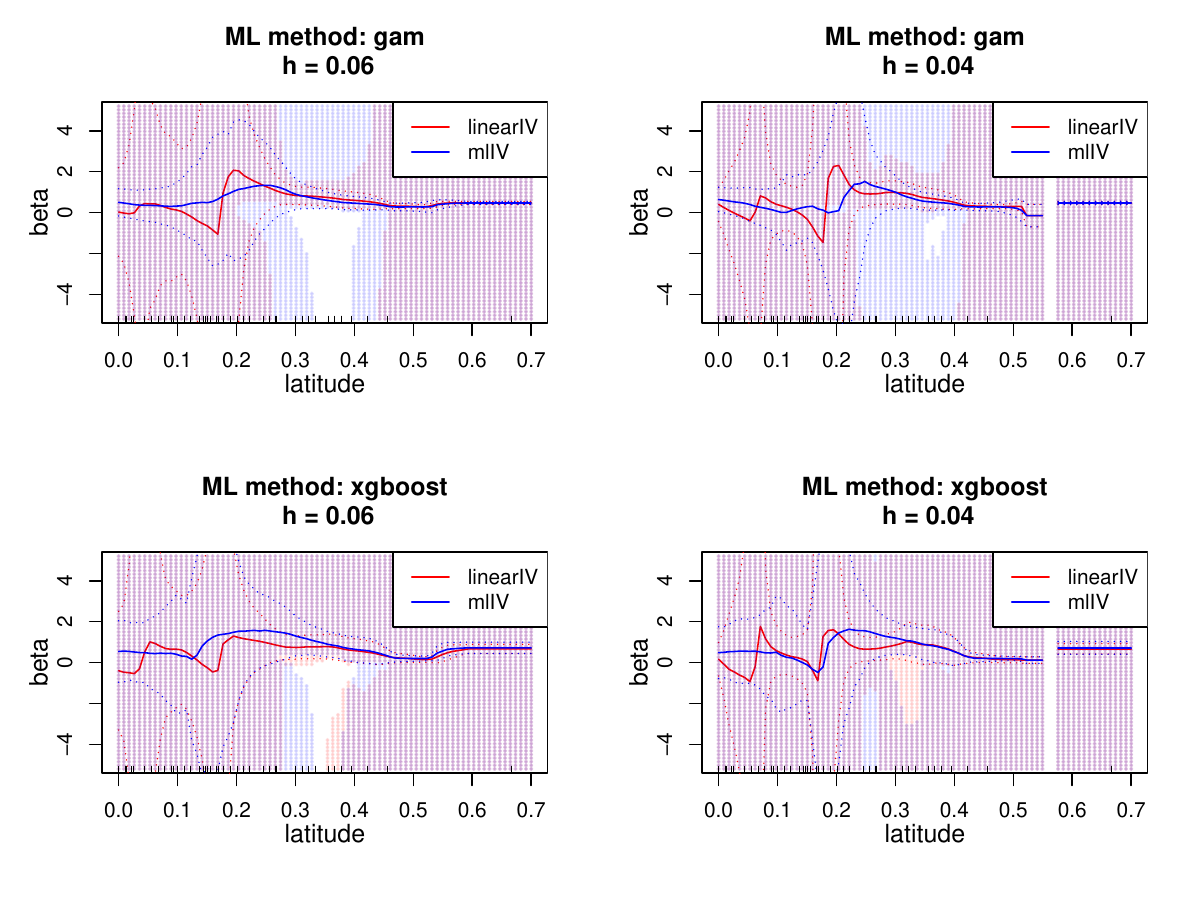}
\caption{Heterogeneous treatment effect for the AJR data. Top: nuisance functions are estimated using \texttt{gam}. Bottom: nuisance functions are estimated using \texttt{xgboost}. On the left, the bandwidth $h$ is chosen according to the normal reference rule / Silverman's rule of thumb. On the right, the bandwidth is chosen according to the normal reference rule multiplied by $N^{1/5}/N^{2/7}$ to ensure undersmoothing. Solid lines are the point estimates for the methods (het. T.E. linearIV) (red) and (het. T.E. mlIV) (blue). Dashed lines are the (pointwise) standard 95\%-confidence interval. Shaded regions are the (pointwise) robust 95\%-confidence sets. For the smaller bandwidth, there is a region around 0.57, for which no results can be estimated because no datapoints are close enough. Red-shaded regions correspond to (het. T.E. linearIV) and blue-shaded regions correspond to (het. T.E. mlIV). Purple-shaded regions are the intersection of the two.}
\label{fig_AJR}
\end{figure}

\subsubsection{Effect of Education on Wage}
We consider the famous dataset by \cite{CardCollege} on the effect of education on wage. The dataset is available from the \texttt{R}-package \texttt{ivmodel} \citep{KangIVmodel}. We refer to the dataset as the Card dataset. It contains data for $N = 3010$ individuals from the  National Longitudinal Survey of Young
Men. We follow the recent analysis by \cite{GuoTSCI} and use the same treatment, response, instrument and covariates as there. Concretely, in our notation, the response $Y_i$ is the logarithm of the wage of individual $i$, $D_i$ is the years of education of individual $i$, the instrument $Z_i$ is an indicator of whether the individual grew up near a four-year college and $X_i$ are covariates, including years of experience and binary indicators for race and geographic information. The idea for using college proximity as an instrument for education is that, while there are factors that influence both education and wage, they should be independent of whether an individual grew up near a college or not. As before, we first apply the methods (hom. T.E. linearIV) and (hom. T.E. mlIV) to the dataset. We use 5-fold cross-fitting with $S = 50$ repetitions for the procedure from Section \ref{sec_RepeatedCrossFitting}. For the nuisance function estimates, we use \texttt{gam}, \texttt{xgboost} and \texttt{random forest}. The results can be found in Table \ref{tab_Card}. In contrast to the AJR dataset, we see that the use of (mlIV) compared to (linearIV) does not lead to a reduction of the standard deviation of the estimator, the length of the standard confidence interval and the size of the robust confidence sets. It seems that on this dataset, there is nothing to gain from using (mlIV) and one loses a bit of performance (i.e. larger confidence intervals) due to the additional variability that is introduced in the procedure by having to learn the instruments instead of using them linearly.
\begin{table}[t]
\begin{center}
\begin{tabular}{ r|l|l|l|l} 
  & $\hat\beta$ & $\hat\sigma/\sqrt{n}$ & standard CI & robust CI  \\ 
  \hline
 linearIV (gam) & 0.14 & 0.05 & $[0.04, 0.24]$ & $[0.05, 0.28]$ \\ 
 mlIV (gam) & 0.14 & 0.05 & $[0.03, 0.24]$ & $[0.03, 0.28]$\\ 
 \hline
 linearIV (xgboost) & 0.12 & 0.04 & $[0.04, 0.21]$ & $[0.04, 0.23]$ \\ 
 mlIV (xgboost) & 0.11 & 0.05 & $[0.02, 0.20]$ & $[0.01, 0.23]$\\ 
 \hline
 linearIV (random forest) & 0.13 & 0.05 & $[0.04, 0.23]$ & $[0.05, 0.25]$\\ 
 MLIV (random forest) & 0.14 & 0.08 & $[-0.01, 0.3]$ & $[0.01, 0.60]$\\
\end{tabular}
\end{center}
\caption{Results for the methods (hom. T.E. linearIV) and (hom. T.E. mlIV) on the Card dataset.}
\label{tab_Card}
\end{table}

Also for the Card dataset, we apply the methods (het. T.E. mlIV) and (het. T.E. linearIV) to see if we can detect heterogeneity of the treatment effect. As the variable $V_i$, we take the years of experience of individual $i$. 
The estimated heterogeneous treatment effects with respect to experience are displayed in Figure \ref{fig_Card}. Again, we only show the results for the machine learning methods \texttt{gam} and \texttt{xgboost} and omit the results for \texttt{random forest} since they look similar. Moreover, we still use the Epanechnikov kernel and the bandwidth according to the normal reference rule as well as a smaller bandwidth to ensure undersmoothing.

Similarly to the AJR datasets, we cannot infer a significant heterogeneity from the plots. Moreover, we see that the methods based on (mlIV) have larger confidence intervals, in particular when the nuisance functions are estimated using \texttt{xgboost}. Interestingly, there are some regions for experience, where the value of the heterogeneous treatment effect seems to be quite narrowly determined. On the other hand, there are regions, where there is a lot of uncertainty in the value of the heterogeneous treatment effect (indicated by larger standard and robust confidence sets). Surprisingly, those effects do not only happen in low-density regions of experience (near the boundary) but also around 10 years of experience, where the density of experience is rather high. Unfortunately, it is hard to determine, if this is an artifact of the sample or indeed has a substantive meaning. 
\begin{figure}[t]
\centering
\includegraphics[width=0.8\textwidth]{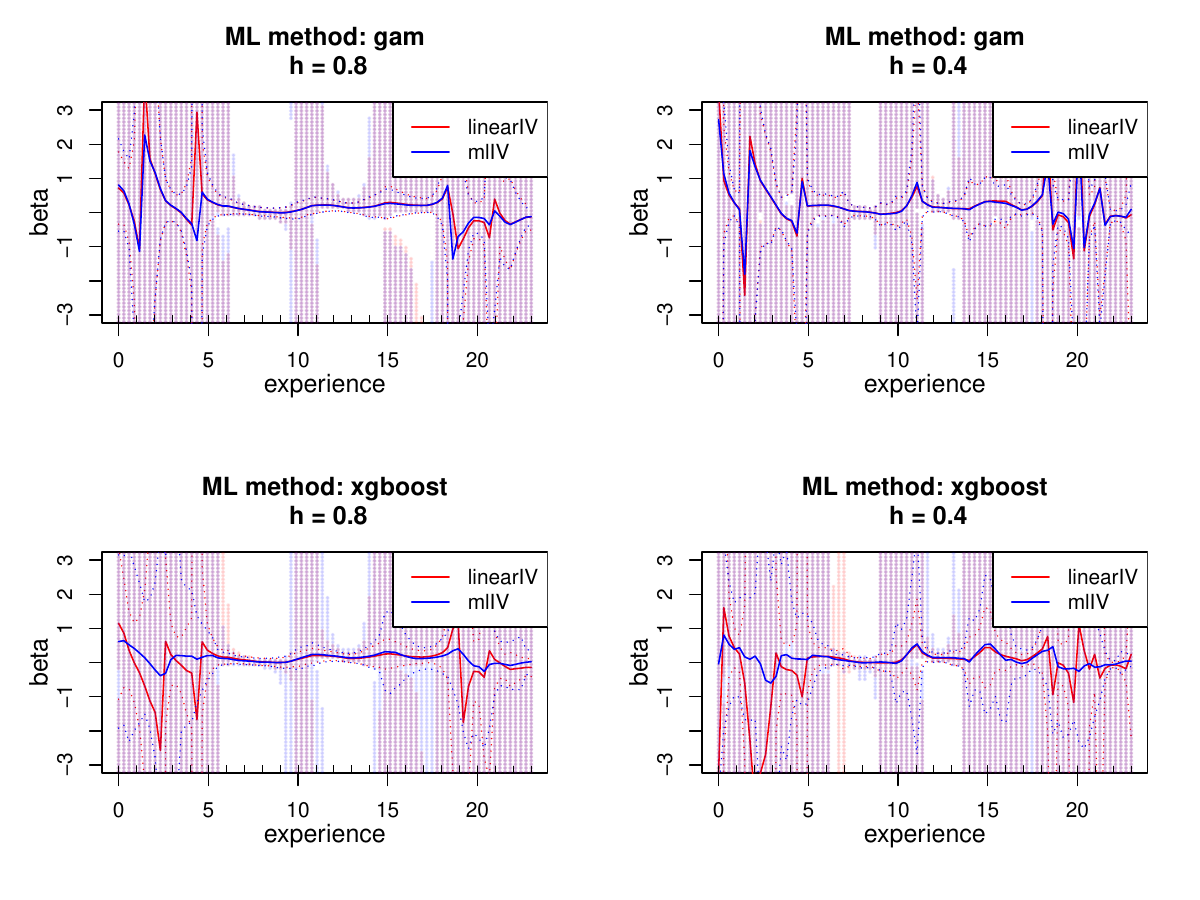}
\caption{Heterogeneous treatment effect for the Card data. Top: nuisance functions are estimated using \texttt{gam}. Bottom: nuisance functions are estimated using \texttt{xgboost}. On the left, the bandwidth $h$ is chosen according to the normal reference rule / Silverman's rule of thumb. On the right, the bandwidth is chosen according to the normal reference rule multiplied by $N^{1/5}/N^{2/7}$ to ensure undersmoothing. Solid lines are the point estimates for the methods (het. T.E. linearIV) (red) and (het. T.E. mlIV) (blue). Dashed lines are the (pointwise) standard 95\%-confidence interval. Shaded regions are the (pointwise) robust 95\%-confidence sets. Red-shaded regions correspond to (het. T.E. linearIV) and blue-shaded regions correspond to (het. T.E. mlIV). Purple-shaded regions are the intersection of the two.}
\label{fig_Card}
\end{figure}

\section{Discussion}\label{sec_Discussion}
We developed a novel method to perform inference on potentially heterogeneous treatment effects using double machine learning and efficient machine learning instruments. Our method is simple and yields results that are easy to interpret. We gave conditions under which our estimator is asymptotically Gaussian and also constructed confidence sets that are robust to a certain extent of weak instruments. Additionally, we hope that our paper is also helpful for readers who are looking for an accessible discussion of the DML estimator with machine learning instruments for the homogeneous treatment effect.

{An important avenue for further research are methods to decide from data if a constant treatment effect model is enough or if a heterogeneous treatment effect model is more appropriate. As an informal test, one can check if the homogeneous treatment effect estimate (i.e., a horizontal line) lies inside the confidence sets for the heterogeneous treatment effect. In fact, from the corresponding plots in our two real-data case studies, we observed that this is (more or less) the case which leads to the conclusion that the heterogeneous treatment effect model is probably not needed. However, the confidence sets that we provided are only pointwise. Hence,} it would also be interesting to have uniform confidence bands. We expect that this is possible using multiplier bootstrap, as for example done in \cite{FanEstimationOfCATE}.

{Moreover, it would be interesting to further investigate bandwidth selection techniques beyond the simple heuristics applied in Section \ref{sec_Applications} (although the heuristics seem to work sufficiently well). For standard nonparametric regression with kernel smoothing, the optimal bandwidth (with respect to mean squared error) can be obtained using cross-validation. In our setting, this is not so easy, although one can in principle find a criterion that is minimized by $\beta(\cdot)$, which could be used for cross-validation. However, the problem is that by construction, when keeping the sample fixed, there can be some combinations of $h$ and $v$ such that the denominator in the definition \eqref{eq_DefEstimatorHet} of $\hat\beta(v)$ is zero, so a simple cross validation based on a criterion involving $\hat\beta(\cdot)$ will run into problems. A related question is if one can choose the bandwidth in some sense adaptively to local smoothness or instrument strength.}

A limitation of our approach is that it only considers heterogeneity with respect to a univariate quantity, i.e. the dimension of $V_i$ is $1$. If the treatment effect is potentially heterogeneous with respect to the full set of covariates $X_i$, i.e. the true model is $Y_i = \beta(X_i)D_i + g(X_i) + \epsilon_i$, equation \eqref{eq_IdentificationHet} identifies $\E[\beta(X_i)w(Z_i, X_i, v)|V_i = v]$ with the weight function
$$w(Z_i, X_i, v) = \frac{(\E[D_i|Z_i, X_i]-\E[D_i|X_i])^2}{\E[(\E[D_i|Z_i, X_i]-\E[D_i|X_i])^2|V_i = v]}.$$
In particular, this is not just the projection $\E[\beta(X_i)|V_i = v]$ of the true heterogeneous treatment effect on $V_i$.

A straightforward extension of our approach would be to use multivariate kernel smoothing to deal with multivariate $V_i$ and the results should carry over. However, note that the conditions on the convergence rates of the nuisance function estimators (Assumption \ref{ass_NuisanceRatesHet}) depend on the bandwidth $h$. In higher dimensions, they will depend on $h^d$, where $h$ is the dimension of $h$ and it will quickly become unrealistic to have machine learning methods that achieve these convergence rates.

We can also take one step back. In principle, we can replace the outer conditional expectations in \eqref{eq_IdentificationHet} by arbitrary machine learning methods potentially also allowing for higher-dimensional $V_i$. A natural assumption might be that $V_i = X_i$ which is potentially high-dimensional, but that the treatment effect is only heterogeneous with respect to a small subset of the components of $X_i$, so one could use methods enforcing sparsity for the outer conditional expectation. If the machine learning method in the numerator allows for inference and the machine learning method in the denominator is consistent, one might still be able to perform inference for the heterogeneous treatment effect. 

{The aim of this work was to present a simple and ready-to-use method for inference in combination with user-chosen machine learning algorithms, but from a theoretical perspective, we did not strive for full generality. For the homogeneous case, estimators with better theoretical properties could probably be obtained using higher order influence functions \citep{RobinsHigherOrderInfluenceFunctions} or triple sample splitting \citep{NeweyCrossFittingAndFastRemainderRates}. For the heterogeneous treatment effect -- at least in the case of binary treatment and without endogeneity -- minimax optimality was discussed by \cite{KennedyTowardsOptimalDoublyRobustEstimation, KennedyMinimaxRates}. Whereas this was not the focus of our work, it would be interesting to obtain similar results for the endogenous setting and allowing for continuous treatment as considered here}.

We leave all these interesting directions to further research.

\section*{Acknowledgements}
We are grateful to Malte Londschien for helpful discussions about weak instrumental variables.
CS received funding from the Swiss National Science Foundation, grant no. 214865.
ZG acknowledges financial support for visiting the Institute of Mathematical Research (FIM) at ETH Zurich.

\begin{appendix}
\section{Technical Conditions for the Heterogeneous Treatment Effect}\label{app_TechAssHTE}
The following assumptions are needed to be able to estimate and perform inference for the heterogeneous treatment effect $\beta(v)$.
\begin{assumption}\label{ass_RegularityHet}
\mbox\newline
    \begin{enumerate}
        \item The density $p_V(\cdot)$ of $V_i$ with respect to Lebesgue measure satisfies that $p_V(v) \geq C$ for a constant $C>0$ independent of $N$. \label{ass_Density}
        \item $\sigma_Q^2(v)\asymp \E[(f(Z_i, X_i)-\phi(X_i))^2|V_i = v]$. \label{ass_CompSigmaQ}
        \item $\E\left[(D_i-\phi(X_i))^2(f(Z_i, X_i)-\phi(X_i))^2|V_i = v\right]\lesssim \sigma_Q^2(v)$. \label{ass_CompSigmaQ2}
        \item The function $\rho(x) = (\beta(x)-\beta(v))\E[(f(Z_i, X_i)-\phi(X_i))^2|V_i = x] p_V(x)$ is twice differentiable for $x$ in the support of $V_i$ and there exists $C<\infty$ independent of $N$ such that $\|\rho''\|_\infty/\sigma_Q(v) < C$.\label{ass_Reg1}
        \item The function $\rho(x) = \E[\epsilon_i^2(f(Z_i, X_i)-\phi(X_i))^2|V_i = x] p_V(x)$ is twice differentiable for $x$ in the support of $V_i$ and there exists $C<\infty$ independent of $N$ such that $\|\rho''\|_\infty/\sigma_Q^2(v) < C$.\label{ass_Reg2}
        \item The function $\rho(x) = (\beta(x)-\beta(v))\E[\epsilon_i(D_i-\phi(X_i))(f(Z_i, X_i)-\phi(X_i))^2|V_i = x] p_V(x)$ is twice differentiable for $x$ in the support of $V_i$ and there exists $C<\infty$ independent of $N$ such that $\|\rho''\|_\infty/\sigma_Q^2(v) < C$.\label{ass_Reg3}
        \item The function $\rho(x) = (\beta(x)-\beta(v))^2\E[(D_i-\phi(X_i))^2(f(Z_i, X_i)-\phi(X_i))^2|V_i = x] p_V(x)$ is twice differentiable for $x$ in the support of $V_i$ and there exists $C<\infty$ independent of $N$ such that $\|\rho''\|_\infty/\sigma_Q^2(v) < C$.\label{ass_Reg4}
        \item The function $\rho(x) = (\beta(x)-\beta(v))^2\E[(D_i-\phi(X_i))^2|V_i = x] p_V(x)$ is twice differentiable for $x$ in the support of $V_i$ and there exists $C<\infty$ independent of $N$ such that $\|\rho''\|_\infty < C$.\label{ass_Reg5}
        \item The function $\rho(x) = \E[(f(Z_i, X_i)-\phi(X_i))^2|V_i = x] p_V(x)$ is twice differentiable for $x$ in the support of $V_i$ and there exists $C<\infty$ independent of $N$ such that $\|\rho''\|_\infty/\sigma_Q^2(v) < C$.\label{ass_Reg6}
        \item The function $\rho(x) = \E[(D_i-\phi(X_i))^2(f(Z_i, X_i)-\phi(X_i))^2|V_i = x] p_V(x)$ is twice differentiable for $x$ in the support of $V_i$ and there exists $C<\infty$ independent of $N$ such that $\|\rho''\|_\infty/\sigma_Q^2(v) < C$.\label{ass_Reg7}
        \item The function $\rho(x) = \E[(D_i-\phi(X_i))^2|V_i = x] p_V(x)$ is twice differentiable for $x$ in the support of $V_i$ and there exists $C<\infty$ independent of $N$ such that $\|\rho''\|_\infty < C$.\label{ass_Reg8}
        \item There exists $\eta>0$ and $C>0$ independent of $N$ such that
        $$\textstyle \frac{\E\left[\left|\left(\epsilon_i + (\beta(V_i)-\beta(v))(D_i-\phi(X_i))\right)\left(f(Z_i, X_i)-\phi(X_i)\right) K\left(\frac{V_i - v}{h}\right) \right|^{2+\eta}\right]}{\sigma_Q(v)^{2+\eta}} \leq Ch.$$\label{ass_Lindeberg}
        \item 
        $$\textstyle\frac{\E\left[\left(\epsilon_i + (\beta(V_i)-\beta(v))(D_i-\phi(X_i))\right)^4\left(f(Z_i, X_i)-\phi(X_i)\right)^4K\left(\frac{V_i - v}{h}\right)^4\right]}{\sigma_Q^4(v) }\ll N h^2.$$\label{ass_ChebyVar}
        \item $$\textstyle \E\left[(D_i-\phi(X_i))^4K\left(\frac{V_i - v}{h}\right)^2\right]= O(h).$$\label{ass_RD4}
    \end{enumerate}
\end{assumption}

\begin{remark}
    Note that all the quantities in Assumption \ref{ass_RegularityHet} may depend on $N$ even if it is not reflected in the notation. Assertion \ref{ass_Density} intuitively means that we need to have enough samples $i$ with $V_i$ close to $v$ to be able to estimate the heterogeneous treatment effect. Assertions \ref{ass_CompSigmaQ} and \ref{ass_CompSigmaQ2} are needed in the case where the IV strength $\sigma_Q^2(v)$ is decreasing. Essentially they state that a decreasing $\sigma_Q^2(v)$ is attributed to $f(Z_i, X_i)$ being close to $\phi(X_i)$. They are for example satisfied if $\E[\epsilon_i^2|Z_i, X_i]$ and $\E[(D_i-\phi(X_i))^2|Z_i, X_i]$ are both bounded away from zero and infinity independent of $N$. Assertions \ref{ass_Reg1}-\ref{ass_Reg8} are smoothness assumptions on $\beta(x)$, $p_V(x)$ and conditional moments of $\epsilon_i$, $D_i$, $f(Z_i, X_i)$ and $\phi(X_i)$. These smoothness conditions need to hold independently of $N$. In view of Lemma \ref{lem_ExpFunKernel} in Appendix \ref{sec_SomeLemmas} (with the kernel $K$ potentially replaced by $K^{2+\eta}$ or $K^4$), assertions \ref{ass_Lindeberg}-\ref{ass_RD4} are comparably weak and hold in particular, if the quantities inside the expectations satisfy similar smoothness conditions as assertions \ref{ass_Reg1}-\ref{ass_Reg8}.
\end{remark}

\allowdisplaybreaks[1]
\section{Proof of Proposition \ref{pro_CompVar}}\label{app_ProofEfficiency}
Applying the tower property for conditional expectations to \eqref{eq_HomDefSigma}, we have that
$$\sigma^2 = \frac{\E\left[\epsilon_i^2(\E[D_i|Z_i,X_i]-\E[D_i|X_i])^2\right]}{\E\left[(\E[D_i|Z_i,X_i]-\E[D_i|X_i])^2\right]^2}.$$
Since by assumption, $\E[\epsilon_i^2|Z_i, X_i] = \E[\epsilon_i^2]$, we can simplify
\begin{align*}
    \sigma^2  &= \frac{\E[\epsilon_i^2]}{\E\left[(\E[D_i|Z_i, X_i]-\E[D_i|X_i])^2\right]},\\
    \sigma^2_{\zeta} &=\frac{\E\left[\epsilon_i^2\right]\E\left[(\zeta(Z_i, X_i)-\E[\zeta(Z_i, X_i)|X_i])^2\right]}{\E\left[(\E[D_i|Z_i, X_i]-\E[D_i|X_i])(\zeta(Z_i, X_i)-\E[\zeta(Z_i, X_i)|X_i])\right]^2}.
\end{align*}
For the denominator of $\sigma^2_{\zeta}$, we have by the Cauchy-Schwarz inequality
\begin{align*}
     &\E\left[(\E[D_i|Z_i, X_i]-\E[D_i|X_i])(\zeta(Z_i, X_i)-\E[\zeta(Z_i, X_i)|X_i])\right]^2\\
    &\leq\E\left[(\E[D_i|Z_i, X_i]-\E[D_i|X_i])^2\right]\E\left[(\zeta(Z_i, X_i)-\E[\zeta(Z_i, X_i)|X_i])^2\right].
\end{align*}
It follows that $\sigma^2\leq \sigma^2_{\zeta}$ with equality if and only if there exists $\alpha\neq 0$ such that $\E[D_i|Z_i, X_i]-\E[D_i|X_i]=\alpha (\zeta(Z_i, X_i)-\E[\zeta(Z_i, X_i)|X_i])$ a.s. 
Now, observe that if $\E[D_i|Z_i, X_i]-\E[D_i|X_i]=\alpha (\zeta(Z_i, X_i)-\E[\zeta(Z_i, X_i)|X_i])$, 
\eqref{eq_CondVarEqual} holds with $\psi(X_i) = \E[D_i|X_i]-\alpha\E[\zeta(Z_i, X_i)|X_i]$. If on the other hand \eqref{eq_CondVarEqual} holds, we have
\begin{align*}
    \E[D_i|Z_i, X_i]-\E[D_i|X_i] &= \alpha \zeta(Z_i, X_i) +\psi(X_i)-\alpha\E[\zeta(Z_i, X_i)|X_i]-\psi(X_i)\\
    &= \alpha(\zeta(Z_i, X_i)-\E[\zeta(Z_i, X_i)|X_i]).
\end{align*}
This concludes the proof.

\section{More on the Robust Confidence Set}
\subsection{Explicit Expression for the Robust Confidence Set}\label{app_ExplicitRobCI}
Using basic manipulations, we can derive explicit expressions for $\mathcal C_N^\alpha(v)$. For this, define the quantities
\begin{align*}
    A_N &= \frac{1}{Nh}\sum_{k=1}^K\sum_{i\in I_k} R_{Y, i}^k R_{f, i}^k K\left(\frac{V_i - v}{h}\right)\\
    B_N &= \frac{1}{Nh}\sum_{k=1}^K\sum_{i\in I_k} R_{D, i}^k R_{f, i}^k K\left(\frac{V_i - v}{h}\right)\\
    C_N & = \frac{1}{Nh}\sum_{k=1}^K\sum_{i\in I_k}\left(R_{Y, i}^k\right)^2\left(R_{f, i}^k\right)^2K\left(\frac{V_i - v}{h}\right)^2\\
    E_N & = \frac{2}{Nh}\sum_{k=1}^K\sum_{i\in I_k}R_{Y, i}^kR_{D, i}^k\left(R_{f, i}^k\right)^2K\left(\frac{V_i - v}{h}\right)^2\\
    F_N & = \frac{1}{Nh}\sum_{k=1}^K\sum_{i\in I_k}\left(R_{D, i}^k\right)^2\left(R_{f, i}^k\right)^2K\left(\frac{V_i - v}{h}\right)^2.
\end{align*}
Then, we have that
$$\textstyle \mathcal C_N^\alpha(v) = \left\{\gamma : (A_N - \gamma B_N)^2\leq \frac{z_{1-\alpha/2}^2}{Nh}\left(C_N -\gamma E_N + \gamma^2 F_N-h(A_N-\gamma B_N)^2\right)\right\}$$
or equivalently
$\mathcal C_N^\alpha(v) = \left\{\gamma : R_N\gamma^2+S_N\gamma +T_N\leq 0\right\}$
with
\begin{align*}
    R_N &= B_N^2 + \frac{z_{1-\alpha/2}^2}{Nh}(hB_N^2-F_N)\\
    S_N &= -2A_N B_N + \frac{z_{1-\alpha/2}^2}{Nh}(E_N-2hA_NB_N)\\
    T_N &= A_N^2 + \frac{z_{1-\alpha/2}^2}{Nh}(hA_N^2-C_N)
\end{align*}

Hence, we have the following explicit expression

\begin{equation}\label{eq_ExpressionRobCI}
    \mathcal C_N^\alpha(v) = \begin{cases}
\left[\frac{-S_N - \sqrt{S_N^2 - 4 R_N T_N}}{2R_N}, \frac{-S_N + \sqrt{S_N^2 - 4 R_N T_N}}{2R_N} \right], & R_N > 0 \\
\mathbb R \setminus \left[\frac{-S_N + \sqrt{S_N^2 - 4 R_N T_N}}{2R_N}, \frac{-S_N - \sqrt{S_N^2 - 4 R_N T_N}}{2R_N} \right], & R_N < 0.
\end{cases}
\end{equation}\
By construction, we have that $\hat\beta(v) = A_N/B_N\in \mathcal C_N^\alpha(v)$. Hence, in the case $R_N >0$, it will always hold that $S_N^2- 4R_N T_N >0$. In the case $R_N < 0$, it may happen that $S_N^2-4R_N T_N < 0$, in which case, $\mathcal C_N^\alpha(v) = \mathbb R$.

\subsection{Adaptivity of the Robust Confidence Set}\label{app_AdaRobCI}
\begin{proposition}\label{pro_AdaCI}
    Assume that there exist $A_0, B_0, C_0, E_0, F_0\in \mathbb R$ with $B_0\neq 0$ such that $A_N = A_0 + o_P(1)$, $B_N = B_0 + o_P(1)$, $C_N = C_0 + o_P(1)$, $E_N = E_0 + o_P(1)$ and $F_N = F_0 + o_P(1)$. Then, the confidence set $\mathcal C_N^\alpha(v)$ is asymptotically equivalent to the standard confidence interval, i.e. for fixed $\alpha \in (0, 1)$,
    $$\textstyle C_N^\alpha(v) = \left[\hat\beta(v) - \frac{z_{1-\alpha/2}\hat\sigma(v)}{\sqrt{Nh}} + o_P\left(\frac{1}{\sqrt{Nh}}\right),\,\hat\beta(v) + \frac{z_{1-\alpha/2}\hat\sigma(v)}{\sqrt {Nh}} + o_P\left(\frac{1}{\sqrt{Nh}}\right)\right]$$
\end{proposition}
\begin{remark}
    If the data generating process is fixed, then under the assumptions of Theorem \ref{thm_AsNormHet}, one can easily show that the conditions of Proposition \ref{pro_AdaCI} are satisfied by using the same arguments as in the proof of Theorem \ref{thm_AsNormHet}.
\end{remark}
\begin{proof}[Proof of Proposition \ref{pro_AdaCI}]
We use expression \eqref{eq_ExpressionRobCI}. Since $B_N = B_0 + o_P(1)$, $F_N = F_0 + o_P(1)$ and $B_0\neq 0$,  we have that $\Prob(R_N > 0)\to 1$ and we only need to consider the case $R_N > 0$ in \eqref{eq_ExpressionRobCI}. Hence, we can write
$$\textstyle \mathcal C_N^\alpha(v) = \left[\frac{A_N B_N \pm \frac{z_{1-\alpha/2}}{\sqrt {Nh}}\sqrt{B_N^2 C_N + A_N^2 F_N - A_N B_N E_N + O_P\left(\frac{1}{Nh}\right)}+ O_P\left(\frac{1}{Nh}\right)}{B_N^2 + O_P\left(\frac{1}{Nh}\right)}\right].$$
Because $B_N^2 = B_0^2 + o_P(1)$ and using that $|\sqrt{a + b}-\sqrt a|\leq \sqrt{|b|}$ for $a>0$ and $b \geq -a$ it follows that 
\begin{align*}
\textstyle 
\mathcal C_N^\alpha(v)&= \textstyle \left[\frac{A_N}{B_N}\pm \frac{z_{1-\alpha/2}}{\sqrt {Nh}}\sqrt{\frac{C_N}{B_N^2}+\frac{A_N^2F_N}{B_N^4} -\frac{A_NE_N}{B_N^3} + O_P\left(\frac{1}{Nh}\right)} + O_P\left(\frac{1}{Nh}\right) \right]\\
& = \textstyle \left[\hat\beta(v)\pm \frac{z_{1-\alpha/2}}{\sqrt {Nh}}\sqrt{\frac{C_N}{B_N^2}+\hat\beta(v)^2 \frac{F_N}{B_N^2} -\hat\beta(v) \frac{E_N}{B_N^2}} + O_P(1/ (Nh)) \right]
\end{align*}
where by \eqref{eq_DefEstimatorHet}, $\hat\beta(v) = A_N/B_N$. Note that
$$\textstyle \frac{C_N}{B_N^2}+\hat\beta(v)^2 \frac{F_N}{B_N^2} -\hat\beta(v) \frac{E_N}{B_N^2}=\frac{\frac{1}{Nh}\sum_{k=1}^K\sum_{i\in I_k}(R_{Y,i}^k-\hat \beta(v)R_{D, i}^k)^2(R_{f, i}^k)^2K\left(\frac{V_i - v}{h}\right)^2}{\left(\frac{1}{Nh}\sum_{k=1}^K\sum_{i \in I_k}R_{D, i}^kR_{f, i}^kK\left(\frac{V_i - v}{h}\right)\right)^2},$$
which is equal to $\hat\sigma^2(v)$ defined in \eqref{eq_DefVarianceHet}.
Hence, the robust and the non-robust confidence intervals are asymptotically equivalent.
\end{proof}

\section{Proofs for Section \ref{sec_HetTE}}
The proofs use similar ideas as \cite{ChernozhukovDML} combined with standard ideas from univariate kernel smoothing.
\subsection{Proof of Theorem \ref{thm_RobNumHet}}
The strategy is to replace the estimated nuisance functions $\hat f^k, \hat\phi_1^k, \hat\phi_2^k, \hat l^k$ by the true functions $f, \phi, l$ in $\hat Q(\beta_0(v), v)$. In fact, the theorem follows from the following two lemmas.
\begin{lemma}\label{lem_RobNumHetOracle}
    Under the conditions of Theorem \ref{thm_RobNumHet},
    \begin{align*}
        &\textstyle \frac{1}{\sqrt{Nh}\sigma_Q(v)}\sum_{i=1}^N(Y_i - l(X_i)-\beta_0(v)(D_i-\phi(X_i)))(f(Z_i, X_i)-\phi(X_i))K\left(\frac{V_i - v}{h}\right)\\
        &\to \mathcal N(0,1), \text{ in distribution.}
    \end{align*}
\end{lemma}
\begin{lemma}\label{lem_RobNumHetDiff}
    Under the conditions of Theorem \ref{thm_RobNumHet},
    \begin{align*}
        &\textstyle \frac{1}{\sqrt{Nh}\sigma_Q(v)}\sum_{i=1}^N(Y_i - l(X_i)-\beta_0(v)(D_i-\phi(X_i)))(f(Z_i, X_i)-\phi(X_i))K\left(\frac{V_i - v}{h}\right)\\
        &\textstyle - \frac{\sqrt{Nh}\hat Q(\beta_0(v), v)}{\sigma_Q(v)} = o_P(1). 
    \end{align*}
\end{lemma}
\subsubsection{Proof of Lemma \ref{lem_RobNumHetOracle}}
Under $H_0$, we have $\beta(v) = \beta_0(v)$ and hence $Y_i - l(X_i) -\beta_0(v)(D_i-\phi(X_i)) = \epsilon_i + (\beta(V_i)-\beta(v))(D_i-\phi(X_i))$. Hence, we need to consider the i.i.d. random variables
$$\textstyle W_i = \frac{\left(\epsilon_i + (\beta(V_i)-\beta(v))(D_i-\phi(X_i))\right)\left(f(Z_i, X_i)-\phi(X_i)\right) K\left(\frac{V_i - v}{h}\right)}{\sqrt{h}\sigma_Q(v)}.$$
By Assumption \ref{ass_Het}, $\E[\epsilon_i|Z_i, X_i] = 0$. We can use the tower property of conditional expectations together with the fact that $V_i$ is measurable with respect to $X_i$, $\E[D_i|Z_i, X_i] = f(Z_i, X_i)$ and Lemma \ref{lem_ExpFunKernel} from Section \ref{sec_SomeLemmas} below to arrive at 
\begin{align*}
\E[W_i]&=\frac{1}{\sqrt{h}\sigma_Q(v)}\E\left[(\beta(V_i)-\beta(v))\E[(f(Z_i, X_i)-\phi(X_i))^2| V_i]K\left(\frac{V_i - v}{h}\right)\right]\\
&= h^{2.5} \frac{\rho''(v)}{2\sigma_Q(v)}\int_{-\infty}^\infty K(x)x^2\mathrm{d}x + \frac{R(v, h)}{\sqrt h \sigma_Q(v)}.
\end{align*}
with $\rho(x) = (\beta(x)-\beta(v))\E[(f(Z_i, X_i)-\phi(X_i))^2| V_i = x]p_V(x)$ and $R(v, h)$ defined accordingly in Lemma \ref{lem_ExpFunKernel}. By assertion \ref{ass_Reg1} of Assumption \ref{ass_RegularityHet}, there exists a constant $C\in \mathbb R$ independent of $N$ such that
\begin{equation}\label{eq_ExpWiHet}
    E[W_i] \leq C h^{5/2}.
\end{equation}
Similarly, from the second assertion of Lemma \ref{lem_ExpFunKernel} together with assertions \ref{ass_Reg2}, \ref{ass_Reg3} and \ref{ass_Reg4} of Assumption \ref{ass_RegularityHet}
\begin{align}
    &\textstyle \E[W_i^2]= \textstyle \frac{1}{h\sigma_Q^2(v)}\E\left[\epsilon_i^2(f(Z_i, X_i)-\phi(X_i))^2K\left(\frac{V_i - v}{h}\right)^2\right]\nonumber\\
    &\quad \textstyle+\frac{2}{h\sigma_Q^2(v)}\E\left[\epsilon_i(\beta(V_i)-\beta(v))(D_i-\phi(X_i))(f(Z_i, X_i)-\phi(X_i))^2K\left(\frac{V_i - v}{h}\right)^2\right]\nonumber\\
    &\quad \textstyle +\frac{1}{h\sigma_Q^2(v)}\E\left[(\beta(V_i)-\beta(v))^2(D_i-\phi(X_i))^2(f(Z_i, X_i)-\phi(X_i))^2K\left(\frac{V_i - v}{h}\right)^2\right]\nonumber\\
    &\textstyle = \frac{1}{\sigma_Q^2(v)}\E\left[\epsilon_i^2(f(Z_i, X_i)-\phi(X_i))^2|V_i = v\right]p_V(v)\int_{-\infty}^\infty K(x)^2\mathrm{d}x + O(h^2)\nonumber\\
    & \textstyle = 1 + O(h^2)\label{eq_ExpWiSquared}
\end{align}
by the definition of $\sigma_Q^2(v)$.
The quantity from Lemma \ref{lem_RobNumHetOracle} is $\frac{1}{\sqrt N}\sum_{i=1}^N W_i$. We apply the Lindeberg-Feller Theorem (see for example \cite{DurrettProb}, Thm. 3.4.10.) to the random variables $(W_i-\E[W_i])/\sqrt{N}$. For this, note that
$$\textstyle \sum_{i=1}^N \E\left[(W_i-\E[W_i])^2\right]/N = \E[W_i^2]-\E[W_i]^2 = 1+O(h^2).$$
Moreover, for all $\delta > 0$, using the inequality $|a + b|^{\gamma}\leq 2^{\gamma}(|a|^\gamma + |b|^\gamma)$ for $\gamma > 0$ and Jensen's inequality
\begin{align*}
    &\textstyle\sum_{i = 1}^N \E\left[\frac{(W_i-\E[W_i])^2}{N}\mathbf{1}\{|W_i - \E[W_i]|> \sqrt{N}\delta\}\right]\\
    &\textstyle = \E\left[{(W_i-\E[W_i])^2}\mathbf{1}\{|W_i - \E[W_i]|> \sqrt{N}\delta\}\right]\\
    & \textstyle \leq \frac{\E\left[|W_i - \E[W_i]|^{2+\eta}\right]}{\delta^\eta N^{\eta/2}}\\
    &\textstyle \leq 2^{2+\eta}\frac{\E\left[|W_i|^{2 + \eta}\right] + |\E[W_i]|^{2+\eta}}{\delta^\eta N^{\eta/2}}\\
    &\textstyle \leq 2^{3+\eta}\frac{\E\left[|W_i|^{2 + \eta}\right]}{\delta^\eta N^{\eta/2}}\\
    &\textstyle = 2^{3+\eta}\frac{\E\left[\left|\left(\epsilon_i + (\beta(V_i)-\beta(v))(D_i-\phi(X_i))\right)\left(f(Z_i, X_i)-\phi(X_i)\right) K\left(\frac{V_i - v}{h}\right) \right|^{2+\eta}\right]}{{\delta^{\eta}N^{\eta/2} h^{\eta/2}\sigma_Q(v)^{2+\eta}}\cdot h}.
\end{align*}
This converges to $0$ for $N\to\infty$, because $Nh\to\infty$ and because of assertion \ref{ass_Lindeberg} of Assumption \ref{ass_RegularityHet}. From the Lindeberg-Feller Theorem, we can conclude that 
$$\textstyle \frac{1}{\sqrt{N}} \sum_{i=1}^N (W_i-\E[W_i]) = \frac{1}{\sqrt{N}}\sum_{i=1}^N W_i - \sqrt{N}\E[W_i]\to\mathcal N(0,1)$$
in distribution. By \eqref{eq_ExpWiHet} and  Assumption \ref{ass_Bandwidth}, it follows that also $\sum_{i=1}^N W_i/\sqrt{N}\to \mathcal N(0,1),$ which concludes to Proof of Lemma \ref{lem_RobNumHetOracle}.

\subsubsection{Proof of Lemma \ref{lem_RobNumHetDiff}}
Define
\begin{alignat*}{2}
\Delta l_i^k      &= l(X_i)-\hat l^k(X_i),       &\quad \Delta f_i^k      &= f(Z_i, X_i)-\hat f^k(Z_i, X_i),\\
\Delta \phi_{1,i}^k &= \phi(X_i)-\hat\phi_1^k(X_i), &\quad \Delta \phi_{2,i}^k &= \phi(X_i)-\hat\phi_2^k(X_i)
\end{alignat*}
Under $H_0: \beta(v) = \beta_0(v)$ we can write (using that $Y_i-l(X_i)-\beta(v)(D_i-\phi(X_i)) = \epsilon_i + (\beta(V_i)-\beta(v))(D_i-\phi(X_i))$ in the last step)
\begin{align}
    &\textstyle\frac{\sqrt{Nh}\hat Q(\beta_0(v), v)}{\sigma_Q(v)} \nonumber\\
    &= \textstyle \frac{1}{\sqrt{Nh}\sigma_Q(v)}\sum_{k = 1}^K\sum_{i\in I_k} (Y_i-\hat l^k(X_i)- \beta(v) (D_i - \hat\phi_1^k(X_i)))\cdot\nonumber\\
    &\textstyle \qquad\qquad\qquad\qquad\qquad\quad\cdot(\hat f^k(Z_i, X_i) - \hat \phi_2^k(X_i))K\left(\frac{V_i - v}{h}\right)\nonumber\\
    &=\textstyle  \frac{1}{\sqrt{Nh}\sigma_Q(v)}\sum_{k=1}^K\sum_{i\in I_k}\left(Y_i-l(X_i)-\beta(v)(D_i-\phi(X_i))+\Delta l_i^k-\beta(v)\Delta\phi_{1,i}^k\right)\cdot\nonumber\\
    &\textstyle \qquad\qquad\qquad\qquad\qquad\quad\cdot\left(f(Z_i, X_i)-\phi(X_i)-\Delta f_i^k+\Delta\phi_{2,i}^k\right)K\left(\frac{V_i - v}{h}\right)\nonumber \\
    &\textstyle = \frac{1}{\sqrt{Nh}\sigma_Q(v)}\sum_{i=1}^N\left(Y_i-l(X_i)-\beta(v)(D_i-\phi(X_i))\right)\cdot\nonumber\\
    &\textstyle \quad\qquad\qquad\qquad\quad\cdot(f(Z_i, X_i) - \phi(X_i))K\left(\frac{V_i - v}{h}\right)\nonumber\\
    &\textstyle + \frac{1}{\sqrt{Nh}\sigma_Q(v)}\sum_{k=1}^K\sum_{i\in I_k}\Bigl[\left(\epsilon_i + (\beta(V_i)-\beta(v))(D_i-\phi(X_i))\right)\left(\Delta\phi_{2,i}^k-\Delta f_i^k\right)\nonumber\\
    &\textstyle\qquad\qquad\qquad\qquad\qquad\quad+\left(\Delta l_i^k -\beta(v) \Delta \phi_{1,i}^k\right)\left(f(Z_i,X_i)-\phi(X_i)\right)\nonumber\\
    &\textstyle\qquad\qquad\qquad\qquad\qquad\quad+\left(\Delta l_i^k -\beta(v) \Delta \phi_{1,i}^k\right)\left(\Delta\phi_{2,i}^k-\Delta f_i^k\right)\Bigr]K\left(\frac{V_i - v}{h}\right).  \label{eq_DecompWi}
\end{align}
Hence, to arrive at Lemma \ref{lem_RobNumHetDiff}, we need to show that for $k = 1,\ldots, K$, 
\begin{align}
    \frac{1}{\sqrt{n h}\sigma_Q(v)}\sum_{i\in I_k}\epsilon_i \Delta f_i^k K\left(\frac{V_i - v}{h}\right) & = o_P(1)\label{eq_RobNumHet1}\\
    \frac{1}{\sqrt{n h}\sigma_Q(v)}\sum_{i\in I_k}\epsilon_i \Delta \phi_{2,i}^k K\left(\frac{V_i - v}{h}\right) & = o_P(1)\label{eq_RobNumHet2}\\
    \frac{1}{\sqrt{n h}\sigma_Q(v)}\sum_{i\in I_k}(\beta(V_i)-\beta(v))(D_i-\phi(X_i)) \Delta f_i^k K\left(\frac{V_i - v}{h}\right) & = o_P(1)\label{eq_RobNumHet3}\\
    \frac{1}{\sqrt{n h}\sigma_Q(v)}\sum_{i\in I_k}(\beta(V_i)-\beta(v))(D_i-\phi(X_i)) \Delta \phi_{i,2}^k K\left(\frac{V_i - v}{h}\right) & = o_P(1)\label{eq_RobNumHet4}\\
    \frac{1}{\sqrt{n h}\sigma_Q(v)}\sum_{i\in I_k}(f(Z_i, X_i)-\phi(X_i)) \Delta l_i^k K\left(\frac{V_i - v}{h}\right) & = o_P(1)\label{eq_RobNumHet5}\\
    \frac{1}{\sqrt{n h}\sigma_Q(v)}\sum_{i\in I_k}(f(Z_i, X_i)-\phi(X_i)) \Delta \phi_{i,1}^k K\left(\frac{V_i - v}{h}\right) & = o_P(1)\label{eq_RobNumHet6}\\
    \frac{1}{\sqrt{n h}\sigma_Q(v)}\sum_{i\in I_k} \Delta l_i^k \Delta \phi_{2,i}^k K\left(\frac{V_i - v}{h}\right) & = o_P(1)\label{eq_RobNumHet7}\\
    \frac{1}{\sqrt{n h}\sigma_Q(v)}\sum_{i\in I_k} \Delta l_i^k \Delta f_i^k K\left(\frac{V_i - v}{h}\right) & = o_P(1)\label{eq_RobNumHet8}\\
    \frac{1}{\sqrt{n h}\sigma_Q(v)}\sum_{i\in I_k} \Delta \phi_{1,i}^k\Delta \phi_{2,i}^k K\left(\frac{V_i - v}{h}\right) & = o_P(1)\label{eq_RobNumHet9}\\
    \frac{1}{\sqrt{n h}\sigma_Q(v)}\sum_{i\in I_k} \Delta \phi_{1,i}^k \Delta f_i^k K\left(\frac{V_i - v}{h}\right) & = o_P(1).\label{eq_RobNumHet10}
\end{align}
Recall from Section \ref{sec_SummaryHTE} that by conditioning on $I_k^c$, we mean conditioning on $\{X_i, Y_i, Z_i, D_i\}_{i\in I_k^c}$. For \eqref{eq_RobNumHet1}, observe that
\begin{align}
    &\textstyle\E\left[\left(\frac{1}{\sqrt{n h}\sigma_Q(v)}\sum_{i\in I_k}\epsilon_i \Delta f_i^k K\left(\frac{V_i - v}{h}\right)\right)^2| I_k^c, X_{I_k}, Z_{I_k}\right]\nonumber\\
    &\textstyle=\frac{1}{nh\sigma_Q^2(v)}\sum_{i\in I_k}\sum_{j\in I_k}\E\left[\epsilon_i\epsilon_j\Delta f_i^k\Delta f_j^k K\left(\frac{V_i - v}{h}\right)K\left(\frac{V_j - v}{h}\right) |I_k^c, X_i, X_j, Z_i, Z_j\right]\nonumber\\
    &\textstyle=\frac{1}{nh\sigma_Q^2(v)}\sum_{i\in I_k}\sum_{j\in I_k}\E\left[\epsilon_i\epsilon_j | X_i, X_j, Z_i, Z_j\right]\Delta f_i^k\Delta f_j^k K\left(\frac{V_i - v}{h}\right)K\left(\frac{V_j - v}{h}\right)\nonumber\\
    &\textstyle=\frac{1}{nh\sigma_Q^2(v)}\sum_{i\in I_k}\E[\epsilon_i^2|Z_i, X_i]\left(\Delta f_i^k\right)^2K\left(\frac{V_i - v}{h}\right)^2\label{eq_ExprEpsSq}\\
    &\textstyle\leq \|K\|_{\infty}^2\|\E[\epsilon_i^2|Z_i, X_i]\|_{L_\infty}^2\frac{1}{nh\sigma_Q^2(v)}\sum_{i\in I_k}\left(\Delta f_i^k\right)^2.\nonumber
\end{align}
It follows that
\begin{align*}
    &\E\left[\left(\frac{1}{\sqrt{n h}\sigma_Q(v)}\sum_{i\in I_k}\epsilon_i \Delta f_i^k K\left(\frac{V_i - v}{h}\right)\right)^2| I_k^c\right]\\
    &\leq \|K\|_{\infty}^2\|\E[\epsilon_i^2|Z_i, X_i]\|_{L_\infty}^2\frac{1}{nh\sigma_Q^2(v)}\sum_{i\in I_k}\E\left[\left(\Delta f_i^k\right)^2|I_k^c\right]\\
    &= \|K\|_{\infty}^2\|\E[\epsilon_i^2|Z_i, X_i]\|_{L_\infty}^2\frac{\|f-\hat f^k\|_{L_2}^2}{h\sigma_Q^2(v)}
\end{align*}
By assertions \ref{ass_Rate_F_Phi_slow} and \ref{ass_CondVariance} of Assumption \ref{ass_NuisanceRatesHet}, the last expression is $o_P(1)$. By Lemma \ref{lem_ConditionalToZero} from Section \ref{sec_SomeLemmas}, we obtain \eqref{eq_RobNumHet1}. In exactly the same way, one can also show $\eqref{eq_RobNumHet2}$.

For \eqref{eq_RobNumHet3}, we use the Cauchy-Schwarz inequality in the second step and assertion \ref{ass_Reg5} of Assumption \ref{ass_RegularityHet} with Lemma \ref{lem_ExpFunKernel} in the fourth step to obtain
\begin{align*}
    &\E\left[\left|\frac{1}{\sqrt{n h}\sigma_Q(v)}\sum_{i\in I_k}(\beta(V_i)-\beta(v))(D_i-\phi(X_i)) \Delta f_i^k K\left(\frac{V_i - v}{h}\right)\right||I_k^c\right]\\
    &\leq\frac{1}{\sqrt{n h}\sigma_Q(v)}\sum_{i\in I_k}\E\left[\left|(\beta(V_i)-\beta(v))(D_i-\phi(X_i)) \Delta f_i^k K\left(\frac{V_i - v}{h}\right)\right||I_k^c\right]\\
    &\leq \frac{1}{\sqrt{n h}\sigma_Q(v)}\sum_{i\in I_k}\E\left[\left|(\beta(V_i)-\beta(v))^2(D_i-\phi(X_i))^2  K\left(\frac{V_i - v}{h}\right)^2\right||I_k^c\right]^{1/2}\cdot\\
    &\qquad\qquad\qquad\qquad\quad\cdot\E\left[\left(\Delta f_i^k\right)^2|I_k^c\right]^{1/2}\\
    &=\frac{\sqrt{n}}{\sqrt h \sigma_Q(v)}\E\left[(\beta(V_i)-\beta(v))^2\E\left[(D_i-\phi(X_i))^2|V_i\right]K\left(\frac{V_i - v}{h}\right)^2\right]^{1/2}\cdot\\
    &\qquad\qquad\qquad\cdot\|f-\hat f^k\|_{L_2}\\
    &= \frac{\sqrt{n}}{\sqrt h \sigma_Q(v)}O(h^{3/2})\|f-\hat f^k\|_{L_2}\\
    &= \frac{\|f-\hat f^k\|_{L_2}}{\sigma_Q(v)} O(n^{1/2}h),
\end{align*}
which is $o_P(1)$ by assertion \ref{ass_Rate_F_Phi_fast} of Assumption \ref{ass_NuisanceRatesHet}. By Lemma \ref{lem_ConditionalToZero}, \eqref{eq_RobNumHet3} follows. In exactly the same way, we can also show \eqref{eq_RobNumHet4}.

For \eqref{eq_RobNumHet5}, we can argue similarly to \eqref{eq_RobNumHet1}. Concretely,
\begin{align}
    &\E\left[\left(\frac{1}{\sqrt{n h}\sigma_Q(v)}\sum_{i\in I_k}(f(Z_i, X_i)-\phi(X_i)) \Delta l_i^k K\left(\frac{V_i - v}{h}\right)\right)^2| I_k^c, X_{I_k}\right]\nonumber\\
    &=\frac{1}{nh\sigma_Q^2(v)}\sum_{i\in I_k}\E[(f(Z_i, X_i)-\phi(X_i))^2|X_i]\left(\Delta l_i^k\right)^2K\left(\frac{V_i - v}{h}\right)^2\label{eq_ExprFPhi}
\end{align}
and hence using H\"older's inequality with $1/p + 1/q = 1$
\begin{align*}
    &\E\left[\left(\frac{1}{\sqrt{n h}\sigma_Q(v)}\sum_{i\in I_k}(f(Z_i, X_i)-\phi(X_i)) \Delta l_i^k K\left(\frac{V_i - v}{h}\right)\right)^2| I_k^c\right]\\
    &=\frac{1}{h\sigma_Q^2(v)}\E\left[\E[(f(Z_i, X_i)-\phi(X_i))^2|X_i]\left(\Delta l_i^k\right)^2K\left(\frac{V_i - v}{h}\right)^2|I_k^c\right]\\
    &\leq \frac{1}{h\sigma_Q^2(v)}\E\left[\E[(f(Z_i, X_i)-\phi(X_i))^2|X_i]^pK\left(\frac{V_i - v}{h}\right)^{2p}\right]^{1/p}\E\left[|\Delta l_i^k|^{2q}| I_k^c\right]^{1/q}\\
    & \leq \frac{\|K\|_{\infty}^2}{h\sigma_Q^2(v)}\|\E[(f(Z_i, X_i)-\phi(X_i))^2|X_i]\|_{L_p}\|l-\hat l^k\|_{L_{2q}}^2
\end{align*}
If we set $p = p_1$ and $q = q_1$, this is $o_P(1)$ using assertion \ref{ass_Rate_l_pq1} of Assumption \ref{ass_NuisanceRatesHet}, and \eqref{eq_RobNumHet5} follows using Lemma \ref{lem_ConditionalToZero}. In exactly the same way, \eqref{eq_RobNumHet6} follows.

For \eqref{eq_RobNumHet7}, we use the Cauchy-Schwarz inequality to obtain
\begin{align*}
    &\E\left[\left|\frac{1}{\sqrt{nh}\sigma_Q(v)}\sum_{i\in I_k}\Delta l_i^k\Delta\phi_{2,i}^k K\left(\frac{V_i - v}{h}\right)\right||I_k^c\right]\\
    &\leq \frac{\sqrt n}{\sqrt{h}\sigma_Q(v)}\E\left[\left|\Delta l_i^k\Delta\phi_{2,i}^k\right| K\left(\frac{V_i - v}{h}\right)| I_k^c\right]\\
    &\leq \frac{\sqrt n\|K\|_{\infty}}{\sqrt{h}\sigma_Q(v)}\E\left[\left(\Delta l_i^k\right)^2|I_k\right]^{1/2}\E\left[\left(\Delta \phi_{2,i}^k\right)^2|I_k^c\right]^{1/2}\\
    & = \frac{\sqrt n\|K\|_{\infty}}{\sqrt{h}\sigma_Q(v)}\|l-\hat l^k\|_{L_2}\|\phi-\hat \phi_2^k\|_{L_2}\\
    & = o_P(1)
\end{align*}
by assertion \ref{ass_RateProducts} of Assumption \ref{ass_NuisanceRatesHet}. Hence, \eqref{eq_RobNumHet7} follows using Lemma \ref{lem_ConditionalToZero}. In exactly the same way, also \eqref{eq_RobNumHet8}, \eqref{eq_RobNumHet9} and \eqref{eq_RobNumHet10} follow. This concludes the proof of Lemma \ref{lem_RobNumHetDiff}.

\subsection{Proof or Theorem \ref{thm_RobVarHet}}
By Theorem \ref{thm_RobNumHet}, $h\hat Q(\beta_0(v), v)/\sigma_Q(v) = O_P\left(h/\sqrt{Nh}\right) = o_P(1)$. Hence it is enough to show the following two Lemmas.
\begin{lemma}\label{lem_RobVarHetOracle}
    Under the conditions of Theorem \ref{thm_RobVarHet},
    \begin{align*}
    &\textstyle\frac{1}{Nh\sigma_Q^2(v)}\sum_{i = 1}^N (Y_i- l(X_i)- \beta_0(v)(D_i - \phi(X_i)))^2(f(Z_i, X_i) - \phi(X_i))^2K\left(\frac{V_i - v}{h}\right)^2\\
    &\textstyle = 1 + o_P(1).
    \end{align*}
\end{lemma}
\begin{lemma}\label{lem_RobVarHetDiff}
    Under the conditions of Theorem \ref{thm_RobVarHet},
    \begin{align*}
    &\textstyle\frac{1}{Nh\sigma_Q^2(v)}\sum_{k = 1}^K\sum_{i\in I_k} (Y_i-\hat l^k(X_i)- \beta_0(v)(D_i - \hat\phi_1^k(X_i)))^2\cdot \\
    &\textstyle\qquad\qquad\qquad\qquad\quad\cdot(\hat f^k(Z_i, X_i) - \hat \phi_2^k(X_i))^2K\left(\frac{V_i - v}{h}\right)^2\\
    &\textstyle- \frac{1}{Nh\sigma_Q^2(v)}\sum_{i = 1}^N (Y_i- l(X_i)- \beta_0(v)(D_i - \phi(X_i)))^2\cdot\\
    &\textstyle\qquad\qquad\qquad\quad\cdot( f(Z_i, X_i) -  \phi(X_i))^2K\left(\frac{V_i - v}{h}\right)^2\\
    &\textstyle = o_P(1).
    \end{align*}
\end{lemma}
\subsubsection{Proof of Lemma \ref{lem_RobVarHetOracle}}
As in the proof of Lemma \ref{lem_RobNumHetOracle}, define
\begin{equation}\label{eq_DefWi}
W_i = \frac{1}{\sqrt h \sigma_Q(v)}(\epsilon_i + (\beta(V_i)-\beta(v))(D_i-\phi(X_i)))(f(Z_i, X_i)-\phi(X_i)) K\left(\frac{V_i - v}{h}\right).
\end{equation}
We need to show $\frac{1}{N}\sum_{i=1}^N W_i^2 -1 = o_P(1)$. By \eqref{eq_ExpWiSquared}, $\E[W_i^2] = 1 + O(h^2) = 1 + o_P(1)$. Hence, it suffices to show that $\frac{1}{N}\sum_{i=1}^N \left(W_i^2-\E[W_i^2]\right) = o_P(1).$
By Chebyshev's Inequality, for all $\delta>0$
$$\Prob\left(\left|\frac{1}{N}\sum_{i=1}^N \left(W_i^2-\E[W_i^2]\right)\right| >\delta\right)\leq \frac{\Var(W_i^2)}{N\delta}\leq \frac{\E[W_i^4]}{N\delta} = o(1)$$
by assertion \ref{ass_ChebyVar} of Assumption \ref{ass_RegularityHet}. This concludes the proof of Lemma \ref{lem_RobVarHetOracle}.
\subsubsection{Proof of Lemma \ref{lem_RobVarHetDiff}}
Define $W_i$ as in \eqref{eq_DefWi} and for $i \in I_k$, $k = 1,\ldots, K$
\begin{align*}
    \Gamma_i^k & =- W_i +  \frac{1}{\sqrt h \sigma_Q(v)}(Y_i-\hat l^k(X_i) -\beta(v)(D_i-\hat\phi_1^k(X_i)))\cdot\\
    &\qquad\qquad\qquad\qquad\quad\cdot(\hat f^k(Z_i, X_i)-\hat \phi_2^k(X_i)) K\left(\frac{V_i - v}{h}\right).
\end{align*}
Using this notation, we need to show that
$$\frac{1}{N}\sum_{k = 1}^K\sum_{i\in I_k}\left(\left(\Gamma_i^k + W_i\right)^2 - W_i^2\right)= o_P(1).$$
By the Cauchy-Schwarz inequality,
\begin{align*}
    &\left|\frac{1}{N}\sum_{k = 1}^K\sum_{i\in I_k}\left(\left(\Gamma_i^k + W_i\right)^2 - W_i^2\right)\right| = \frac{1}{N}\sum_{k = 1}^K\sum_{i\in I_k}\left|\left(\Gamma_i^k\right)^2 + 2 \Gamma_i^k W_i\right|\\
    &\leq \frac{1}{N}\sum_{k=1}^K\sum_{i\in I_k}\left(\Gamma_i^k\right)^2 + 2\sqrt{\frac{1}{N}\sum_{k=1}^K\sum_{i\in I_k}\left(\Gamma_i^k\right)^2}\sqrt{\frac{1}{N}\sum_{i=1}^N W_i^2}.
\end{align*}
From Lemma \ref{lem_RobVarHetOracle}, we know that $\frac{1}{N}\sum_{i=1}^N W_i^2= 1+o_P(1)$ and hence it remains to show that $\frac{1}{N}\sum_{k=1}^K\sum_{i\in I_k}\left(\Gamma_i^k\right)^2=o_P(1)$. We use the same decomposition as in \eqref{eq_DecompWi} to see that 
\begin{align*}
    \Gamma_i^k&=\frac{1}{\sqrt h\sigma_Q(v)}\Bigl[\left(\epsilon_i + (\beta(V_i)-\beta(v))(D_i-\phi(X_i))\right)\left(\Delta\phi_{2,i}^k-\Delta f_i^k\right)\\
    &\qquad\qquad\qquad\quad+\left(\Delta l_i^k -\beta(v) \Delta \phi_{1,i}^k\right)\left(f(Z_i,X_i)-\phi(X_i)\right)\\
    &\qquad\qquad\qquad\quad+\left(\Delta l_i^k -\beta(v) \Delta \phi_{1,i}^k\right)\left(\Delta\phi_{2,i}^k-\Delta f_i^k\right)\Bigr]K\left(\frac{V_i - v}{h}\right).
\end{align*}
Hence, to arrive at Lemma \ref{lem_RobVarHetDiff}, it suffices to show that for $k = 1,\ldots, K$, 
\begin{align}
    \frac{1}{{n h}\sigma_Q^2(v)}\sum_{i\in I_k}\epsilon_i^2 (\Delta f_i^k)^2 K\left(\frac{V_i - v}{h}\right)^2 & = o_P(1)\label{eq_RobVarHet1}\\
    \frac{1}{{n h}\sigma_Q^2(v)}\sum_{i\in I_k}\epsilon_i^2 (\Delta \phi_{2,i}^k)^2 K\left(\frac{V_i - v}{h}\right)^2 & = o_P(1)\label{eq_RobVarHet2}\\
     \frac{1}{{n h}\sigma_Q^2(v)}\sum_{i\in I_k}(\beta(V_i)-\beta(v))^2(D_i-\phi(X_i))^2 (\Delta f_i^k)^2 K\left(\frac{V_i - v}{h}\right)^2 & = o_P(1)\label{eq_RobVarHet3}\\
     \frac{1}{{n h}\sigma_Q^2(v)}\sum_{i\in I_k}(\beta(V_i)-\beta(v))^2(D_i-\phi(X_i))^2 (\Delta \phi_{2,i}^k)^2 K\left(\frac{V_i - v}{h}\right)^2 & = o_P(1)\label{eq_RobVarHet4}\\
     \frac{1}{{n h}\sigma_Q^2(v)}\sum_{i\in I_k}(f(Z_i, X_i)-\phi(X_i))^2 (\Delta l_i^k)^2 K\left(\frac{V_i - v}{h}\right)^2 & = o_P(1)\label{eq_RobVarHet5}\\
     \frac{1}{{n h}\sigma_Q^2(v)}\sum_{i\in I_k}(f(Z_i, X_i)-\phi(X_i))^2 (\Delta \phi_{1,i}^k)^2 K\left(\frac{V_i - v}{h}\right)^2 & = o_P(1)\label{eq_RobVarHet6}\\
     \frac{1}{{n h}\sigma_Q^2(v)}\sum_{i\in I_k} (\Delta l_i^k)^2 (\Delta \phi_{1,i}^k)^2 K\left(\frac{V_i - v}{h}\right)^2 & = o_P(1)\label{eq_RobVarHet7}\\
     \frac{1}{{n h}\sigma_Q^2(v)}\sum_{i\in I_k} (\Delta l_i^k)^2 (\Delta f_i^k)^2 K\left(\frac{V_i - v}{h}\right)^2 & = o_P(1)\label{eq_RobVarHet8}\\
     \frac{1}{{n h}\sigma_Q^2(v)}\sum_{i\in I_k} (\Delta \phi_{1,i}^k)^2 (\Delta \phi_{2,i}^k)^2K\left(\frac{V_i - v}{h}\right)^2 & = o_P(1)\label{eq_RobVarHet9}\\
     \frac{1}{{n h}\sigma_Q^2(v)}\sum_{i\in I_k} (\Delta \phi_{1,i}^k)^2 (\Delta f_i^k)^2 K\left(\frac{V_i - v}{h}\right)^2 & = o_P(1).\label{eq_RobVarHet10}
\end{align}
For \eqref{eq_RobVarHet1}, observe that
\begin{align*}
    &\E\left[\frac{1}{{n h}\sigma_Q^2(v)}\sum_{i\in I_k}\epsilon_i^2 (\Delta f_i^k)^2 K\left(\frac{V_i - v}{h}\right)^2|I_k^c, Z_{I_k}, X_{I_k}\right]\\
    &=\frac{1}{{n h}\sigma_Q^2(v)}\sum_{i\in I_k}\E\left[\epsilon_i^2|Z_i, X_i\right] (\Delta f_i^k)^2 K\left(\frac{V_i - v}{h}\right)^2,
\end{align*}
which is the same expression as \eqref{eq_ExprEpsSq}. Hence, \eqref{eq_RobVarHet1} follows as before. The same argument can be applied for \eqref{eq_RobVarHet2}. From the proof of \eqref{eq_RobNumHet3}, we know that 
$$\frac{1}{\sqrt{n h}\sigma_Q(v)}\sum_{i\in I_k}\left|(\beta(V_i)-\beta(v))(D_i-\phi(X_i)) \Delta f_i^k K\left(\frac{V_i - v}{h}\right)\right|  = o_P(1).$$
Hence, \eqref{eq_RobVarHet3} follows, since
\begin{align*}
    &\frac{1}{{n h}\sigma_Q^2(v)}\sum_{i\in I_k}(\beta(V_i)-\beta(v))^2(D_i-\phi(X_i))^2 (\Delta f_i^k)^2 K\left(\frac{V_i - v}{h}\right)^2\\
    &\leq \left(\frac{1}{\sqrt{n h}\sigma_Q(v)}\sum_{i\in I_k}\left|(\beta(V_i)-\beta(v))(D_i-\phi(X_i)) \Delta f_i^k K\left(\frac{V_i - v}{h}\right)\right|\right)^2.
\end{align*}
In the same way, also \eqref{eq_RobVarHet4} follows. For \eqref{eq_RobVarHet5}, we argue as for \eqref{eq_RobVarHet1} by noting that the conditional expectation of the expression given $(I_k^c, X_{I_k})$ is equal to \eqref{eq_ExprFPhi} that has been analyzed before. The same reasoning applies to \eqref{eq_RobVarHet6}. For \eqref{eq_RobVarHet7}, we have established in the proof of \eqref{eq_RobNumHet7} that 
$$\frac{1}{\sqrt{nh}\sigma_Q(v)}\sum_{i\in I_k}\left|\Delta l_i^k\Delta \phi_{2,i}^k K\left(\frac{V_i - v}{h}\right)\right| = o_P(1).$$
Hence, \eqref{eq_RobVarHet7} follows by noting that
\begin{align*}
&\frac{1}{{n h}\sigma_Q^2(v)}\sum_{i\in I_k}(\Delta l_i^k)^2 (\Delta \phi_{2,i}^k)^2 K\left(\frac{V_i - v}{h}\right)^2\\
&\leq \left(\frac{1}{\sqrt{nh}\sigma_Q(v)}\sum_{i\in I_k}\left|\Delta l_i^k\Delta \phi_{2,i}^k K\left(\frac{V_i - v}{h}\right)\right|\right)^2.
\end{align*}
In the same way, also \eqref{eq_RobVarHet8}, \eqref{eq_RobVarHet9} and \eqref{eq_RobVarHet10} follow.

\subsection{Proof of Theorem \ref{thm_AsNormHet}}
We can write
\begin{align*}
    &\textstyle\frac{\sqrt{Nh}(\hat\beta(v)-\beta(v))}{\sigma(v)}\\
    &\textstyle=\frac{\frac{1}{\sqrt{Nh}}\sum_{k = 1}^K\sum_{i\in I_k} (Y_i-\hat l^k(X_i)-\beta(v)(D_i-\hat\phi_1^k(X_i)))(\hat f^k(Z_i, X_i) - \hat \phi_2^k(X_i))K\left(\frac{V_i - v}{h}\right)}{\frac{\sigma(v)}{Nh}\sum_{k = 1}^K\sum_{i\in I_k} (D_i-\hat \phi_1^k(X_i))(\hat f^k(Z_i, X_i) - \hat \phi_2^k(X_i))K\left(\frac{V_i - v}{h}\right)}\\
    &\textstyle=\frac{\sqrt{Nh}\hat Q(\beta(v), v)/\sigma_Q(v)}{\frac{\sigma(v)/\sigma_Q(v)}{Nh}\sum_{k = 1}^K\sum_{i\in I_k} (D_i-\hat \phi_1^k(X_i))(\hat f^k(Z_i, X_i) - \hat \phi_2^k(X_i))K\left(\frac{V_i - v}{h}\right)}
\end{align*}
From Theorem \ref{thm_RobNumHet}, we know that the numerator converges in distribution to $\mathcal N(0,1)$. By Slutsky's Theorem, it is enough to show that the denominator converges to $1$ in probability. Hence, it is enough to prove the following two lemmas.
\begin{lemma}\label{lem_DenomHetOracle}
    Under the conditions of Theorem \ref{thm_AsNormHet},
    $$\frac{\sigma(v)/\sigma_Q(v)}{Nh}\sum_{i = 1}^N (D_i- \phi(X_i))(f(Z_i, X_i) - \phi(X_i))K\left(\frac{V_i - v}{h}\right)- 1 = o_P(1).$$
\end{lemma}

\begin{lemma}\label{lem_DenomHetDiff}
    Under the conditions of Theorem \ref{thm_AsNormHet},
    \begin{align*}
        &\textstyle \frac{\sigma(v)/\sigma_Q(v)}{Nh}\sum_{k = 1}^K\sum_{i\in I_k} (D_i-\hat \phi_1^k(X_i))(\hat f^k(Z_i, X_i) - \hat \phi_2^k(X_i))K\left(\frac{V_i - v}{h}\right)\\
        &\textstyle -\frac{\sigma(v)/\sigma_Q(v)}{Nh}\sum_{i = 1}^N (D_i- \phi(X_i))(f(Z_i, X_i) - \phi(X_i))K\left(\frac{V_i - v}{h}\right)
        &\textstyle = o_P(1).
    \end{align*}
\end{lemma}
\subsubsection{Proof of Lemma \ref{lem_DenomHetOracle}}
Note that by the tower property of conditional expectations,
\begin{align}
   \frac{\sigma(v)}{\sigma_Q(v)} &= \left(\E[(D_i-\phi(X_i))(f(Z_i, X_i)-\phi(X_i))|V_i = v] p_V(v)\right)^{-1}\nonumber\\
   &= \left(\E[(f(Z_i, X_i)-\phi(X_i))^2|V_i = v] p_V(v)\right)^{-1}.\label{eq_SigDivSigQ}
\end{align}
By Lemma \ref{lem_ExpFunKernel} and assertion \ref{ass_Reg6} of Assumption \ref{ass_RegularityHet}
\begin{align*}
    &\E\left[(D_i- \phi(X_i))(f(Z_i, X_i) - \phi(X_i))K\left(\frac{V_i - v}{h}\right)\right]\\
    &=\E\left[\E[(f(Z_i, X_i)-\phi(X_i))^2|V_i]K\left(\frac{V_i - v}{h}\right)\right]\\
    & =  h \E[(f(Z_i, X_i)-\phi(X_i))^2|V_i = v]p_V(v) + \sigma_Q^2(v)O(h^3).
\end{align*}
From assertions \ref{ass_Density} and \ref{ass_CompSigmaQ} of Assumption \ref{ass_RegularityHet}, it follows that
$$U_i = \frac{\sigma(v)/\sigma_Q(v)}{h} (D_i- \phi(X_i))(f(Z_i, X_i) - \phi(X_i))K\left(\frac{V_i - v}{h}\right)$$
satisfy $\E[U_i] = 1 + O(h^2) = 1 + o(1)$. To prove Lemma \ref{lem_DenomHetOracle}, it therefore remains to prove $\frac{1}{N}\sum_{i=1}^N (U_i-\E[U_i]) = o_P(1)$. For this, using Chebyshev's inequality, it is enough to show
\begin{equation}\label{eq_ExpViSq}
    \frac{\E[U_i^2]}{N}=o(1).
\end{equation}
We calculate $\frac{\E[U_i^2]}{N}$ using \eqref{eq_SigDivSigQ}, apply Lemma \ref{lem_ExpFunKernel} and assertions \ref{ass_Reg7}, \ref{ass_CompSigmaQ2}, \ref{ass_CompSigmaQ} and \ref{ass_Density} of Assumption \ref{ass_RegularityHet} to get
\begin{align*}
    \textstyle\frac{\E[U_i^2]}{N}& = \textstyle\frac{\E\left[(D_i-\phi(X_i))^2(f(Z_i, X_i)-\phi(X_i))^2K\left(\frac{V_i - v}{h}\right)^2\right]}{Nh^2 \E[(f(Z_i, X_i)-\phi(X_i))^2|V_i = v]^2 p_V(v)^2}\\
    &\textstyle=\frac{h\E\left[(D_i-\phi(X_i))^2(f(Z_i, X_i)-\phi(X_i))^2|V_i = v\right]p_V(v)\int_{-\infty}^{\infty} K(x)^2\mathrm d x + \sigma_Q^2(v)O(h^3)}{Nh^2 \E[(f(Z_i, X_i)-\phi(X_i))^2|V_i = v]^2 p_V(v)^2}\\
    &\textstyle\lesssim \frac{h\sigma_Q^2(v)+h^3\sigma_Q^2(v)}{N h^2\sigma_Q^4(v)}\\
    &\textstyle \lesssim \frac{1}{Nh\sigma_Q^2(v)}.
\end{align*}
Since by the assumption of Theorem \ref{thm_AsNormHet}, $\sigma_Q^2(v) \gg (Nh)^{-1}$, \eqref{eq_ExpViSq} follows. This concludes the proof of Lemma \ref{lem_DenomHetOracle}. 
\subsubsection{Proof of Lemma \ref{lem_DenomHetDiff}}
Let $\Delta \phi_{1,i}^k = \phi(X_i)-\hat\phi_1^k(X_i)$, $\Delta \phi_{2,i}^k = \phi(X_i)-\hat\phi_2^k(X_i)$ and $\Delta f_i^k = f(X_i)-\hat f^k(X_i)$ as before. Then, the quantity given in Lemma \ref{lem_DenomHetDiff} is equal to
\begin{align*}
    &\frac{\sigma(v)/\sigma_Q(v)}{Nh}\sum_{k=1}^K\sum_{i\in I_k}\bigl[(D_i-\phi(X_i))(\Delta \phi_{2,i}^k -\Delta f_i^k)\\
    &\quad\quad\quad + \Delta \phi_{1,i}^k(f(Z_i, X_i)-\phi(X_i))\\
    &\quad\quad\quad + \Delta \phi_{1,i}^k(\Delta \phi_{2,i}^k- \Delta f_i^i)\bigr]K\left(\frac{V_i - v}{h}\right).
\end{align*}
Moreover, we know from \eqref{eq_SigDivSigQ} and assertions \ref{ass_Density} and \ref{ass_CompSigmaQ} of Assumption \ref{ass_RegularityHet} that $\sigma(v)/\sigma_Q(v)\asymp \frac{1}{\sigma_Q^2(v)}$
Hence, it remains to show that for $k = 1, \ldots, K$,
\begin{align}
    \frac{1}{nh\sigma_Q^2(v)}\sum_{i\in I_k}(D_i-\phi(X_i))\Delta f_i^kK\left(\frac{V_i - v}{h}\right) = o_P(1)\label{eq_EstDenHet1}\\
    \frac{1}{nh\sigma_Q^2(v)}\sum_{i\in I_k}(D_i-\phi(X_i))\Delta \phi_{2,i}^kK\left(\frac{V_i - v}{h}\right) = o_P(1)\label{eq_EstDenHet2}\\
    \frac{1}{nh\sigma_Q^2(v)}\sum_{i\in I_k}(f(Z_i, X_i)-\phi(X_i))\Delta \phi_{1,i}^k K\left(\frac{V_i - v}{h}\right)= o_P(1)\label{eq_EstDenHet3}\\
    \frac{1}{nh\sigma_Q^2(v)}\sum_{i\in I_k}\Delta f_i^k\Delta \phi_{1,i}^kK\left(\frac{V_i - v}{h}\right) = o_P(1)\label{eq_EstDenHet4}\\
    \frac{1}{nh\sigma_Q^2(v)}\sum_{i\in I_k}\Delta\phi_{1,i}^k\Delta\phi_{2,i}^kK\left(\frac{V_i - v}{h}\right) = o_P(1).\label{eq_EstDenHet5}
\end{align}
For \eqref{eq_EstDenHet1}, we use the Cauchy-Schwarz inequality to get
\begin{align*}
    &\left|\frac{1}{nh\sigma_Q^2(v)}\sum_{i\in I_k}(D_i-\phi(X_i))\Delta f_i^kK\left(\frac{V_i - v}{h}\right)\right|\\
    &\leq \frac{1}{\sigma_Q^2(v)}{\sqrt{\frac{1}{nh}\sum_{i\in I_k}(D_i-\phi(X_i))^2K\left(\frac{V_i - v}{h}\right)}}{\sqrt{\frac{1}{nh}\sum_{i\in I_k}(\Delta f_i^k)^2K\left(\frac{V_i - v}{h}\right)}}.
\end{align*}
Hence, to show \eqref{eq_EstDenHet1}, it is enough to prove that
\begin{align}
    \frac{1}{nh}\sum_{i\in I_k}(D_i-\phi(X_i))^2K\left(\frac{V_i - v}{h}\right) & = O_P(1)\label{eq_EstDenHet6}\\
    \frac{1}{\sigma_Q^4(v) nh}\sum_{i\in I_k}(\Delta f_i^k)^2K\left(\frac{V_i - v}{h}\right) & = o_P(1).\label{eq_EstDenHet7}
\end{align}
\eqref{eq_EstDenHet6} can be obtained as in the proof of Lemma \ref{lem_RobVarHetOracle} using Assertions \ref{ass_Reg8} and \ref{ass_RD4} of Assumption \ref{ass_RegularityHet} and Lemma \ref{lem_ExpFunKernel}. \eqref{eq_EstDenHet7} follows similarly to the proof of \eqref{eq_RobNumHet7} by noticing that $\frac{\|K\|_{\infty}}{\sigma_Q^4(v) h}\|f-\hat f^k\|_{L_2}^2 = o_P(1)$ by the assumption in Theorem \ref{thm_AsNormHet}.
Hence, we obtain \eqref{eq_EstDenHet1}. In the same way, one can obtain \eqref{eq_EstDenHet2}. For \eqref{eq_EstDenHet3}, it follows from \eqref{eq_RobNumHet6} that
$$\frac{1}{nh\sigma_Q^2(v)}\sum_{i\in I_k}(f(Z_i, X_i)-\phi(X_i))\Delta \phi_{1,i}^k K\left(\frac{V_i - v}{h}\right)=\frac{1}{\sigma_Q(v)\sqrt{nh}} o_P(1).$$
By the assumption of Theorem \ref{thm_AsNormHet}, $\sigma_Q^2(v)\gg (Nh)^{-1}$ and hence, \eqref{eq_EstDenHet3} follows. In the same way, \eqref{eq_EstDenHet4} and \eqref{eq_EstDenHet5} follow from \eqref{eq_RobNumHet10} and \eqref{eq_RobNumHet9}.

\subsection{Proof of Theorem \ref{thm_VarHet}}
We know from Lemma \ref{lem_DenomHetOracle} and Lemma \ref{lem_DenomHetDiff} that for the denominator of $\hat\sigma^2(v)$ defined in \eqref{eq_DefVarianceHet},
\begin{align*}
    &\left(\frac{\sigma(v)/\sigma_Q(v)}{Nh}\sum_{k = 1}^K\sum_{i\in I_k} (D_i-\hat \phi_1^k(X_i))(\hat f^k(Z_i, X_i) - \hat \phi_2^k(X_i))K\left(\frac{V_i - v}{h}\right)\right)^2\\
    &= 1 + o_P(1).
\end{align*}
Therefore, it remains to show that for the numerator we have
\begin{align*}
    &\frac{1}{Nh\sigma_Q^2(v)}\sum_{k=1}^K\sum_{i\in I_k}\left(Y_i-\hat l^k(X_i)-\hat \beta(v)(D_i-\hat\phi_1^k(X_i))\right)^2\cdot\\
    &\qquad\qquad\qquad\qquad\cdot\left(\hat f^k(Z_i, X_i)-\hat\phi_2^k(X_i))\right)^2K\left(\frac{V_i - v}{h}\right)^2\\
    & = 1 + o_P(1).
\end{align*}
This equality was shown to hold if one replaces $\hat\beta(v)$ by $\beta(v)$ in Lemma \ref{lem_RobVarHetOracle} and Lemma \ref{lem_RobVarHetDiff}. Hence, it is enough to show
\begin{align}
 &\frac{1}{Nh\sigma_Q^2(v)}\sum_{k=1}^K\sum_{i\in I_k}\left(Y_i-\hat l^k(X_i)-\hat \beta(v)(D_i-\hat\phi_1^k(X_i))\right)^2\cdot\\
 &\qquad\qquad\qquad\qquad\cdot\left(\hat f^k(Z_i, X_i)-\hat\phi_2^k(X_i))\right)^2K\left(\frac{V_i - v}{h}\right)^2\nonumber\\
 &-\frac{1}{Nh\sigma_Q^2(v)}\sum_{k=1}^K\sum_{i\in I_k}\left(Y_i-\hat l^k(X_i)-\beta(v)(D_i-\hat\phi_1^k(X_i))\right)^2\cdot\\
 &\qquad\qquad\qquad\qquad\quad\cdot\left(\hat f^k(Z_i, X_i)-\hat\phi_2^k(X_i))\right)^2K\left(\frac{V_i - v}{h}\right)^2\nonumber\\
 &= o_P(1).
    \label{eq_VarDiffBetaHat}
\end{align}
Define
\begin{align*}
    \Xi_i^k & = \frac{1}{\sqrt{h}\sigma_Q(v)}(\beta(v)-\hat\beta(v))(D_i-\hat\phi_1^k(X_i))\cdot\\
    &\qquad\qquad\qquad\cdot\left(\hat f^k(Z_i, X_i)-\hat\phi_2^k(X_i)\right)K\left(\frac{V_i - v}{h}\right)\\
    \Theta_i^k &= \frac{1}{\sqrt{h}\sigma_Q(v)}\left(Y_i-\hat l^k(X_i)-\beta(v)(D_i-\hat\phi_1^k(X_i))\right)\cdot\\
    &\qquad\qquad\qquad\cdot\left(\hat f^k(Z_i, X_i)-\hat\phi_2^k(X_i)\right)K\left(\frac{V_i - v}{h}\right).
\end{align*}
The expression in \eqref{eq_VarDiffBetaHat} is then equal to
$$\frac{1}{N}\sum_{k=1}^K\sum_{i\in I_k}\left((\Xi_i^k+\Theta_i^k)^2-(\Theta_i^k)^2\right).$$
By the same argument as in the proof of Lemma \ref{lem_RobVarHetDiff}, it remains to show
\begin{align}
    \frac{1}{N}\sum_{k=1}^K\sum_{i\in I_k}(\Theta_i^k)^2 &= O_P(1) \label{eq_SumThetaSq}\\
    \frac{1}{N}\sum_{k=1}^K\sum_{i\in I_k}(\Xi_i^k)^2 & = o_P(1) \label{eq_SumXiSq}
\end{align}
\eqref{eq_SumThetaSq} holds by Lemma \ref{lem_RobVarHetOracle} and Lemma \ref{lem_RobVarHetDiff}. For \eqref{eq_SumXiSq}, observe that
\begin{align*}
    \frac{1}{N}\sum_{i=1}^N(\Xi_i^k)^2
    &=\frac{(\beta(v)-\hat\beta(v))^2}{Nh\sigma_Q^2(v)}\sum_{k=1}^K\sum_{i\in I_k}(D_i-\hat\phi_1^k(X_i))^2\cdot\\
    &\qquad\qquad\qquad\qquad\qquad\quad\cdot\left(\hat f^k(Z_i, X_i)-\hat\phi_2^k(X_i)\right)^2K\left(\frac{V_i - v}{h}\right)^2.
\end{align*}
By Theorem \ref{thm_AsNormHet}, we know that $\sqrt{Nh}(\hat\beta(v)-\beta(v))/\sigma(v) = O_P(1)$. 
From \eqref{eq_SigDivSigQ} and assertions \ref{ass_Density} and \ref{ass_CompSigmaQ} of Assumption \ref{ass_RegularityHet}, we know that $\sigma(v)\asymp 1/\sigma_Q(v)$. Since by assumption $\sigma_Q(v)\gg (Nh)^{-1/2}$, it follows that $(\hat\beta(v)-\beta(v))^2 = o_P(1)$. It remains to prove that
\begin{align}
&\frac{1}{Nh\sigma_Q^2(v)}\sum_{k=1}^K\sum_{i\in I_k}(D_i-\hat\phi_1^k(X_i))^2\left(\hat f^k(Z_i, X_i)-\hat\phi_2^k(X_i)\right)^2K\left(\frac{V_i - v}{h}\right)^2\nonumber\\
&= O_P(1).\label{eq_VarHetBounded}
\end{align}
As before, we first consider the ``oracle'' version of \eqref{eq_VarHetBounded} with $\hat \phi_{1,2}^k$ and $\hat f^k$ replaced by $\phi$ and $f$. For this, note that by assertion \ref{ass_Reg7} of Assumption \ref{ass_RegularityHet} and Lemma \ref{lem_ExpFunKernel},
\begin{align*}
    \E\left[\frac{1}{Nh\sigma_Q^2(v)}\sum_{i= 1}^N(D_i-\phi(X_i))^2\left(f(Z_i, X_i)-\phi(X_i)\right)^2K\left(\frac{V_i - v}{h}\right)^2\right]&=O(1)
\end{align*}
and by Lemma \ref{lem_ConditionalBounded}, we conclude that the oracle version of \eqref{eq_VarHetBounded} is $O_P(1)$. It remains to show that the difference between \eqref{eq_VarHetBounded} and its ``oracle'' version is $O_P(1)$ (note that we only need $O_P(1)$ here and not $o_P(1)$). Similarly to before, this reduces to show for $k = 1,\ldots, K$
\begin{align}
    \frac{1}{nh\sigma_Q^2(v)}\sum_{i\in I_k}(D_i-\phi(X_i))^2(\Delta f_i^k)^2K\left(\frac{V_i - v}{h}\right)^2&=O_P(1)\label{eq_EstVarHet1}\\
    \frac{1}{nh\sigma_Q^2(v)}\sum_{i\in I_k}(D_i-\phi(X_i))^2(\Delta \phi_{2,i}^k)^2K\left(\frac{V_i - v}{h}\right)^2&=O_P(1)\label{eq_EstVarHet2}\\
    \frac{1}{nh\sigma_Q^2(v)}\sum_{i\in I_k}(f(Z_i, X_i)-\phi(X_i))^2(\Delta \phi_{1,i}^k)^2K\left(\frac{V_i - v}{h}\right)^2&=O_P(1)\label{eq_EstVarHet3}\\
    \frac{1}{nh\sigma_Q^2(v)}\sum_{i\in I_k}(\Delta f_i^k)^2(\Delta \phi_{1,i}^k)^2K\left(\frac{V_i - v}{h}\right)^2&=O_P(1)\label{eq_EstVarHet4}\\
    \frac{1}{nh\sigma_Q^2(v)}\sum_{i\in I_k}(\Delta \phi_{1,i}^k)^2(\Delta \phi_{2,i}^k)^2K\left(\frac{V_i - v}{h}\right)^2&=O_P(1).\label{eq_EstVarHet5}
\end{align}
For \eqref{eq_EstVarHet1}, using the Cauchy-Schwarz inequality, it suffices to show that $\frac{1}{nh}\sum_{i\in I_k}(D_i-\phi(X_i))^4 K\left(\frac{V_i - v}{h}\right)^2 = O_P(1)$ and $\frac{1}{nh\sigma_Q^4(v)}\sum_{i\in I_k}(\Delta f_i^k)^4 K\left(\frac{V_i - v}{h}\right)^2 = O_P(1)$. The expectation of the first term is bounded by assertion \ref{ass_RD4} of Assumption \ref{ass_RegularityHet} and hence the term is $O_P(1)$ by Lemma \ref{lem_ConditionalBounded}. The second follows by taking the expectation conditional on $I_k^c$ since $\|f-\hat f\|_{L_4}/\sigma_Q(v) = O_P\left(h^{1/4}\right)$ by the assumption of Theorem \ref{thm_VarHet}. Hence, \eqref{eq_EstVarHet1} follows and \eqref{eq_EstVarHet2} follows in exactly the same way. Finally, \eqref{eq_EstVarHet3}, \eqref{eq_EstVarHet4} and \eqref{eq_EstVarHet5} follow from \eqref{eq_RobVarHet6}, \eqref{eq_RobVarHet10} and \eqref{eq_RobVarHet9}. This concludes the proof of \eqref{eq_VarHetBounded} and hence also the proof of Theorem \ref{thm_VarHet}.

\subsection{Some Lemmas}\label{sec_SomeLemmas}
\begin{lemma}\label{lem_ExpFunKernel}
     Assume that the kernel $K$ satisfies Assumption \ref{ass_Kernel}. Moreover, assume that the function $\rho:\mathbb R\to\mathbb R$ defined by $\rho(x) = v(x) \cdot p_V(x)$ for some function $v(\cdot)$ is two times differentiable for $x$ in the support of $V_i$. Then,
    \begin{enumerate}
        \item $$\E\left[K\left(\frac{V_i - v}{h}\right)v(V_i)\right] = h\rho(v) + h^3\frac{\rho''(v)}{2}\int_{-\infty}^\infty K(x) x^2 \mathrm{d}x + R(v, h)$$
    with
    \begin{align*}
        &|R(v, h)|\leq h^3\Biggl(\|\rho ''\|_\infty \int_{|x|>h^{-1/2}}K(x)x^2\mathrm{d}x \\
    &\qquad\qquad\qquad\quad+ \frac{1}{2}\int_{-\infty}^\infty K(x) x^2 \mathrm{d}x \cdot \sup_{|s-v|\leq h^{1/2}}|\rho''(s)-\rho''(v)|\Biggr).
    \end{align*}
    In particular, if $\rho''$ is continuous at $v$, then $R(v, h) = o(h^3)$. If $|\rho''|$ bounded but not continuous at $v$, we still have $R(v, h) = O(h^3)$.
    \item 
    \begin{align*}
    &\E\left[K\left(\frac{V_i - v}{h}\right)^2v(V_i)\right] \\
    &= h\rho(v)\int_{-\infty}^\infty K(x)^2\mathrm{d}x + h^3\frac{\rho''(v)}{2}\int_{-\infty}^\infty K(x)^2 x^2 \mathrm{d}x + \bar R(v, h)
    \end{align*}
    with
    \begin{align*}
    &|\bar R (a, h)|\leq h^3\Biggl(\|\rho ''\|_\infty \int_{|x|>h^{-1/2}}K(x)^2x^2\mathrm{d}x\\
    &\qquad\qquad\qquad\quad+ \frac{1}{2}\int_{-\infty}^\infty K(x)^2 x^2 \mathrm{d}x \cdot \sup_{|s-v|\leq h^{1/2}}|\rho''(s)-\rho''(v)|\Biggr).
    \end{align*}
    In particular, if $\rho''$ is continuous at $v$, then $\bar R(v, h) = o(h^3)$. If $|\rho''|$ bounded but not continuous at $v$, we still have $\bar R(v, h) = O(h^3)$.
    \end{enumerate}
    
\end{lemma}
\begin{proof}
    We only prove the first statement. The second statement follows by replacing $K$ with $K^2$ in the proof. From Taylor's theorem, we can write for some $\xi_x\in [v, x]$ (or in $[x, v]$)
    $$\rho(x) = \rho(v) + \rho'(v)(x-v) +\rho''(v)\frac{(x-v)^2}{2} + (\rho''(\xi_x)-\rho''(v))\frac{(x-v)^2}{2}.$$
    Hence, using that $K(\cdot)$ is symmetric,
    \begin{align*}
        \E\left[K\left(\frac{V_i - v}{h}\right)v(V_i)\right]&=\int K\left(\frac{x-v}{h}\right)v(x) p_V(x)\mathrm{d}x\\
        &=\int K\left(\frac{x-v}{h}\right)\rho(x)\mathrm{d}x\\
        &=\rho(v)\int K\left(\frac{x-v}{h}\right)\mathrm{d}x\\
        &+\frac{\rho''(v)}{2}\int K\left(\frac{x-v}{h}\right)(x-v)^2\mathrm{d}x + R(v, h)\\
        & = h\rho(v) + h^3\frac{\rho''(v)}{2}\int K(x) x^2 \mathrm{d}x + R(v, h)
    \end{align*}
    with
    $$R(v, h) = \int K\left(\frac{x-v}{h}\right)\left(\rho''(\xi_x)-\rho''(v)\right)\frac{(x-v)^2}{2}\mathrm{d}x.$$
    To conclude, note that for any $t>0$, 
    \begin{align*}
        &|R(v, h)|\\
        &\leq |\int_{|x-v|>th}K\left(\frac{x-v}{h}\right)\left(\rho''(\xi_x)-\rho''(v)\right)\frac{(x-v)^2}{2}\mathrm{d}x|\\
        & + |\int_{|x-v|\leq th}K\left(\frac{x-v}{h}\right)\left(\rho''(\xi_x)-\rho''(v)\right)\frac{(x-v)^2}{2}\mathrm{d}x|\\
        &\leq \int_{|x-v|>th}K\left(\frac{x-v}{h}\right)2\|\rho''\|_{\infty}\frac{(x-v)^2}{2}\mathrm{d}x\\
        &+ \int_{|x-v|\leq th}K\left(\frac{x-v}{h}\right)\sup_{|s-a|\leq ht}|\rho''(s)-\rho''(v)|\frac{(x-v)^2}{2}\mathrm{d}x\\
        &\leq h^3\left(\|\rho ''\|_\infty \int_{|x|>t}K(x)x^2\mathrm{d}x + \frac{1}{2}\int_{-\infty}^\infty K(x) x^2 \mathrm{d}x \cdot \sup_{|s-v|\leq th}|\rho''(s)-\rho''(v)|\right).
    \end{align*}
    Choosing $t = h^{-1/2}$ gives the result.
\end{proof}

\begin{lemma}\label{lem_ConditionalToZero}
    Let $(X_n)_{n\in \mathbb N}$ be a sequence of integrable random variables and $(\mathcal F_n)_{n\in \mathbb N}$ be a sequence of sigma-algebras. If $\E[|X_n||\mathcal F_n] = o_P(1)$, then we have that $X_n = o_P(1)$.
\end{lemma}
\begin{proof}
    Let $\epsilon>0$. By Markov's inequality,
    $$\Prob(|X_n|\geq \epsilon) = \Prob(|X_n|\land \epsilon \geq \epsilon)\leq \E\left[\frac{\E[|X_n|\land \epsilon|\mathcal F_n]}{\epsilon}\right]\leq\E\left[\frac{\E[|X_n||\mathcal F_n]}{\epsilon}\land 1\right]$$
    
    Since $\E[|X_n||\mathcal F_n] = o_P(1)$, for each $\delta \in (0,1)$, there exists $n_0\in \mathbb N$ such that for all $n\geq n_0$, we have $\Prob(\E[|X_n||\mathcal F_n]\geq \epsilon\delta/2]\leq \delta/2$. It follows that for all $n\geq n_0$, 
    $$\Prob(|X_n|\geq \epsilon)\leq \frac{\epsilon\delta/2}{\epsilon}\Prob\left(\E[|X_n||\mathcal F_n]\leq \epsilon\delta/2\right)+\Prob\left(\E[|X_n||\mathcal F_n]\geq \epsilon\delta/2\right)\leq \delta.$$
    Since $\epsilon$ and $\delta$ were arbitrary, it follows that $X_n = o_P(1)$.
\end{proof}

\begin{lemma}\label{lem_ConditionalBounded}
    Let $(X_n)_{n\in \mathbb N}$ be a sequence of integrable random variables and $(\mathcal F_n)_{n\in \mathbb N}$ be a sequence of sigma-algebras. If $\E[|X_n||\mathcal F_n] = O_P(1)$, then we have that $X_n = O_P(1)$.
\end{lemma}
\begin{proof}
    Let $\epsilon\in (0,1)$. Since $\E[|X_n||\mathcal F_n] = O_P(1)$, there exists $M > 0$ and $n_0\in\mathbb N$, such that for all $n\geq n_0$
    $$\Prob(\E[|X_n||\mathcal F_n]> \epsilon M/2)\leq \epsilon/2.$$
    By Markov's inequality, for all $n\geq n_0$, 
    \begin{align*}
        \Prob(|X_n|\geq M) &= \Prob(|X_n|\land M\geq M)\\
        &\leq \E\left[\frac{\E[|X_n|\land M|\mathcal F_n]}{M}\right]\\
        &\leq \E\left[\frac{\E[|X_n||\mathcal F_n]}{M}\land 1\right]\\
        &\leq \Prob\left(\frac{\E[|X_n||\mathcal F_n]}{M}> \epsilon/2\right)+\epsilon/2\leq \epsilon.
    \end{align*}
    Since $\epsilon\in (0,1)$ was arbitrary, $X_n = O_P(1)$.
\end{proof}

\section{Further Simulations}\label{sec_FurtherSim}
{
\subsection{Varying IV Strength}\label{sec_VaryIVStrength}
We consider the same simulation setup \eqref{eq_SimSetup} as in Section \ref{sec_Simulations} with two differences:
\begin{enumerate}
    \item We introduce an additional parameter $s$ (for the IV \textit{strength}) to the function $f(z,x)$, namely
    \begin{description}
        \item[(Z lin.)] $f(z, x) = -\sin(x) + s\cdot z$,
        \item[(Z nonlin.)] $f(z, x) = -\sin(x) + s\cdot(\cos(z) + 0.2z)$
    \end{description}
    If the strength $s = 0$, the treatment $D_i$ is not related to the instrument $Z_i$. If $s$ increases, the strength of the instrument $Z_i$ increases.
    \item We replace the definitions of $\delta_i$ and $\epsilon_i$ by 
    \begin{align*}
        \delta_i &\gets 0.7H_i + 0.1 E_{\delta, i},\\
        \epsilon_i &\gets 0.7H_i + 0.1 E_{\epsilon,i}
    \end{align*}
    to make them more correlated, i.e., increase the amount of endogeneity.
\end{enumerate}
We fix $N = 500$ and vary the strength $s\in \{0,0.1,\ldots, 1\}$. Again, we simulate $n_\text{sim} = 1000$ datasets and use generalized additive models (\texttt{gam}) for the nuisance functions, the Epanechnikov kernel, and choose the bandwidth as the normal reference rule times $N^{1/5}/N^{1/7}$. The corresponding plots for the settings (hom.)/(Z lin.) and (hom.)/(Z nonlin.) can be found in Figure \ref{fig_SimHom1DWeak} and the plots for the settings (het.)/(Z lin.) and (het.)/(Z nonlin.) can be found in Figure \ref{fig_SimHet1DWeak}.
\begin{figure}[t]
\centering
\includegraphics[width=0.8\textwidth]{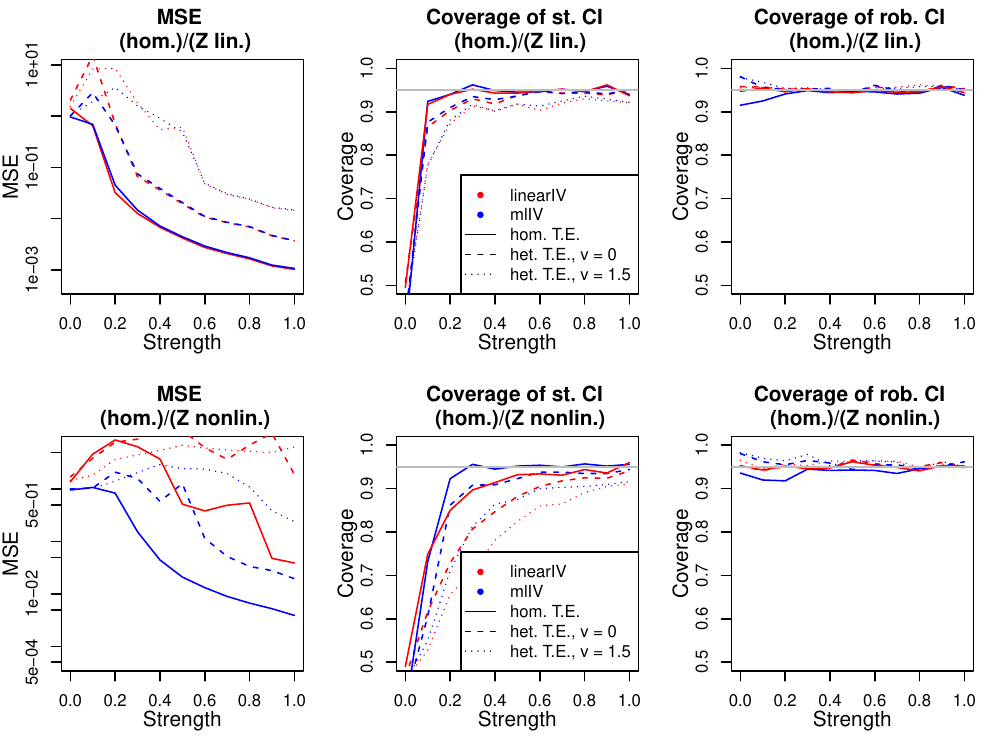}
\caption{Simulation results for setting (hom.)/(Z lin.) (top) and (hom.)/(Z nonlin.) (bottom) {with varying IV strength $s$}. Left panel: mean squared error (MSE) of $\hat \beta$ (hom. T.E.), $\hat\beta(0)$ (het. T.E. for $v = 0$) and $\hat\beta(1.5)$ (het. T.E. for $v = 1.5$). Middle panel: coverage of standard confidence intervals. Right panel: coverage of robust confidence sets. Methods based on (linearIV) are in red and methods based on (mlIV) are in blue. Solid lines correspond to the methods (hom. T.E. linearIV) and (hom. T.E. mlIV). Dashed lines correspond to the methods (het. T.E. linearIV) and (het. T.E. mlIV) for $v = 0$. Dotted lines correspond to the methods (het. T.E. linearIV) and (het. T.E. mlIV) for $v = 1.5$.} 
\label{fig_SimHom1DWeak}
\end{figure}
\begin{figure}[t]
\centering
\includegraphics[width=0.8\textwidth]{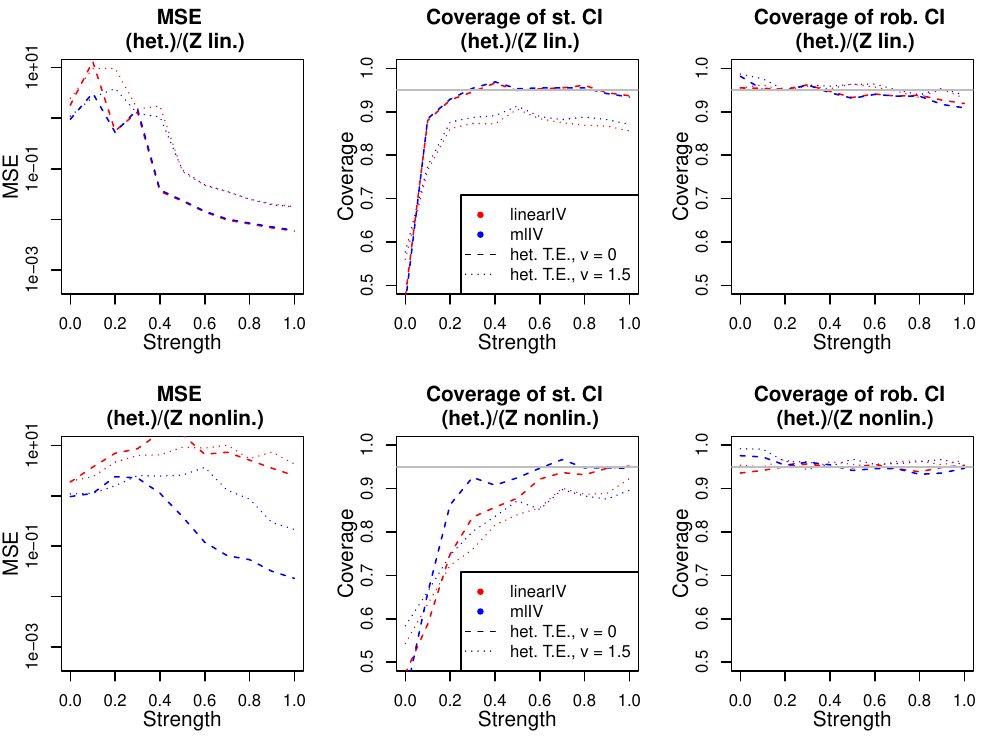}
\caption{Simulation results for setting (het.)/(Z lin.) (top) and (het.)/(Z nonlin.) (bottom) {with varying IV strength $s$}. Left panel: mean squared error (MSE) of  $\hat\beta(0)$ (het. T.E. for $v = 0$) and $\hat\beta(1.5)$ (het. T.E. for $v = 1.5$). Middle panel: coverage of standard confidence intervals. Right panel: coverage of robust confidence sets. Methods based on (linearIV) are in red and methods based on (mlIV) are in blue. Dashed lines correspond to the methods (het. T.E. linearIV) and (het. T.E. mlIV) for $v = 0$. Dotted lines correspond to the methods (het. T.E. linearIV) and (het. T.E. mlIV) for $v = 1.5$.} 
\label{fig_SimHet1DWeak}
\end{figure}
In both figures, we see that the mean squared error tends to decrease with increasing IV strength. Moreover, for small to moderately small IV strengths, the standard confidence intervals suffer from severe undercoverage and for the (Z nonlin.) settings, the undercoverage is worse for the methods based on (linearIV). In contrast, the robust confidence sets exhibit empirical coverage close to the nominal level of 0.95 regardless of the IV strength.
}

\subsection{Multidimensional $X$}\label{sec_MultiDimSim}
{We keep the framework of Section \ref{sec_Simulations} but use 5-dimensional covariates $X_i$ and use boosted regression trees using the \texttt{R} package \texttt{xgboost} \citep{ChenXgboostPaper, ChenXgboostPackage} for the nuisance function estimations. We slightly alter the simulation setup \eqref{eq_SimSetup} and generate data according to}
\begin{align}
    X_i&\sim_{i.i.d.}\mathcal N_5(0, \Sigma),\, \Sigma_{i,j} = 0.5 + 0.5\cdot1\{i = j\},\nonumber\\
    H_i, E_{Z,i}, E_{\delta, i}, E_{\epsilon, i} &\sim_{i.i.d.} \mathcal N(0, 1),\nonumber\\
    Z_i&\gets 0.5 X_{i,1}- 0.5 X_{i,2} + E_{Z,i},\nonumber\\
    \delta_i &\gets 0.7 H_i + 0.7 E_{\delta, i},\label{eq_SimSetup5D}\\
    \epsilon_i &\gets \mathrm{sign}(H_i) - 0.5 + 0.5 E_{\epsilon, i},\nonumber\\
    D_i&\gets f(Z_i, X_i) + \delta_i,\nonumber\\
    V_i&\gets X_{i,1},\nonumber\\
    Y_i&\gets \beta(V_i) D_i + g(X_i) + \epsilon_i.\nonumber
\end{align}
Note that the components of $X_i$ all have variance $1$ and are correlated. Moreover, we change the function $g$ to $g(x) = \tanh(x_1)-x_3$ and change the settings (Z lin.) and (Z nonlin.) to
\begin{description}
    \item[(Z lin.)] $f(z,x) = -\sin(x_1) + x_2 + z$,
    \item[(Z nonlin.)] $f(z,x) = -\sin(x_1) + x_2 + \cos(z) + 0.2 z$.
\end{description}
The remaining simulation setup stays the same as in Section \ref{sec_Simulations}.
Again, we first consider the settings (hom.)/(Z lin.) and (hom.)/(Z nonlin.) and use the methods (hom. T.E. mlIV) and (hom. T.E. linearIV) as well as the methods (het. T.E. mlIV) and (het. T.E. linearIV) for $\hat\beta(0)$ and $\hat\beta(1.5)$. The results are shown in Figure \ref{fig_SimHom5DStrong}.
\begin{figure}[t]
\centering
\includegraphics[width=0.8\textwidth]{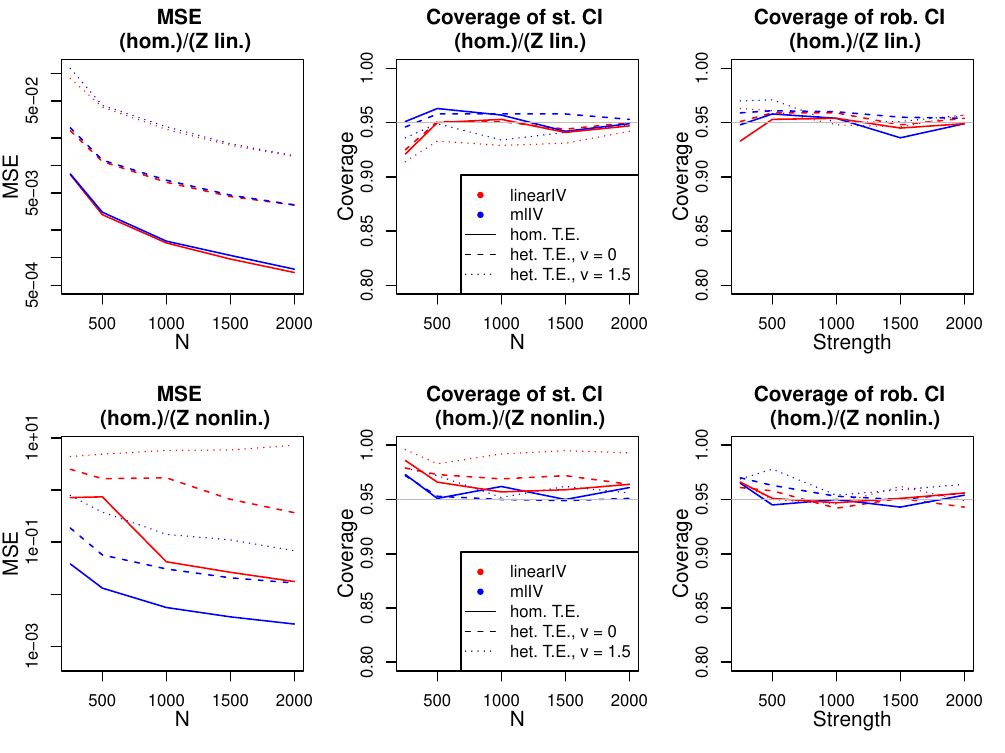}
\caption{Simulation results for setting (hom.)/(Z lin.) (top) and (hom.)/(Z nonlin.) (bottom) for 5-dimensional $X_i$ {and nuisance function estimation using \texttt{xgboost}}. Left panel: mean squared error (MSE) of $\hat \beta$ (hom. T.E.), $\hat\beta(0)$ (het. T.E. for $v = 0$)and $\hat\beta(1.5)$ (het. T.E. for $v = 1.5$). Middle panel: coverage of standard confidence intervals. Right panel: coverage of robust confidence sets. Methods based on (linearIV) are in red and methods based on (mlIV) are in blue. Solid lines correspond to the methods (hom. T.E. linearIV) and (hom. T.E. mlIV). Dashed lines correspond to the methods (het. T.E. linearIV) and (het. T.E. mlIV) for $v = 0$. Dotted lines correspond to the methods (het. T.E. linearIV) and (het. T.E. mlIV) for $v = 1.5$.} 
\label{fig_SimHom5DStrong}
\end{figure}
The results are very similar to Section \ref{sec_Simulations}: in the (Z lin.) settings, the estimators based on (linearIV) exhibit almost identical results to the estimators based on (mlIV), whereas in the (Z nonlin.) settings, the estimators based on (mlIV) have significantly smaller MSEs, while still exhibiting good empirical coverage of both standard and robust confidence intervals.

Secondly, we consider the settings (het.)/(Z lin.) and (het.)/(Z nonlin.) and use the methods (het. T.E. mlIV) and (het. T.E. linearIV) for $\hat\beta(0)$ and $\hat\beta(1.5)$. The results are shown in Figure \ref{fig_SimHet5DStrong}. Again, the results are comparable to those from Section \ref{sec_Simulations}.

\begin{figure}[t]
\centering
\includegraphics[width=0.8\textwidth]{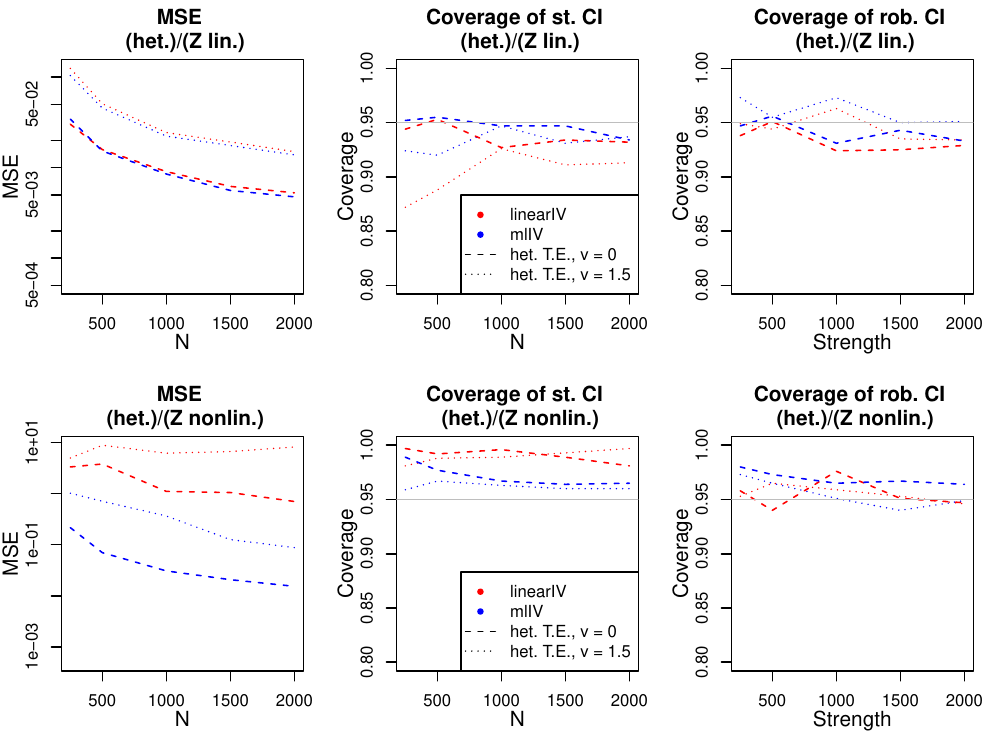}
\caption{Simulation results for setting (het.)/(Z lin.) (top) and (het.)/(Z nonlin.) (bottom) for 5-dimensional $X_i$ {and nuisance function estimation using \texttt{xgboost}}. Left panel: mean squared error (MSE) of  $\hat\beta(0)$ (het. T.E. for $v = 0$) and $\hat\beta(1.5)$ (het. T.E. for $v = 1.5$). Middle panel: coverage of standard confidence intervals. Right panel: coverage of robust confidence sets. Methods based on (linearIV) are in red and methods based on (mlIV) are in blue. Dashed lines correspond to the methods (het. T.E. linearIV) and (het. T.E. mlIV) for $v = 0$. Dotted lines correspond to the methods (het. T.E. linearIV) and (het. T.E. mlIV) for $v = 1.5$.} 
\label{fig_SimHet5DStrong}
\end{figure}

{
Finally, let us consider the influence of the IV strength in the setting with 5-dimensional $X_i$. As in Section \ref{sec_VaryIVStrength}, we include a parameter $s$ corresponding to different IV strengths and modify $\delta_i$ and $\epsilon_i$ to increase their correlation. Otherwise, we keep
the setup from \eqref{eq_SimSetup5D}. The resulting plots for the settings (hom.)/(Z lin.) and (hom.)/(Z nonlin.) can be found in Figure \ref{fig_SimHom5DWeak} and the plots for the settings (het.)/(Z lin.) and (het.)/(Z nonlin.) can be found in Figure \ref{fig_SimHet5DWeak}.
\begin{figure}[t]
\centering
\includegraphics[width=0.8\textwidth]{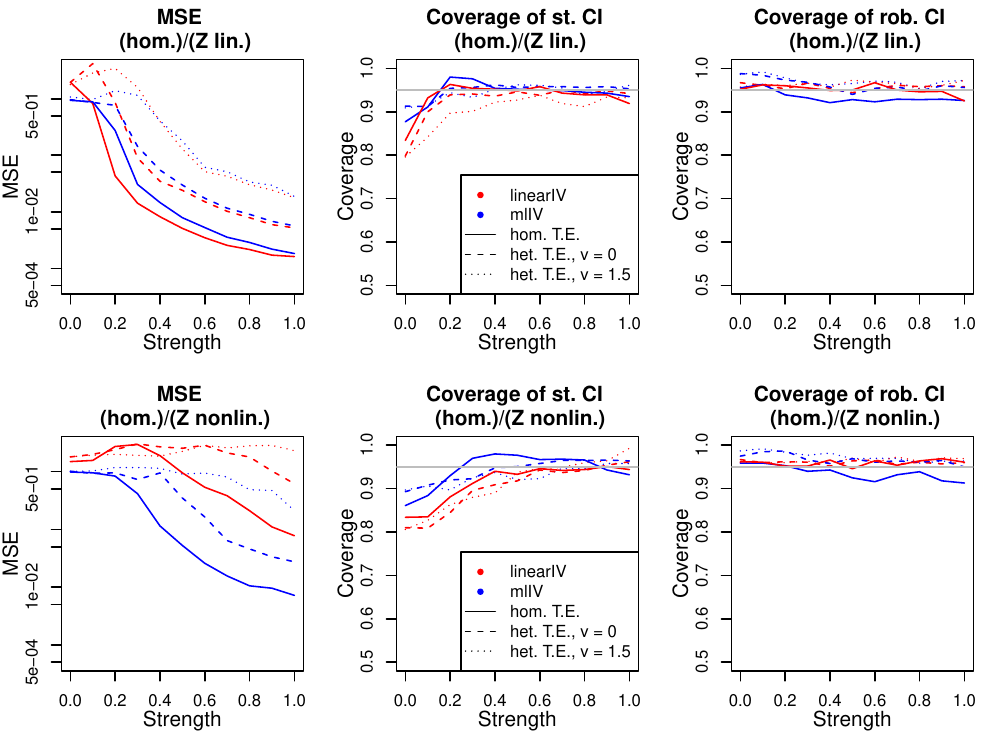}
\caption{{Simulation results for setting (hom.)/(Z lin.) (top) and (hom.)/(Z nonlin.) (bottom) with varying IV strength $s$, 5-dimensional $X_i$ and nuisance function estimation using \texttt{xgboost}. Left panel: mean squared error (MSE) of $\hat \beta$ (hom. T.E.), $\hat\beta(0)$ (het. T.E. for $v = 0$) and $\hat\beta(1.5)$ (het. T.E. for $v = 1.5$). Middle panel: coverage of standard confidence intervals. Right panel: coverage of robust confidence sets. Methods based on (linearIV) are in red and methods based on (mlIV) are in blue. Solid lines correspond to the methods (hom. T.E. linearIV) and (hom. T.E. mlIV). Dashed lines correspond to the methods (het. T.E. linearIV) and (het. T.E. mlIV) for $v = 0$. Dotted lines correspond to the methods (het. T.E. linearIV) and (het. T.E. mlIV) for $v = 1.5$.}} 
\label{fig_SimHom5DWeak}
\end{figure}
\begin{figure}[t]
\centering
\includegraphics[width=0.8\textwidth]{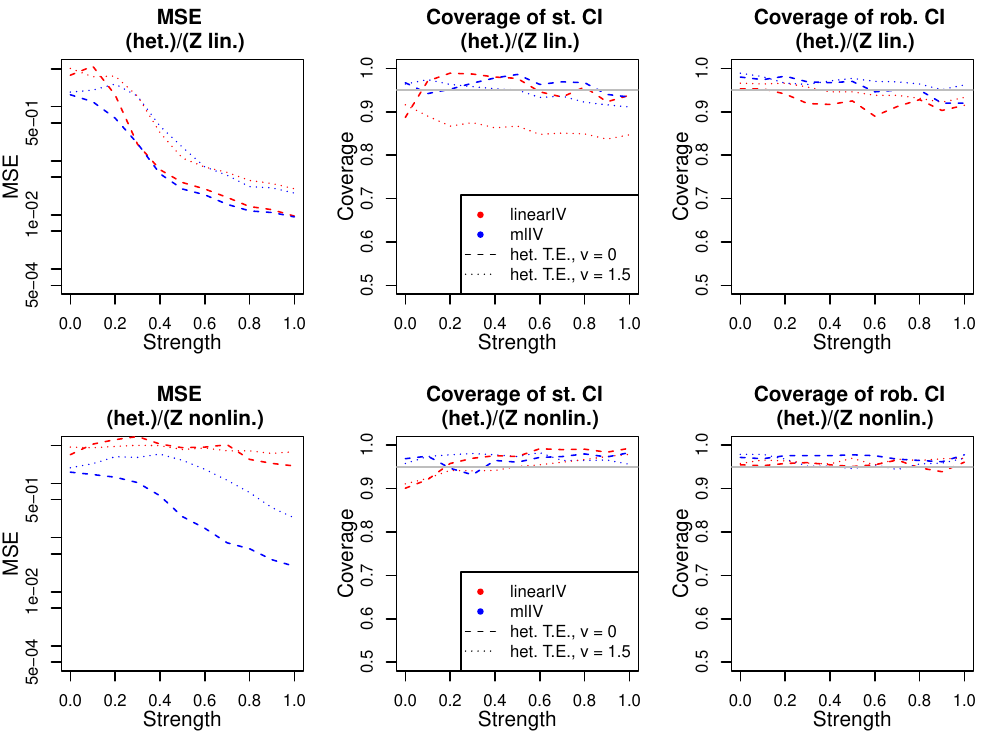}
\caption{{Simulation results for setting (het.)/(Z lin.) (top) and (het.)/(Z nonlin.) (bottom) with varying IV strength $s$, 5-dimensional $X_i$ and nuisance function estimation using \texttt{xgboost}. Left panel: mean squared error (MSE) of  $\hat\beta(0)$ (het. T.E. for $v = 0$) and $\hat\beta(1.5)$ (het. T.E. for $v = 1.5$). Middle panel: coverage of standard confidence intervals. Right panel: coverage of robust confidence sets. Methods based on (linearIV) are in red and methods based on (mlIV) are in blue. Dashed lines correspond to the methods (het. T.E. linearIV) and (het. T.E. mlIV) for $v = 0$. Dotted lines correspond to the methods (het. T.E. linearIV) and (het. T.E. mlIV) for $v = 1.5$.}} 
\label{fig_SimHet5DWeak}
\end{figure}
Whereas for the (hom.) settings, there is still a considerable (though less pronounced) undercoverage of the standard confidence intervals for small IV strength, this effect seems to be almost gone for the (het.) settings.}

\end{appendix}

\bibliography{Literature.bib}

\begin{thebibliography}{58}
\providecommand{\natexlab}[1]{#1}
\providecommand{\url}[1]{\texttt{#1}}
\expandafter\ifx\csname urlstyle\endcsname\relax
  \providecommand{\doi}[1]{doi: #1}\else
  \providecommand{\doi}{doi: \begingroup \urlstyle{rm}\Url}\fi

\bibitem[Acemoglu et~al.(2001)Acemoglu, Johnson, and
  Robinson]{AcemogluColonialOrigins}
D.~Acemoglu, S.~Johnson, and J.~A. Robinson.
\newblock The colonial origins of comparative development: An empirical
  investigation.
\newblock \emph{American Economic Review}, 91\penalty0 (5):\penalty0
  1369–1401, 2001.

\bibitem[Ahrens et~al.(2024)Ahrens, Hansen, Schaffer, and
  Wiemann]{AhrensDMLStata}
A.~Ahrens, C.~B. Hansen, M.~E. Schaffer, and T.~Wiemann.
\newblock ddml: Double/debiased machine learning in stata.
\newblock \emph{The Stata Journal}, 24\penalty0 (1):\penalty0 3–45, 2024.

\bibitem[Anderson and Rubin(1949)]{AndersonRubinTest}
T.~W. Anderson and H.~Rubin.
\newblock Estimation of the parameters of a single equation in a complete
  system of stochastic equations.
\newblock \emph{The Annals of Mathematical Statistics}, 20\penalty0
  (1):\penalty0 46--63, 1949.

\bibitem[Athey et~al.(2019)Athey, Tibshirani, and Wager]{AtheyGRF}
S.~Athey, J.~Tibshirani, and S.~Wager.
\newblock {Generalized random forests}.
\newblock \emph{The Annals of Statistics}, 47\penalty0 (2):\penalty0 1148 –
  1178, 2019.

\bibitem[Bach et~al.(2024)Bach, Schacht, Chernozhukov, Klaassen, and
  Spindler]{BachHyperparameterTuning}
P.~Bach, O.~Schacht, V.~Chernozhukov, S.~Klaassen, and M.~Spindler.
\newblock Hyperparameter tuning for causal inference with double machine
  learning: A simulation study.
\newblock In \emph{Proceedings of the Third Conference on Causal Learning and
  Reasoning}, volume 236, pages 1065--1117. PMLR, 2024.
\newblock URL \url{https://proceedings.mlr.press/v236/bach24a.html}.

\bibitem[Belloni et~al.(2012)Belloni, Chen, Chernozhukov, and
  Hansen]{BelloniSparseModels}
A.~Belloni, D.~Chen, V.~Chernozhukov, and C.~Hansen.
\newblock Sparse models and methods for optimal instruments with an application
  to eminent domain.
\newblock \emph{Econometrica}, 80\penalty0 (6):\penalty0 2369–2429, 2012.

\bibitem[Bühlmann and van~de Geer(2011)]{BuehlmannHDStats}
P.~Bühlmann and S.~van~de Geer.
\newblock \emph{Statistics for High-Dimensional Data: Methods, Theory and
  Applications}.
\newblock Springer Series in Statistics. Springer Berlin Heidelberg, 2011.

\bibitem[Cai and Xiong(2012)]{CaiPartiallyVaryingCoefficientIVModels}
Z.~Cai and H.~Xiong.
\newblock Partially varying coefficient instrumental variables models.
\newblock \emph{Statistica Neerlandica}, 66\penalty0 (2):\penalty0 85--110,
  2012.

\bibitem[Cai et~al.(2006)Cai, Das, Xiong, and
  Wu]{CaiFunctionalCoefficientIVModels}
Z.~Cai, M.~Das, H.~Xiong, and X.~Wu.
\newblock Functional coefficient instrumental variables models.
\newblock \emph{Journal of Econometrics}, 133\penalty0 (1):\penalty0 207–241,
  2006.

\bibitem[Card(1995)]{CardCollege}
D.~Card.
\newblock Using geographic variation in college proximity to estimate the
  return to schooling.
\newblock In \emph{Aspects of Labor Market Behavior: Essays in Honor of John
  Vanderkamp}, pages 201--222. National Longitudinal Survey of Young Men, 1995.

\bibitem[{Chen} et~al.(2020){Chen}, {Chen}, and {Lewis}]{ChenMostlyHarmlessML}
J.~{Chen}, D.~L. {Chen}, and G.~{Lewis}.
\newblock {Mostly Harmless Machine Learning: Learning Optimal Instruments in
  Linear IV Models}.
\newblock \emph{arXiv e-prints}, art. arXiv:2011.06158, 2020.

\bibitem[Chen and Guestrin(2016)]{ChenXgboostPaper}
T.~Chen and C.~Guestrin.
\newblock Xgboost: A scalable tree boosting system.
\newblock \emph{Proceedings of the 22nd ACM SIGKDD International Conference on
  Knowledge Discovery and Data Mining}, 2016.

\bibitem[Chen et~al.(2023)Chen, He, Benesty, Khotilovich, Tang, Cho, Chen,
  Mitchell, Cano, Zhou, Li, Xie, Lin, Geng, Li, and Yuan]{ChenXgboostPackage}
T.~Chen, T.~He, M.~Benesty, V.~Khotilovich, Y.~Tang, H.~Cho, K.~Chen,
  R.~Mitchell, I.~Cano, T.~Zhou, M.~Li, J.~Xie, M.~Lin, Y.~Geng, Y.~Li, and
  J.~Yuan.
\newblock \emph{xgboost: Extreme Gradient Boosting}, 2023.
\newblock URL \url{https://CRAN.R-project.org/package=xgboost}.
\newblock R package version 1.7.5.1.

\bibitem[Chen and White(1999)]{ChenImprovedRates}
X.~Chen and H.~White.
\newblock Improved rates and asymptotic normality for nonparametric neural
  network estimators.
\newblock \emph{IEEE Transactions on Information Theory}, 45\penalty0
  (2):\penalty0 682--691, 1999.

\bibitem[Chernozhukov et~al.(2016)Chernozhukov, Hansen, and
  Spindler]{ChernozhukovHDM}
V.~Chernozhukov, C.~Hansen, and M.~Spindler.
\newblock {hdm}: High-dimensional metrics.
\newblock \emph{R Journal}, 8\penalty0 (2):\penalty0 185--199, 2016.
\newblock URL
  \url{https://journal.r-project.org/archive/2016/RJ-2016-040/index.html}.

\bibitem[Chernozhukov et~al.(2018)Chernozhukov, Chetverikov, Demirer, Duflo,
  Hansen, Newey, and Robins]{ChernozhukovDML}
V.~Chernozhukov, D.~Chetverikov, M.~Demirer, E.~Duflo, C.~Hansen, W.~Newey, and
  J.~Robins.
\newblock Double/debiased machine learning for treatment and structural
  parameters.
\newblock \emph{The Econometrics Journal}, 21\penalty0 (1):\penalty0 C1–C68,
  2018.

\bibitem[Chernozhukov et~al.(2024)Chernozhukov, Hansen, Kallus, Spindler, and
  Syrgkanis]{ChernozhukovCausalML}
V.~Chernozhukov, C.~Hansen, N.~Kallus, M.~Spindler, and V.~Syrgkanis.
\newblock \emph{{Applied Causal Inference Powered by ML and AI}}.
\newblock 2024.
\newblock URL \url{CausalML-book.org}.

\bibitem[Cleveland et~al.(1991)Cleveland, Grosse, and
  Shyu]{ClevelandLocalRegressionMOdels}
W.~S. Cleveland, E.~Grosse, and W.~M. Shyu.
\newblock Local regression models.
\newblock In J.~M. Chambers and T.~J. Hastie, editors, \emph{Statistical Models
  in S}, pages 309--376. Wadsworth \& Brooks, Pacific Grove, 1991.

\bibitem[Dieterle and Snell(2016)]{DieterleASimpleDiagnostic}
S.~G. Dieterle and A.~Snell.
\newblock A simple diagnostic to investigate instrument validity and
  heterogeneous effects when using a single instrument.
\newblock \emph{Labour Economics}, 42:\penalty0 76–86, 2016.

\bibitem[Durrett(2019)]{DurrettProb}
R.~Durrett.
\newblock \emph{Probability: Theory and Examples}.
\newblock Cambridge Series in Statistical and Probabilistic Mathematics.
  Cambridge University Press, 5 edition, 2019.

\bibitem[Emmenegger and B\"uhlmann(2021)]{EmmeneggerRegularizingDML}
C.~Emmenegger and P.~B\"uhlmann.
\newblock {Regularizing double machine learning in partially linear endogenous
  models}.
\newblock \emph{Electronic Journal of Statistics}, 15\penalty0 (2):\penalty0
  6461 – 6543, 2021.

\bibitem[Fan and Zhang(1999)]{FanStatisticalEstimation}
J.~Fan and W.~Zhang.
\newblock {Statistical estimation in varying coefficient models}.
\newblock \emph{The Annals of Statistics}, 27\penalty0 (5):\penalty0 1491 –
  1518, 1999.

\bibitem[Fan and Zhang(2008)]{FanStatisticalMethods}
J.~Fan and W.~Zhang.
\newblock Statistical methods with varying coefficient models.
\newblock \emph{Statistics and Its Interface}, 1\penalty0 (1):\penalty0
  179--195, 2008.

\bibitem[Fan et~al.(2022)Fan, Hsu, Lieli, and Zhang]{FanEstimationOfCATE}
Q.~Fan, Y.-C. Hsu, R.~P. Lieli, and Y.~Zhang.
\newblock Estimation of conditional average treatment effects with
  high-dimensional data.
\newblock \emph{Journal of Business \& Economic Statistics}, 40\penalty0
  (1):\penalty0 313–327, 2022.

\bibitem[Foster and Syrgkanis(2023)]{FosterOrthogonalStatisticalLearning}
D.~J. Foster and V.~Syrgkanis.
\newblock {Orthogonal statistical learning}.
\newblock \emph{The Annals of Statistics}, 51\penalty0 (3):\penalty0 879 –
  908, 2023.

\bibitem[{Guo} et~al.(2022){Guo}, {Zheng}, and {B{\"u}hlmann}]{GuoTSCI}
Z.~{Guo}, M.~{Zheng}, and P.~{B{\"u}hlmann}.
\newblock Robustness against weak or invalid instruments: Exploring nonlinear
  treatment models with machine learning.
\newblock \emph{arXiv e-prints}, art. arXiv:2203.12808, 2022.

\bibitem[Hastie and Tibshirani(1993)]{HastieVaryingCoefficientModels}
T.~Hastie and R.~Tibshirani.
\newblock Varying-coefficient models.
\newblock \emph{Journal of the Royal Statistical Society: Series B
  (Methodological)}, 55\penalty0 (4):\penalty0 757--779, 1993.

\bibitem[Kang et~al.(2023)Kang, Jiang, Zhao, and Small]{KangIVmodel}
H.~Kang, Y.~Jiang, Q.~Zhao, and D.~Small.
\newblock \emph{ivmodel: Statistical Inference and Sensitivity Analysis for
  Instrumental Variables Model}, 2023.
\newblock URL \url{https://CRAN.R-project.org/package=ivmodel}.
\newblock R package version 1.9.1.

\bibitem[Kennedy(2023)]{KennedyTowardsOptimalDoublyRobustEstimation}
E.~H. Kennedy.
\newblock {Towards optimal doubly robust estimation of heterogeneous causal
  effects}.
\newblock \emph{Electronic Journal of Statistics}, 17\penalty0 (2):\penalty0
  3008 – 3049, 2023.

\bibitem[Kennedy et~al.(2024)Kennedy, Balakrishnan, Robins, and
  Wasserman]{KennedyMinimaxRates}
E.~H. Kennedy, S.~Balakrishnan, J.~M. Robins, and L.~Wasserman.
\newblock {Minimax rates for heterogeneous causal effect estimation}.
\newblock \emph{The Annals of Statistics}, 52\penalty0 (2):\penalty0 793 –
  816, 2024.

\bibitem[{Liu} et~al.(2020){Liu}, {Shang}, and {Cheng}]{LiuDeepIV}
R.~{Liu}, Z.~{Shang}, and G.~{Cheng}.
\newblock On deep instrumental variables estimate.
\newblock \emph{arXiv e-prints}, art. arXiv:2004.14954, 2020.

\bibitem[{Newey} and {Robins}(2018)]{NeweyCrossFittingAndFastRemainderRates}
W.~K. {Newey} and J.~R. {Robins}.
\newblock {Cross-Fitting and Fast Remainder Rates for Semiparametric
  Estimation}.
\newblock \emph{arXiv e-prints}, art. arXiv:1801.09138, Jan. 2018.

\bibitem[Nie and Wager(2020)]{NieQuasiOracle}
X.~Nie and S.~Wager.
\newblock Quasi-oracle estimation of heterogeneous treatment effects.
\newblock \emph{Biometrika}, 108\penalty0 (2):\penalty0 299--319, 2020.

\bibitem[Okui et~al.(2012)Okui, Small, Tan, and
  Robins]{OkuiDoublyRobustIVRegression}
R.~Okui, D.~S. Small, Z.~Tan, and J.~M. Robins.
\newblock Doubly robust instrumental variable regresion.
\newblock \emph{Statistica Sinica}, 22, 2012.

\bibitem[Oprescu et~al.(2019)Oprescu, Syrgkanis, and Wu]{OprescuORF}
M.~Oprescu, V.~Syrgkanis, and Z.~S. Wu.
\newblock Orthogonal random forest for causal inference.
\newblock In \emph{Proceedings of the 36th International Conference on Machine
  Learning}, volume~97, pages 4932--4941. PMLR, 2019.

\bibitem[Robins et~al.(2008)Robins, Li, Tchetgen, van~der Vaart,
  et~al.]{RobinsHigherOrderInfluenceFunctions}
J.~Robins, L.~Li, E.~Tchetgen, A.~van~der Vaart, et~al.
\newblock Higher order influence functions and minimax estimation of nonlinear
  functionals.
\newblock In \emph{Probability and statistics: essays in honor of David A.
  Freedman}, volume~2, pages 335--422. Institute of Mathematical Statistics,
  2008.

\bibitem[Rubin(2005)]{RubinPotentialOutcomes}
D.~B. Rubin.
\newblock Causal inference using potential outcomes: Design, modeling,
  decisions.
\newblock \emph{Journal of the American Statistical Association}, 100\penalty0
  (469):\penalty0 322--331, 2005.

\bibitem[Scheidegger(2025)]{IVDMLCRAN}
C.~Scheidegger.
\newblock \emph{{IVDML}: Double Machine Learning with Instrumental Variables
  and Heterogeneity}, 2025.
\newblock R package version 1.0.0.

\bibitem[Semenova and Chernozhukov(2020)]{SemenovaDMLCATE}
V.~Semenova and V.~Chernozhukov.
\newblock Debiased machine learning of conditional average treatment effects
  and other causal functions.
\newblock \emph{The Econometrics Journal}, 24\penalty0 (2):\penalty0 264--289,
  2020.

\bibitem[Silverman(1986)]{SilvermanDensityEstimation}
B.~W. Silverman.
\newblock \emph{{Density estimation for statistics and data analysis}}.
\newblock Chapman \& Hall/CRC monographs on statistics and applied probability.
  Chapman and Hall, London, 1986.

\bibitem[Staiger and Stock(1997)]{StaigerIVRegressionWeakInstruments}
D.~Staiger and J.~H. Stock.
\newblock Instrumental variables regression with weak instruments.
\newblock \emph{Econometrica}, 65\penalty0 (3):\penalty0 557--586, 1997.

\bibitem[Stock and Watson(2020)]{StockIntroductionToEconometrics}
J.~H. Stock and M.~W. Watson.
\newblock \emph{Introduction to Econometrics}.
\newblock Pearson, Harlow, 4 edition, 2020.

\bibitem[Su et~al.(2013)Su, Murtazashvili, and
  Ullah]{SuLocalLinearGMMEstimation}
L.~Su, I.~Murtazashvili, and A.~Ullah.
\newblock Local linear gmm estimation of functional coefficient iv models with
  an application to estimating the rate of return to schooling.
\newblock \emph{Journal of Business \& Economic Statistics}, 31\penalty0
  (2):\penalty0 184–207, 2013.

\bibitem[Syrgkanis et~al.(2019)Syrgkanis, Lei, Oprescu, Hei, Battocchi, and
  Lewis]{SyrgkanisMLEstimationOfHTE}
V.~Syrgkanis, V.~Lei, M.~Oprescu, M.~Hei, K.~Battocchi, and G.~Lewis.
\newblock Machine learning estimation of heterogeneous treatment effects with
  instruments.
\newblock In \emph{Advances in Neural Information Processing Systems},
  volume~32, 2019.

\bibitem[Tsiatis(2007)]{TsiatisSemiparametric}
A.~Tsiatis.
\newblock \emph{Semiparametric Theory and Missing Data}.
\newblock Springer Series in Statistics. Springer, New York, 2007.

\bibitem[van~der Vaart(1998)]{VanDerVaartAsymptotic}
A.~W. van~der Vaart.
\newblock \emph{Asymptotic Statistics}.
\newblock Cambridge Series in Statistical and Probabilistic Mathematics.
  Cambridge University Press, 1998.

\bibitem[Vansteelandt and Didelez(2018)]{VansteelandtImprovingRobustness}
S.~Vansteelandt and V.~Didelez.
\newblock Improving the robustness and efficiency of covariate-adjusted linear
  instrumental variable estimators.
\newblock \emph{Scandinavian Journal of Statistics}, 45\penalty0 (4):\penalty0
  941–961, 2018.

\bibitem[Wager(2024)]{WagerCausalInference}
S.~Wager.
\newblock {Causal Inference: A Statistical Learning Approach}.
\newblock Available at
  \url{https://web.stanford.edu/~swager/causal_inf_book.pdf}, 2024.

\bibitem[Wager and Athey(2018)]{WagerEstimationInferenceHTE}
S.~Wager and S.~Athey.
\newblock Estimation and inference of heterogeneous treatment effects using
  random forests.
\newblock \emph{Journal of the American Statistical Association}, 113\penalty0
  (523):\penalty0 1228–1242, 2018.

\bibitem[{Wager} and {Walther}(2015)]{WagerAdaptiveConcentration}
S.~{Wager} and G.~{Walther}.
\newblock {Adaptive Concentration of Regression Trees, with Application to
  Random Forests}.
\newblock \emph{arXiv e-prints}, art. arXiv:1503.06388, Mar. 2015.

\bibitem[Wasserman(2006)]{WassermannAllOfNonparametricStatistics}
L.~Wasserman.
\newblock \emph{All of Nonparametric Statistics}.
\newblock Springer Texts in Statistics. Springer New York, 2006.

\bibitem[Wood(2006)]{WoodGLM1}
S.~N. Wood.
\newblock \emph{Generalized Additive Models: An Introduction with R}.
\newblock Chapman and Hall/CRC, New York, 2006.

\bibitem[Wood(2011)]{WoodGLM2}
S.~N. Wood.
\newblock Fast stable restricted maximum likelihood and marginal likelihood
  estimation of semiparametric generalized linear models.
\newblock \emph{Journal of the Royal Statistical Society Series B: Statistical
  Methodology}, 73\penalty0 (1):\penalty0 3--36, 2011.

\bibitem[Wooldridge(2013)]{WooldridgeIntroductoryEconometrics}
J.~M. Wooldridge.
\newblock \emph{Introductory Econometrics: A Modern Approach}.
\newblock South-Western Cengage Learning, Mason, 5 edition, 2013.

\bibitem[Wright and Ziegler(2017)]{RangerPackage}
M.~N. Wright and A.~Ziegler.
\newblock {ranger}: A fast implementation of random forests for high
  dimensional data in {C++} and {R}.
\newblock \emph{Journal of Statistical Software}, 77\penalty0 (1):\penalty0
  1--17, 2017.

\bibitem[{Young} and {Shah}(2024)]{YoungROSE}
E.~H. {Young} and R.~D. {Shah}.
\newblock Rose random forests for robust semiparametric efficient estimation.
\newblock \emph{arXiv e-prints}, art. arXiv:2410.03471, 2024.

\bibitem[Yuan and Zhou(2016)]{YuanMinimaxOptimalRates}
M.~Yuan and D.-X. Zhou.
\newblock {Minimax optimal rates of estimation in high dimensional additive
  models}.
\newblock \emph{The Annals of Statistics}, 44\penalty0 (6):\penalty0 2564 –
  2593, 2016.

\bibitem[{Zimmert} and {Lechner}(2019)]{ZimmertNonparametricEstimation}
M.~{Zimmert} and M.~{Lechner}.
\newblock Nonparametric estimation of causal heterogeneity under
  high-dimensional confounding.
\newblock \emph{arXiv e-prints}, art. arXiv:1908.08779, 2019.

\end{thebibliography}

\end{document}